\documentclass[11pt]{article}
\usepackage[letterpaper, voffset = 0in, left = 1in, right = 1in, inner = 0in, outer = 0in, top = 0in, bottom = 0in, headsep=0in, voffset = .05in]{geometry}

\usepackage{natbib}
\usepackage{graphicx}
\usepackage[caption=false]{subfig}
\usepackage{amssymb}
\usepackage{amsmath}
\usepackage{amssymb}
\usepackage{amsthm}
\usepackage{tikz}
\usepackage{scalerel}
\usepackage{dsfont}
\usepackage{dcolumn}
\usepackage{mathrsfs}
\usepackage{breqn}

\usepackage[linesnumbered,ruled,vlined]{algorithm2e}

\usepackage{makecell}


\usepackage{footnote}
\makesavenoteenv{algorithm, setspace}
\usepackage{mathrsfs}
\usepackage{mathtools}
\usepackage{wrapfig}
\usepackage{framed}
\usepackage{extarrows}
\usepackage{caption}
\usepackage[T1]{fontenc}
\usepackage{kantlipsum}
\PassOptionsToPackage{usenames,dvipsnames}{xcolor}
\usepackage[breakable, theorems, skins]{tcolorbox}
\tcbset{enhanced}

\usepackage{epstopdf}
\usepackage{amsmath}
\usepackage{amsfonts}
\usepackage{amssymb}
\usepackage{graphicx}
\usepackage{comment}
\usepackage{tikz}
\usetikzlibrary{arrows}
\usetikzlibrary{shapes}

\newcommand{\indep}{\perp \!\!\! \perp}
\newcommand{\R}{{\mathbb{R}}}
\newcommand{\E}{{\mathbb{E}}}
\newcommand{\pr}{{\mathbb{P}}}

\newcommand{\T}{{\mathbb{T}}}

\newcommand{\pps}{\texttt{synthesis} }
\newcommand{\bff}{\texttt{bff} }

\newtheorem{lemma}{Lemma}[section]
\newtheorem{theorem}{Theorem}[section]
\newtheorem{corollary}{Corollary}[theorem]

\newtheorem{definition}{Definition}
\newtheorem{assumption}{Assumption}
\newtheorem{prop}{Proposition}[section]
\newtheorem{remark}{Remark}[section]

\newtheorem{exmp}{Example}[section]

\allowdisplaybreaks

\title{$\T$-Stochastic Graphs}
\author{Sijia Fang and Karl Rohe}
\date{\vspace{-2em}}

\begin{document}

\begin{LARGE}
\textbf{$\T$-Stochastic Graphs}
\end{LARGE}
\vspace{.3in}

Sijia Fang and Karl Rohe

\textit{UW-Madison Statistics Department,
Madison, WI,
USA.}

E-mail: sfang44@wisc.edu karlrohe@stat.wisc.edu

\begin{quote}
\noindent
\textbf{Summary:} 

Previous statistical approaches to hierarchical clustering for social network analysis all construct an ``ultrametric'' hierarchy.  
While the assumption of ultrametricity has been discussed and studied in the phylogenetics literature, it has not yet been acknowledged in the social network literature.
We show that ``non-ultrametric structure'' in the network introduces significant instabilities in the existing top-down recovery algorithms.
To address this issue, we introduce an instability diagnostic plot and use it to examine a collection of empirical networks. These networks appear to violate the ``ultrametric'' assumption. 
We propose a deceptively simple and yet general class of probabilistic models called $\T$-Stochastic Graphs which impose no topological restrictions on the latent hierarchy. To illustrate this model, we propose six alternative forms of hierarchical network models and then show that all six are 
equivalent to the $\T$-Stochastic Graph model. These alternative models motivate a novel approach to hierarchical clustering that combines spectral techniques with the well-known Neighbor-Joining algorithm from phylogenetic reconstruction.   We prove this spectral approach is statistically consistent.

\end{quote}

\section{Introduction}

Empirical social networks can have hundreds of communities \citep{fanpvalues, zhang2018attention, varimax_rohe}.  
Interpreting and naming them is challenging and time-consuming. To aid these efforts, it can be helpful to understand which communities are closer than others.  For example, if the communities are hierarchically structured \citep{ravasz2002hierarchical, lancichinetti2009detecting, SHEN20091706}, it can dramatically lessen the burden of interpreting the potentially hundreds of communities.

Our proposed approach builds on the previous approaches to modeling hierarchies in social networks \citep{li, lei, clauset2008hierarchical}.
Many of these papers propose ``top-down'' algorithms and prove that they are statistically consistent under specific models.  Unfortunately, in empirical networks, these algorithms are often unstable. To understand why, Section  \ref{sec:instability} proposes a simple and quick-to-compute \textit{top-down instability diagnostic} and  illustrates the diagnostic on a collection of large empirical social networks.

To overcome the instabilities in previous approaches, this paper proposes the $\T$-Stochastic Graph model, a simple model class for generating random graphs with the latent hierarchical structure encoded in a tree graph $\T$ (Section \ref{sec:tsg}). Importantly, $\T$ can be any weighted tree graph with no restrictions imposed on the topology or the edge weights. The leaf nodes in $\T$ correspond to people in the social network and the non-leaf nodes can be interpreted as ``communities'' or ``blocks''. The probability that two people are friends is parameterized by the distance between their leaf nodes in $\T$. The previously proposed models in \citep{li, lei, clauset2008hierarchical} are all contained in this model class (shown in Appendix \ref{appendix:other_models}); as such, each of them can be expressed as $\T$-Stochastic Graph models with certain restrictions on the tree $\T$. Among these restrictions, the most decisive one is a type of homogeneous assumption that assumes all leaf nodes are equal distant to the root node; this is the ``ultrametric'' assumption.

To gain more intuition for the $\T$-Stochastic Graph model, we propose six alternative ways of imagining and modeling hierarchical structures in social networks (Section \ref{sec:intuitions}).  These models include 
\begin{enumerate}
\itemsep0em 
    \item a type of Stochastic Blockmodel, 
    \item a Random Dot Product Graph where the latent positions are not independent, but sampled from a Graphical Model with a tree structure, 
    \item a ``hierarchically overlapping'' Stochastic Blockmodel, 
    \item a Stochastic Process that generates random edges one-at-a-time in a ``top-down'' fashion,
    \item a Stochastic Process that generates random edges one-at-a-time in a ``bottom-up'' fashion, and
    \item an axiomatic notion ``Hierarchical Stochastic Equivalence.'' 
\end{enumerate} 
While these six alternative perspectives might appear distinct from each other, 
we prove that all six models are equivalent to the $\T$-Stochastic Graph model (and thus equivalent to one another). As such, these models provide alternative ways of imagining the $\T$-Stochastic Graph model. 

To estimate the latent hierarchy $\T$ from a $\T$-Stochastic Graph, we propose an algorithm with three steps.  The first step fits a Degree-Corrected Stochastic Blockmodel (DCSBM) \citep{karrer2011stochastic} that partially recovers the hierarchical structure in $\T$. 
In the second step, each element of the  ``$k \times k$ block probability matrix'' from the DCSBM is log-transformed; under the $\T$-Stochastic Graph model, this estimates a distance matrix.  In the last step, we give this estimated distance matrix to the popular Neighbor-Joining (NJ) algorithm \citep{saitou1987neighbor} from the literature on phylogenetic tree reconstruction. Because this three step algorithm synthesizes techniques from the spectral and phylogenetic literature, and the reconstructed tree graph synthesizes the leaves, we call it \texttt{synthesis}.
The theoretical results in Section \ref{sec:estimation} study when the \pps algorithm consistently estimates $\T$.

Figure \ref{fig:journal_tree} illustrates the results of \pps  applied to a citation network of 22,688 academic journals with $k=100$ dimensions. Appendix \ref{sec:real_data_journal} explains this data example in greater detail. Previously, \cite{fanpvalues} studied this graph; using a resampling technique, \cite{fanpvalues} found statistical evidence for at least 100 dimensions in this network. Then, \cite{varimax_rohe} extracted keywords from the journal titles for each of the $k=100$ dimensions and ``manually'' interpret these 100 dimensions. Figure \ref{fig:journal_tree} displays the hierarchy that \pps estimates on these 100 dimensions.  As described in the figure caption, this hierarchy helps to interpret the dimensions and the relationships between them. Section \ref{sec:real_data} provides another empirical example with \pps on the Wikipedia hyperlink network, where the recovered hierarchy resembles a world map.

Taken together, our study contributes to the understanding of modeling social network hierarchies in five ways. In particular, this paper

\begin{enumerate}
\itemsep0em 
    \item highlights the ``ultrametric'' assumption
     and provides a diagnostic for it,
    \item proposes a class of models that do not have restrictions on the latent tree topology,
    \item shows how this class of models is equivalent to six other parameterizations of hierarchical networks, and also generalizes the previous models,
    \item proposes a spectral algorithm to estimate the latent hierarchy, and
    \item shows that this algorithm is consistent.
\end{enumerate}

\begin{figure}[!ht]
    \centering
    \vspace{-0.1in}
	\includegraphics[width=5in]{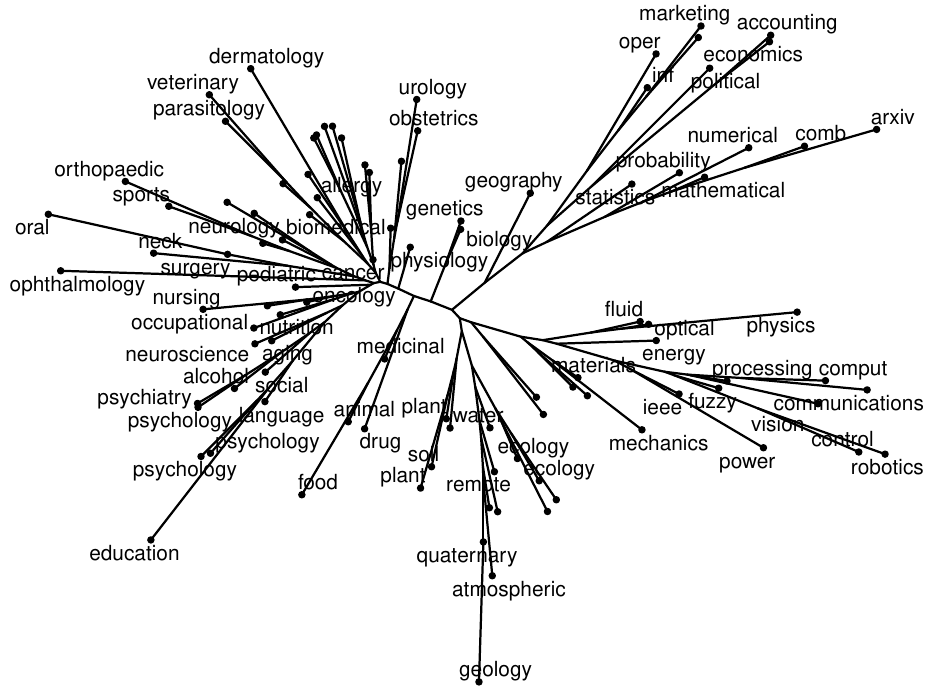}
     \caption{%
     This visualization comes from the citation pattern of 22,688 academic journals.  Each leaf in this tree corresponds to an estimated dimension (cluster) in the network. Starting from the top right and going clockwise, this hierarchy finds
     \textit{economics}, \textit{math} (statistics cluster is closer to the core of the tree, between economics and math), then \textit{physics}, \textit{engineering}, \textit{materials}, \textit{earth sciences (geology, ecology, plant, etc)}, \textit{pharmaceuticals}, \textit{education/psychology}, then a big mix of medical areas.}
   \label{fig:journal_tree}
\end{figure}

\section{$\T$-Stochastic Graphs} \label{sec:tsg}

In the $\T$-Stochastic Graph model, there are, in fact, two types of graphs: (i) the observed social network and (ii) the latent hierarchy $\T$. These are very different graphs. To emphasize their difference, we refer to the observed social network only via its adjacency matrix $A \in \R^{n \times n}$. In this social network $A$, there are $n$ people or ``nodes'' and $A_{ij}\ge0$ measures the strength of the relationship between nodes $i$ and $j$:
\[A_{ij}=0 \Longleftrightarrow \text{nodes $i$ and $j$ are not friends}.\]
The latent hierarchy is a tree graph that we refer to as $\T = (V, E)$, with node set $V = \{1, \dots, n, \dots, N\}$ and edge set $E = \{(u,v)|\text{$u$ and $v$ are connected}\}$. Given a set of weights $\{w_{uv}\}$ assigned to every edge $(u, v)\in E$, we define $d(i, j)$, the distance between any pair of nodes $i, j\in V$, as the summation of edge weights along the shortest path between $i$ and $j$, and we say $d(\cdot, \cdot)$ is an \textbf{additive distance}\footnote{The usual definition on additive distance does not start from edge weights \citep{choi2011learning, erdHos1999few}, we employ this equivalent definition for simplicity, more details can be found in Remark \ref{rmk:differnt_def_additive} in Appendix \ref{appendix:cov->dist}.} on $\T$.
In both $A$ and $\T$, we will refer to the \textbf{degree} of node $i$, denoted by $deg(i)$, as the number of connections to node $i$. Importantly, these degrees are different in $A$ and $\T$. In latent tree $\T$, $deg(i) = \sum_{j} \mathds{1}\{ (i,j) \in E \}$; in the observed social network $A$, $deg(i) = \sum_j \mathds{1}\{A_{ij} > 0\}$. Whether $deg$ refers to the degree in $\T$ or $A$ will always be clear from the context. In $\T$, all nodes with degree one are called \textbf{leaf nodes}, all nodes with degree $>1$ are called \textbf{internal nodes}.

The $\T$-Stochastic Graph model (Definition \ref{def:tsg} below) relates $\T$ to $A$.  In this model, $A$ is random and $\T$ is fixed; in particular, $\T$ is a latent structure that describes the probability distribution for $A$. In statistical estimation (Section \ref{sec:estimation}), we observe $A$ and we want to estimate $\T$.

\begin{definition}\label{def:random_graph}
In this paper, we say a symmetric adjacency matrix $A \in \R^{n \times n}$ is a \textbf{random graph} if  it has independent elements \footnote{For independence, we mean all $\{A_{ij}\}_{0<i\leq j< n}$ are independent, notice that symmetry constrains $A_{ij} = A_{ji}$} with positive variance and $\E(A_{ij}) = \lambda_{ij}$, for some set of $\lambda_{ij}$'s.
\end{definition}

For example, $A_{ij}$ could be Bernoulli($\lambda_{ij}$) or Poisson($\lambda_{ij})$. 
The $\T$-Stochastic Graph model parameterizes these $\lambda_{ij}$'s using an additive distance $d(\cdot, \cdot)$ on $\T$.

\begin{definition}\label{def:tsg}

Suppose that  $A\in \R^{n \times n}$ is a random graph and $\T$ is a tree graph. 
We say that $A$ is a $\T$-\textbf{Stochastic Graph} if the nodes $\{1, \dots, n\}$ in $A$ match the leaf nodes $\{1, \dots, n\}$ in $\T$ such that for all  $i \ne j$,
\begin{equation}\label{eq:lambdadef}
\lambda_{ij} = \exp(-d(i,j)),
\end{equation}
where $d(i, j)$ is the distance between nodes $i$, $j$ in $\T$ and $\lambda_{ij}$ comes from the nodes $i,j$ in $A$. 
\end{definition}

\begin{figure}[!ht] %
   \centering
   \includegraphics[width=6.5in]{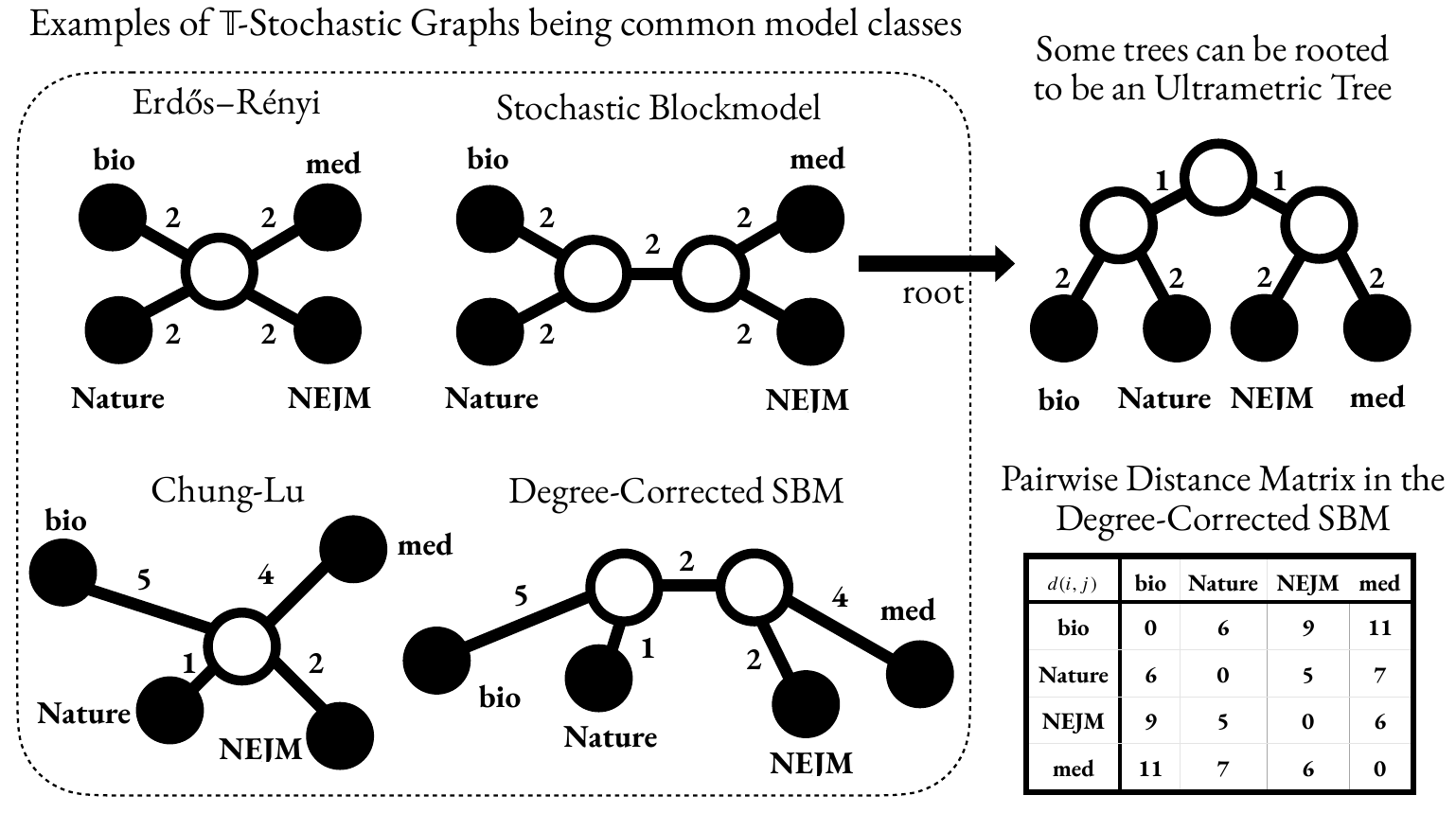} 
   \caption{This figure presents some toy model $\T$-Stochastic Graphs, imagine NEJM is the New England Journal of Medicine (a high-status medical journal), med is a lower-status medical journal, Nature is a high-status Biology journal, and bio is a lower status Biology journal. If $\T$ is a star graph with equal edge weights, then the $\T$-Stochastic Graph is an Erd\H{o}s-R\'enyi graph \citep{erdHos1999few}.  If $\T$ is a star graph with different edge weights, then it is a Chung-Lu graph \citep{chung2002connected}. If there are two internal nodes in $\T$, then the $\T$-Stochastic Graph is a (Degree-Corrected) Stochastic Blockmodel \citep{holland, karrer2011stochastic} with two blocks. The bottom right tree captures both the clustering structures and high/low status of each journal.}
   \vspace{-0.1in}
   \label{fig:tsg_common_class_ultrametric_distance}
\end{figure}

Figure \ref{fig:tsg_common_class_ultrametric_distance} presents some examples of $\T$ with four leaf nodes ($n=4$), each corresponding to a common model class.
Superficially, the most catching thing about $\T$-Stochastic Graphs is that Equation \eqref{eq:lambdadef} contains both an $\exp(\cdot)$ transformation and a distance $d(\cdot, \cdot)$. However, we will see that this often arises naturally.  For example, the previously defined hierarchical models in \cite{clauset2008hierarchical, lei, li} are not defined with exponential transformation and additive distance, but one can construct an additive distance $d(\cdot, \cdot)$ on $\T$ so that those models are $\T$-Stochastic Graphs satisfying Equation \eqref{eq:lambdadef}.

The most important feature of the $\T$-Stochastic Graph model is that it does not make any constraints on the latent hierarchy $\T$ or the additive distance $d(\cdot, \cdot)$. On the contrary, the hierarchical models in \cite{clauset2008hierarchical, lei, li} implicitly make various assumptions on $\T$ and $d(\cdot, \cdot)$. Appendix \ref{appendix:other_models} provides a detailed discussion about these constrains, and proves that all these models are special cases of $\T$-Stochastic Graphs. Among these constrains, one that is imposed on all these models is the following property:

\begin{definition} \label{def:ultrametric}
A tree graph $\T$ with distance $d(\cdot, \cdot)$is \textbf{ultrametric} if there exists a \textbf{root node} $r\in V$, such that all leaf nodes in $\T$ are equidistant to $r$, i.e., $d(i, r)\equiv c$ for any leaf node $i$.
\end{definition}

For example, the two trees on the bottom of Figure \ref{fig:tsg_common_class_ultrametric_distance} is not ultrametric.  The top left tree is ultrametric. The top right tree 
does not satisfy Definition \ref{def:ultrametric}, however, it can be ``rooted'' in a way that does not change the pairwise distances between leaves and makes the tree satisfy Definition \ref{def:ultrametric}. 
In such cases, it is often said that this tree is ultrametric.

Roughly speaking, the ultrametric constraint makes the social network more homogeneous. For example, the bottom two trees in Figure \ref{fig:tsg_common_class_ultrametric_distance} correspond to model classes that account for degree heterogeneity, which prevents these trees from being rooted ultrametrically. 
This is not the only pattern that prevents ultrametric rooting, but it is one that is easy to see in illustrations.

Section \ref{sec:instability} discusses how the ultrametric assumption leads to algorithms that have a fundamental instability on empirical networks; we propose a simple diagnostic plot to identify and help understand this instability.

\begin{remark}\label{rmk:multiply_c}
Sometimes we consider $\lambda_{ij} = c \exp(-d(i,j))$ for a positive constant $c>0$. This generalization is equivalent to allowing negative distance $d(i, p(i))$, where $i$ is any leaf node and $p(i)$ is the neighbor of leaf node $i$. See Remark \ref{appendix_rmk:multiply_c} in Appendix \ref{appendix:tgb->tsg} for more explanation.
\end{remark}

\begin{remark}\label{rmk:negative_edge_weights}
Negative edge weights between two internal nodes are different from negative edge weights between a leaf node and its parent. The former one leads to a violation of homophilous structure\footnote{Homophilous structure means nodes closer in the tree are more likely to be friends than nodes further apart, which is also referred to as ``assortative'' or ``affinity'' in other literature, see Appendix \ref{appendix: other_models_assortativity_affinity} for further discussions.} while the latter does not. See Appendix \ref{appendix: negative_edge_weight} for more discussions. In the following sections, we assume internal edge weights to be non-negative unless explicitly stated otherwise. In particular, this issue arises in Section \ref{sec:bus} and Section \ref{sec:hse}. 
When we refer to ``negative edge weights'' in this paper, it specifically means negative edge weights between internal nodes.
\end{remark}

\begin{remark}\label{rmk:identifiability}
For identifiability purposes, we assume $deg(u)\geq 3$ for any internal node $u$ in $\T$ following \citep{pearl1988probabilistic, choi2011learning}. Furthermore, we assume $w_{uv}\neq 0$ for any $(u, v)\in E$ to avoid two nodes $u$ and $v$ from being identical to each other.

\end{remark}

\subsection{A fundamental instability in estimating $\T$ with  ``top-down'' approaches} \label{sec:instability}

\begin{quote}
This section examines the spectral properties of multiple large empirical social networks to see why previous ``top-down'' approaches are insufficient and thus why we should consider the more general $\T$-Stochastic Graphs.  Later sections can be read before this section.
\end{quote}

Previous statistical approaches to hierarchical clustering in social networks have primarily focused on ``top-down'' partitioning. In ``top-down'' approaches, the $n$ nodes are iteratively split into two (or more) clusters \citep{lei, li, aizenbud2021spectral}. 
To partition a group of $n$ nodes into two groups, these approaches typically construct some $n\times n$ matrix (e.g. $A$ or a graph Laplacian), compute its second eigenvector $\hat x \in \R^n$, and partition node $i$ based upon the +/- sign of $\hat x_i$, the $i$th element of that vector.

\begin{quote}
\textbf{The fundamental instability of top-down splitting:} \textit{The splitting eigenvector $\hat x \in \R^n$ for a $\T$-Stochastic Graph is a random vector. So, if an element $\hat x_i$ is close to zero, then slight perturbations can change the cluster that node $i$ is assigned to. While previous theoretical results show that under certain assumptions, $\hat x_i$ is well separated from zero with large probability; empirically, we found the most common values of $\hat x_i$ are often very close to zero. This inconsistency between theoretical results and empirical evidence suggests a gap between the assumptions and the data. Moreover, the algorithms are sensitive to this gap.} 
\end{quote}

\begin{figure}[!ht] %
   \centering
   \textbf{The instability diagnostic examines whether the splitting vector has values in two well-separated clusters.  Unfortunately, most values  are often very close to zero.}
   \includegraphics[width=6in]{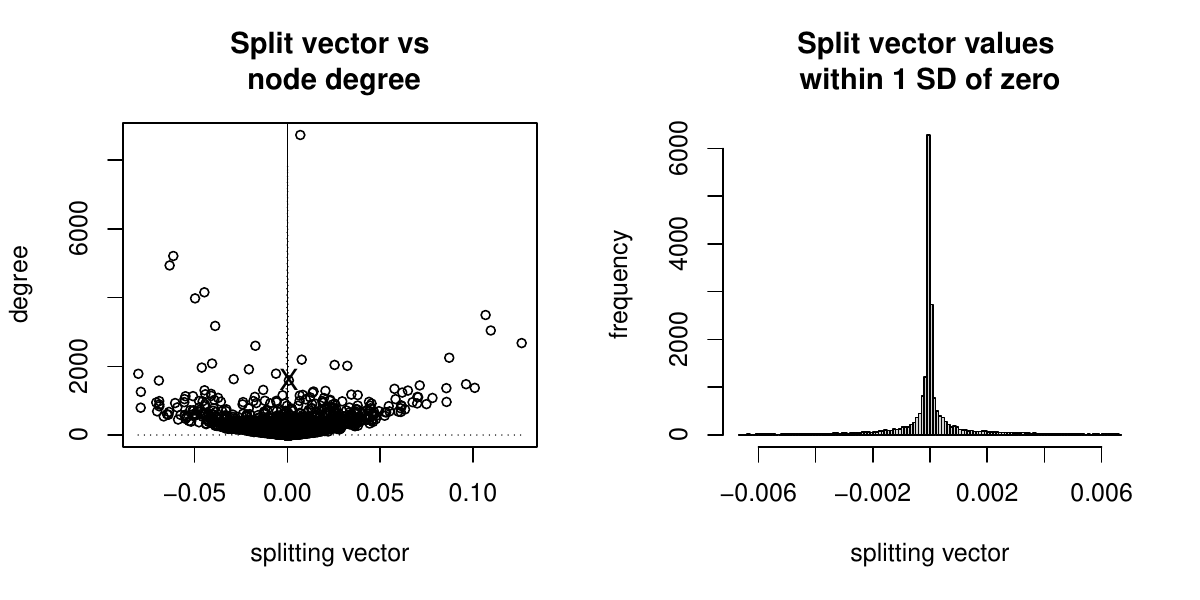} 
   \vspace{-0.25in}
   \caption{Left panel: Each point is a journal in the symmetrized  journal citation graph. The horizontal axis corresponds to the splitting vector $\hat x \in \R^{22,688}$ (the second eigenvector of $A$); the vertical axis gives the journal degree in the symmetrized graph.  The dashed line gives a kernel density estimate for the elements in $\hat x \in \R^{22,688}$. There is also a solid line at $\hat x_i = 0$. \textit{The points $\hat x_i$ are so tightly clustered around zero that the density estimate of $\hat x_i$ and the vertical line are nearly indistinguishable.} Right panel: This is a histogram of the values $\hat x_i$ that are within a single standard deviation of zero.  Notice that the scale of the horizontal axis in the left panel is zoomed in.  The single mode at zero is highly peaked.}
   \label{fig:secondeighist}
\end{figure}

Figure \ref{fig:secondeighist} presents $\hat x$ of the same citation graph in Figure \ref{fig:journal_tree}.  In the left panel, each point is an academic journal (i.e. a node in the citation graph).  The leftmost point is Nucleic Acids Research and the rightmost point is JAMA (Journal of the American Medical Association). The journals on the right (i.e. large positive values in $\hat x$) tend to be  prestigious medical journals, while the points on the left (i.e. large negative values in $\hat x$) tend to be prestigious molecular biology and genomics journals.  The vertical axis gives the degree of the node (while citations are directed edges, we have symmetrized all edges).  The highest point is PLOS ONE.  The journal with an X over that is close to the boundary is Journal of Immunology.  Both PLOS ONE and Journal of Immunology seem to be very important to both the left and right sides. Splitting them at the first step seems unwise.

For the Statistics literature, the splitting vector $\hat x$ makes an even more regrettable error; JRSS-B is on the left and JASA is on the right.  So, \textbf{two of the most prestigious statistics journals would be put into different clusters in the first split and would never reconvene in a top-down approach.} JRSS-B is connected to 194 other journals and JASA to 513 other journals. So, this is not simply a problem of errors on small degree nodes.

Previous theoretical results made various types of assumptions to ensure values in $\hat x$ form two well-separated clusters. For example, \cite{li} proposed a hierarchical model with balanced, \textbf{binary}, and \textbf{ultrametric}\footnote{Appendix \ref{appendix:other_models} gives rigorous definitions of the binary and ultrametric assumptions and discusses how both Li's model and Lei's model are $\T$-Stochastic Graphs with these assumptions enforced on $\T$.} constraints. Theorem 1 in \citep{li} shows that under this model, the population splitting vector only has two values, one positive and one negative, with the same magnitudes. 
Theorem 2.1 in \citep{lei} presented similar results on eigenvectors of the Laplacian matrix, when the balanced assumption is removed. Both results imply a two-mode structure for the sample splitting vector: one mode on the positive side and one mode on the negative side. In summary, if the random graph $A$ is generated under an ultrametric and binary hierarchical model, $\hat x_i$ are expected to be well separated from zero.

\begin{theorem}\label{thm:ultrametric_split}
(Theorem 1 in \cite{li})
Let $A\in \R^{n\times n}$ be generated from Li's model, then the second eigenvalue of $\E A$ has multiplicity 1 and the entries of the corresponding eigenvector $u_2$ obey
\begin{align*}
    u_{2, i} = \pm 
    \begin{cases}
    1/\sqrt{n}, \quad & i\in \mathcal{G}_0 \\
    -1/\sqrt{n} , & i\in \mathcal{G}_1 
    \end{cases}
\end{align*}
where $\mathcal{G}_0$ and $\mathcal{G}_1$ give the correct true partition of the nodes.
\end{theorem}

As a consequence, these model assumptions are easy to diagnose by plotting $\hat x_i$ against the degree of node $i$ and making a histogram of the $\hat x_i$ values. Figure \ref{fig:splitvecs} presents diagnostic plots for some large networks summarized in Table \ref{tab:splitvecs}. Similar to Figure \ref{fig:secondeighist}, all histograms present a sharp peak at zero, and many near-zero $\hat x_i$ values correspond to high-degree nodes. 
Table \ref{tab:splitvecs} provides additional summaries. Most splitting vectors are highly concentrated around zero (3rd to 5th columns in Table \ref{tab:splitvecs}). Some networks (e.g. Facebook Page and Wiki Page) even have more than 95\% elements within the range of $\pm 0.1\sigma$, where $\sigma$ is the sample standard deviation. Perhaps another threshold could be found to avoid cutting at the highly peaked zero value. Alternatively, one could apply $k$-means to separate $\hat x_i$ (e.g. \texttt{HCD-Spec} in \cite{li}). However, these variations do not escape the fundamental problem that most splitting vectors have only one mode.
The last column presents the $p$-value for Silverman's unimodality test \citep{silverman1981using, JSSv097i09}. 
Despite these networks having tens of thousands or millions of nodes, only one of the nine networks rejects the unimodal null hypothesis. This conflicts with previous ultrametric models. Though the previous study presented successful results of bi-partition methods on some real datasets, they only investigated small or moderate-sized networks (about hundreds of nodes). For large social networks, we have not found an example of diagnostic plot that looks different from those in Figure \ref{fig:splitvecs}.

\begin{table}[!ht]
    \centering
    \textbf{Splitting vector values in most large networks are highly concentrated around zero, and fail to reject the unimodal hypothesis test.}
    \footnotesize
    \begin{tabular}{ccccccc}
    \hline
    Network & $|V|$ & $|E|$ 
    & $<0.1\sigma$
    & $<0.05\sigma$ & $<0.01\sigma$ & p-value 
    \\
    \hline
    Journal Citation \citep{ammar-etal-2018-construction} & 22688 & 474841 & 0.56 & 0.45 & 0.20 & 0.54 
    \\
    Facebook (Social) \citep{leskovec2012learning} & 4039 & 88234  & 0.86 & 0.76 & 0.61 & <0.01
    \\
    Epinions \citep{richardson2003trust} & 75879 & 508837  & 0.84 & 0.75 & 0.42 & 0.16
    \\
    Slashdot \citep{leskovec2009community} & 82168 & 948464  & 0.51 & 0.33 & 0.10 & 0.42
    \\
    Wiki (Vote) \citep{leskovec2010signed} & 7115 & 103689  & 0.53 & 0.39 & 0.19 & 0.87
    \\
    Facebook (Page) \citep{rozemberczki2019multiscale} & 22470 & 171002  & 0.97 & 0.94 & 0.81 & 0.56
    \\
    Github \citep{rozemberczki2019multiscale} & 37700 & 289003  & 0.35 & 0.21 & 0.05 & 0.45
    \\
    LastFM \citep{feather} & 7624 & 27806  & 0.81 & 0.71 & 0.42 & 0.89
    \\
    Wiki (Page) \citep{snapnets}  & 2464429 & 76229780  & 0.98 & 0.95 & 0.60 & 0.48 \\
    \hline
         
    \end{tabular}
    \caption{Summary statistics for splitting vectors in large social networks. 
    The third to fifth columns measure the level of concentration around zero. Let $\sigma$ be the standard deviation of $\hat x_i$, then they measure $\sum_i\mathds{1}\{|\hat x_i| < c\sigma\}/n$ with $c = 0.1, 0.05, 0.01$ accordingly. The last column displays the p-value for Silverman's test \citep{silverman1981using}. The null hypothesis is that $\hat x_i$ comes from a unimodal distribution.}
    \label{tab:splitvecs}
\end{table}

\begin{figure}[!ht] %
   \centering
   \textbf{Instability diagnostic plots for large social networks exhibit a single mode around zero. Moreover, many nodes around zero are important high-degree nodes.}
   \includegraphics[width=6.5in]{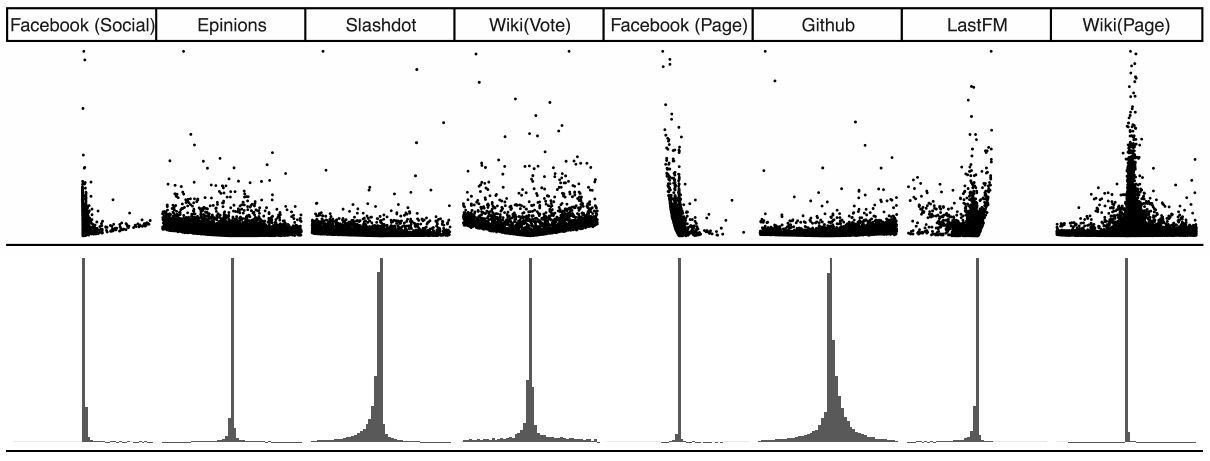} 
   \caption{Instablility diagnostic plots for social networks in Table \ref{tab:splitvecs}, plotted within the range of $\scaleto{\left[-2/\sqrt{n},2/\sqrt{n}\right]}{9.5pt}$, where $n$ is the number of vertices in each network.}
   \label{fig:splitvecs}
\end{figure}

While the two-mode structure implied by the binary and ultrametric assumptions conflicts with the one peak pattern in empirical data, networks sampled from the $\T$-Stochastic Graph can display a unimodal structure and still be estimable. Appendix \ref{appendix:parametric_bootstrap} discusses this in more details by comparing the diagonostic plots of ``parametric bootstrap'' networks from Li's model and our $\T$-Stochastic Graphs.

\begin{figure}[!ht] %
   \centering
   \includegraphics[width=4.7in]{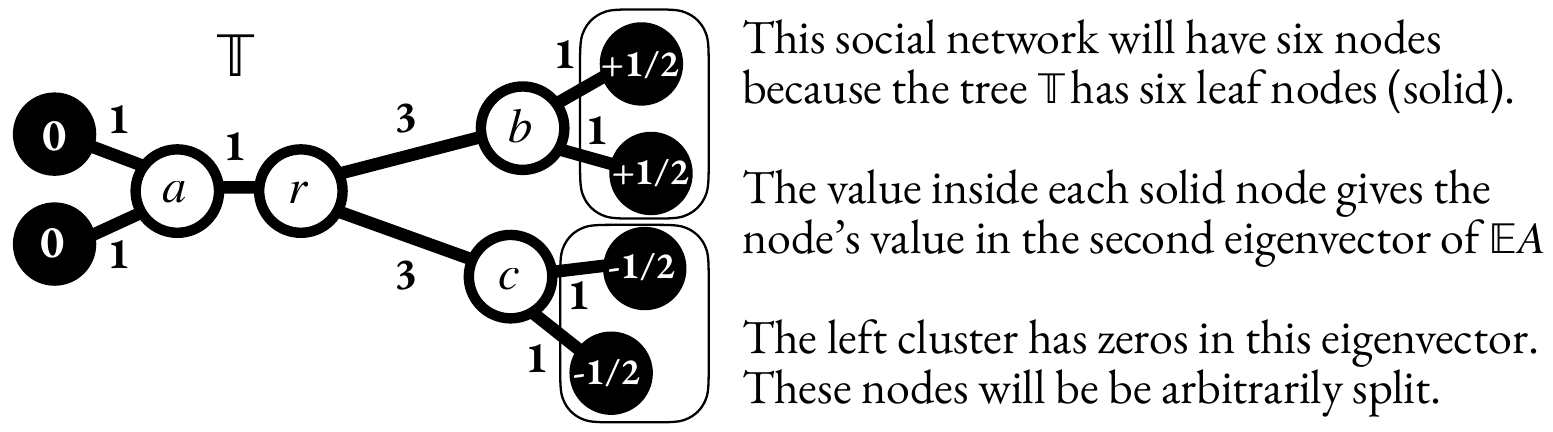} 
   \caption{In this toy model, the second eigenvector of $\E A$ assigns the two nodes on the left values of zero.  When this happens, ``top-down'' hierarchical clustering will be unstable.}
   \label{fig:hollow_nodes}
\end{figure}

To understand how \textbf{non-ultrametricity allows zero values in the  splitting vector}, Figure \ref{fig:hollow_nodes} gives a small $\T$ for a toy model. 
In this example, $\T$ is a non-ultrametric tree with six leaf nodes (solid black) and the resulting adjacency matrix $A$ is $6 \times 6$. Define $x \in \R^6$ as the second eigenvector of $\E A$. This eigenvector assigns a value $x_i$ to each of the six leaf nodes in the tree and these values are given inside the solid black. While the sign of $x_i$ correctly partitions the nodes on the right, the nodes on the left descend from the internal node $a$ and are both assigned value 0 by $x$. Given the observed graph $A$, we compute its second eigenvector $\hat x$ as an estimator of $x$ and use $\hat x$ to build the first split. In $\hat x$, the two nodes that descend from $a$ will be arbitrarily split.

We do not emphasize the difference between binary trees and non-binary trees. This is because, in the $\T$ Stochastic Graph, non-binary trees are binary trees with some zero edge weights. That said, previous bi-partition methods implicitly assume strictly positive edge lengths to assure an eigen gap between the second and the third eigenvalues (Equation (3) in \citep{lei} and Equation (8) in \citep{li}). When zero edge weights are allowed, the multiplicity of the second eigenvalue can be greater than one. Consequently, \textbf{the splitting vector for non-binary trees are not unique}. This non-uniqueness makes it impractical to extend bi-partition methods to non-binary structures whereas our \pps algorithm does not suffer from this problem.

\section{Six alternative constructions for $\T$-Stochastic Graphs} \label{sec:intuitions}

\begin{quote}
This section presents six alternative constructions of hierarchies in social networks. While these models do not initially resemble the $\T$-Stochastic Graph, we show that all these models are equivalent to it. These six alternatives help to clarify the types of assumptions made by the $\T$-Stochastic Graph model and identify relationships to other models. Moreover, the intuition from Section \ref{sec:tgb} on the Graphical Blockmodel motivates a spectral estimation procedure.
\end{quote}

Section \ref{sec:tsg->sbm} first builds some intuition of $\T$-Stochastic Graph by presenting a connection to the Degree-Corrected Stochastic Blockmodel (DCSBM). Starting from Section \ref{sec:tgb}, six equivalent models are introduced, below is a brief description of each model, numbered by the subsections that discuss them.

\noindent \textbf{\ref{sec:tgb} ($\T$-Graphical Blockmodel):} This is a Degree-Corrected Stochastic Blockmodel \citep{karrer2011stochastic} where the block-to-block connectivity matrix $B$ is the covariance matrix for a subset of variables in a Gaussian Graphical Model (GGM) \citep{lauritzen1996graphical} with conditional independence structure $ \T$.  Section \ref{sec:estimation} shows how this model provides a clear path for statistical estimation of $\T$ via popular and well-studied estimation techniques.

\noindent \textbf{\ref{sec:tg-rdpg} ($\T$-Graphical RDPG):} This is a Random Dot Product Graph (RDPG) \citep{rdpg} 
where the latent positions are not independent but sampled in a certain way from a Gaussian Graphical Model where the conditional independence structure is encoded in 
$ \T$.

\noindent \textbf{\ref{sec:osb} ($\T_r$-Overlapping Blockmodel):} %
This is an Overlapping Stochastic Blockmodel \citep{latouche2011overlapping} defined by a rooted tree $\T_r$. Leaf nodes in $\T_r$ correspond to nodes in $A$; non-leaf nodes in $\T_r$ correspond to blocks in the Overlapping Stochastic Blockmodel. %
Any leaf node belongs to all blocks between that leaf and the root.

\noindent \textbf{\ref{sec:tds} ($\T_r$-Top Down Stochastic Process):}  Given a rooted tree $\T_r$ (which need not be ultrametric), start the process at the root node $r$.  For each node in the tree, you have some probabilities of walking along any edge away from the root or stopping at that node. When you stop at a node $u$, select two leaf nodes $i$ and $j$ that descend on a path from $u$, away from the root, and connect them.

\noindent \textbf{\ref{sec:bus} ($\T$-Bottom Up Stochastic Process):} Select a leaf node $i$ at random, with probability proportional to some $\theta_i$. Take a non-backtracking random walk on $\T$, starting at leaf $i$ and terminating at a leaf $j$. Connect $i$ to $j$.

\noindent \textbf{\ref{sec:hse} ($\T$-HSE):} \textit{Hierarchical Stochastic Equivalence} (HSE) is an axiom related to ``Stochastic Equivalence,'' the axiom that first motivated the Stochastic Blockmodel in \cite{holland}. 

The following sections require more notation. Consider any tree graph $\T = (V, E)$. 
One can ``root'' $\T$ by picking any node $r\in V$ as the root, we denote this rooted tree as $\T_r$. Let $V_\ell \subset V$ be the set of leaf nodes in $\T$, i.e.,  $V_\ell = \{i\in V: deg(i) = 1\}$. For any leaf node $i \in V_\ell$, we use $p(i)$ to denote its only neighbor node, sometimes we also say this is the ``parent'' of $i$. Define $V_z \subset V$ as the set of nodes in $\T$ that are connected to a leaf node, i.e., $V_z = \{u\in V: u = p(i) \text{ for some $i\in V_\ell$}\}$.

\subsection{Every $\T$-Stochastic Graph is a Degree-Corrected Stochastic Blockmodel}\label{sec:tsg->sbm}

\begin{definition}\label{def:dcsbm}
The \textbf{Degree-Corrected Stochastic Blockmodel} (DCSBM) 
is a random graph with
$$\lambda_{ij}^{DC} = \theta_i \theta_j B_{z(i), z(j)}$$ 
for node specific degree parameters $\theta_i \in \R_+$, block assignments $z(i) \in \{1,\dots, k\}$, and a connectivity matrix $B \in \R_+^{k \times k}$ that is \textbf{full rank}. If $\theta_i \equiv 1, \forall i$, then it reduces to a Stochastic Blockmodel (SBM).
\end{definition}

\begin{theorem}\label{thm:tsg->dcsbm}
Every $\T$-Stochastic Graph is a DCSBM.
\end{theorem}

The key component of the proof involves construction; this construction gives essential insight into the $\T$-Stochastic Graph. Given $\T$, construct the parameters $z, \theta,$ and $B$ for a DCSBM as:
\begin{enumerate}
\itemsep0.1em 
    \item Consider $V_z\subset V$, the set of nodes in $\T$ that are connected to a leaf node. Each node in $V_z$ will correspond to a block in the DCSBM.
    \item For each node  $u \in V_z$, find all leaf nodes $i$ that connect to $u$ and define $z(i) =u$; they belong to this block. Define $\theta_i = \exp(-d(i,u)) = \exp(-w_{iu})$. If $\theta_i$ is large, then $i$ is close to $u$.
    \item Let $k = |V_z|$ be the number of nodes in $V_z$.  Define $B \in \R_+^{k \times k}$ such that for each pair of nodes $u,v \in V_z$, $B_{uv} = \exp(-d(u,v))$.
\end{enumerate}
We wish to show that for each $i\neq j$, $\lambda_{ij}$ in the $\T$-Stochastic Graph equals $\lambda_{ij}^{DC}$ in the DCSBM. This is illustrated in the following sequence of equalities. A comprehensive proof also needs to show that $B$ is full rank, this is can be found in Appendix \ref{appendix:tsg->tgb}. 
\begin{eqnarray*}
\lambda_{ij} &=& 
\exp(-d(i, j)) = 
\exp\left(-\left(d(i, z(i)) + d(z(i) , z(j)) + d(z(j), j)\right)\right)\\
&=& {\overbrace {\textstyle \exp(-d(i, z(i)))}^{\mathclap{\theta_i}}} \quad {\overbrace {\textstyle \exp(-d(z(i) , z(j)))}^{\mathclap{B_{z(i), z(j)}}}} \quad {\overbrace {\textstyle \exp(-d(z(j), j))}^{\mathclap{\theta_j}}} \; = \; \theta_i B_{z(i), z(j)} \theta_j
= \lambda_{ij}^{DC}.
\end{eqnarray*}

The $\T$-Stochastic Graph in Figure \ref{fig:tsg_dcsbm_tree} provides an illustration of this process.  
On the left side of Figure \ref{fig:tsg_dcsbm_tree}, node $p(i)$ is connected to leaf node $i$, therefore $p(i)\in V_z$, and $z(i) = p(i)$ represents the block that node $i$ belongs to. %
The magnitude of $\theta_i$ is decided by the distance between $i$ and $p(i)$. If the distance $d(i, p(i))$ increases, then the degree correction parameter $\theta_i = \exp(-d(i, p(i)))$ decreases. In other words, if the distance between a leaf node and its parent increases, the expected number of connections that this leaf node will have in the sampled social network decreases. This matches our intuition about hierarchies. When leaf node $i$ is close to its neighbor node $p(i)$, it should be close to other leaf nodes in the tree and thus is more likely to be connected to them, leading to a higher expected degree. 
The same intuition holds for internal nodes. Since $B_{z(i)z(j)} = \exp \left(-d(p(i), p(j)) \right)$, the closer $p(i)$ and $p(j)$ are, the more likely that two nodes from these two blocks will form a connection.

While the above statement starts from a $\T$-Stochastic Graph and explains how there is an equivalent DCSBM, the reverse is not true in general. If one wishes to construct a DCSBM that is also a $\T$-Stochastic Graph, the connectivity matrix needs to satisfy certain constraints. Section \ref{sec:tgb} presents a new hierarchical model class to explain this.

\begin{figure}[htbp] %
   
   \centering
      \textbf{Every $\T$-Stochastic Graph is also a DCSBM, where the connectivity matrix $B$ is determined by the distance matrix of a subset of nodes in $\T_z$}
       \vspace{.1in}
     
   \includegraphics[width=6.5in]{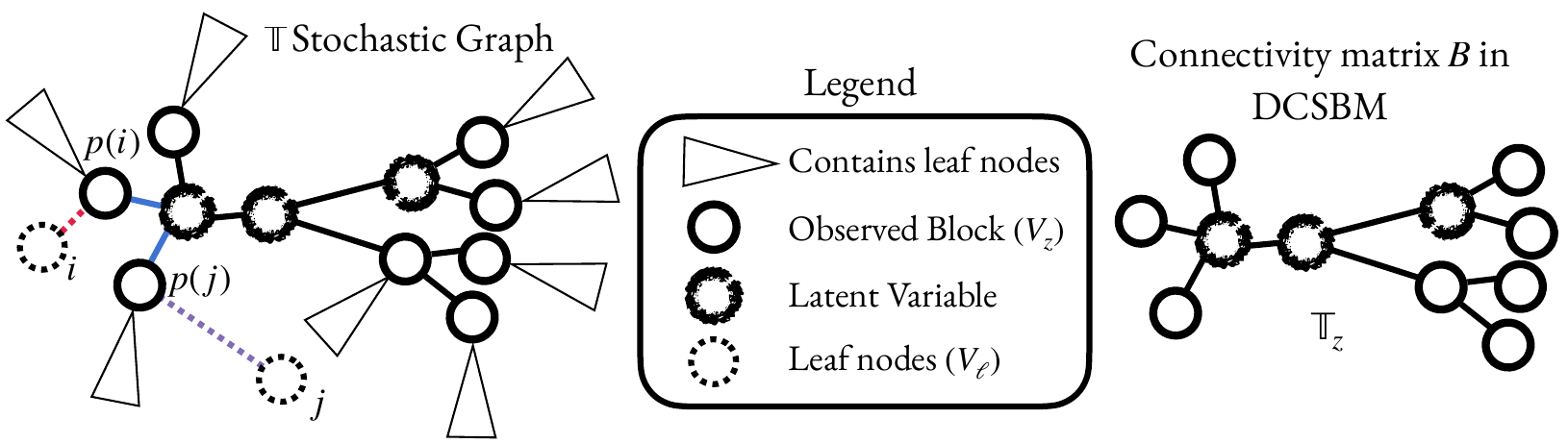} 
   \caption{On the left, each triangle contains several leaf nodes of $\T$; each leaf node corresponds to a node in the observed social network $A$; the dashed circles serve as examples of leaf nodes $i$ and $j$.  To find the $\T_z$, remove all leaf nodes from $\T$.  There are two types of nodes in $\T_z$ that we call ``observed blocks'' and ``latent variables''; latent variables are not connected to any leaf nodes of $\T$ and observed blocks are connected to leaf nodes of $\T$.  We call them ``observed blocks'' because every observed block node in $\T$ corresponds to a block in a DCSBM. The latent variables are not connected to leaf nodes and thus do not appear in $B$. 
}
   \label{fig:tsg_dcsbm_tree}
\end{figure}

\subsection{Using a Gaussian Graphical Model on $\T$ to parameterize DCSBMs that produce $\T$-Stochastic Graphs}\label{sec:tgb}

The last section shows that all $\T$-Stochastic Graphs are DCSBMs, but the converse is not true: many DCSBMs are not $\T$-Stochastic Graphs. To parameterize the DCSBMs that do generate $\T$-Stochastic Graphs, we make a connection to Gaussian Graphical Models.

\begin{definition}\label{def:ggm}
We say a random vector $W$ that follows a multivariate Gaussian distribution with mean vector $\mu\in \R^p$ and covariance matrix $\Sigma\in \R^{p\times p}$ comes from a Gaussian Graphical Model (GGM) on graph $\mathbb{G} = (V,E)$ with $V = \{1, \dots, p\}$ if 
\begin{equation}\label{eq:ggm_def}
(\Sigma^{-1})_{ij} = 0 \  \Longrightarrow \ (i,j) \not \in E
\end{equation}
for all pairs $i,j \in V$. Throughout this paper, we presume that $\mu = 0$, $\Sigma_{ii} = 1$ for all $i$, and $0<\vert \Sigma_{ij} \vert<1$ for all $i\neq j$.\footnote{The first two assumptions ($\mu= 0$ and $\Sigma_{ii} = 1$) are for simplicity and the equivalence between models can be easily extended to general cases, see Appendix \ref{appendix:covariance} for more discussions. The $0<\vert \Sigma_{ij} \vert<1$ assumption is for identifiability purposes, it makes sure two nodes connected by an edge are neither perfectly dependent nor independent. See \citep{choi2011learning, pearl1988probabilistic} for similar assumptions.}
\end{definition}

The $\T$-Graphical Blockmodel proposed below is a DCSBM with the connectivity matrix $B$ defined by a GGM on $\T$.

\begin{definition}\label{def:tgb}
The graph $A$ follows the $\T$-\textbf{Graphical Blockmodel} if $A$ is a DCSBM with block assignments $z(\cdot)$ and connectivity matrix $B$ such that
\vspace{-0.1in}
\begin{enumerate}
\itemsep0em 
    \item nodes in graph $A$ are indexed by nodes in $V_\ell$, the set of leaf nodes in $\T$, 
    \item rows and columns of $B$ are indexed by nodes in $V_z$, parents of leaf nodes in $\T$,
    \item node $i$ belongs to block $u$ if $u$ is the parent of $i$ in $\T$, i.e., $z(i) = u \Longleftrightarrow p(i) = u$,
    \item for any pair of nodes $u, v \in V_z$, 
\begin{equation}
    B_{uv} = \left|\Sigma_{uv}\right|, 
\end{equation}
where $\Sigma$ is the covariance matrix of a GGM on $\T$.
\end{enumerate}

\end{definition}

\begin{theorem}\label{thm:tsg_tgb}
A random graph is a $\T$-Graphical Blockmodel if and only if it is a $\T$-Stochastic Graph.
\end{theorem}

The proof of Theorem \ref{thm:tsg_tgb} can be found in Appendix \ref{appendix:tsg_tgb}. The following propositions shed light on this perhaps unexpected equivalence in Theorem \ref{thm:tsg_tgb}. Specifically, Proposition \ref{prop:cov->dist} provides a way to construct an additive distance $d(\cdot, \cdot)$ from a GGM on $\T$ such that the covariance of the GGM and the distance 
can be connected by an exponential transformation. Proposition \ref{prop:dist->cov} explains the other direction: given any distance $d(\cdot, \cdot)$ on tree $\T$, an exponential transformation can build the covariance matrix of some GGM on $\T$. The proof of Proposition \ref{prop:cov->dist} and \ref{prop:dist->cov} can be found in Appendix \ref{appendix:cov->dist} and \ref{appendix:dist->cov}.

\begin{prop}\label{prop:cov->dist}
(\cite{erdHos1999few})
If $W\sim N(0,\Sigma)$ comes from a GGM on tree $\T = (V, E)$, for any two nodes $i, j \in V$, define
\begin{equation}\label{eq:ggm1}
 d(i,j) = -\log(|\Sigma_{ij}|),
\end{equation}
then $d(\cdot, \cdot)$ is an additive distance on $\T$. That is, there exists a set of edge weights $w_{uv}$ such that $d(i, j)$ is the summation of $w_{uv}$ on the shortest path between node $i$ and $j$.
\end{prop}

\begin{prop}\label{prop:dist->cov}
Given any tree $\T = (V, E)$ and an additive distance $d(\cdot, \cdot)$, define $\Sigma$ with 
\[\Sigma_{ij} = \exp(-d(i, j))\]
for any two nodes $i, j \in V$,
then $\Sigma$ is positive definite and satisfies Equation \eqref{eq:ggm_def}.
\end{prop}

\subsection{$\T$-Graphical RDPG} \label{sec:tg-rdpg}

An alternative formulation of the $\T$-Stochastic Graph is the following parameterization of the Random Dot Product Graph (RDPG) with latent positions generated by a GGM on $\T$. The previous model ($\T$-Graphical Blockmodel) depends on the population covariance matrix from a GGM, and uses nodes in $V_z$ to define the connectivity matrix $B$; this section presents a model that employs the sample covariance matrix and only takes samples from $V_\ell$, the set of leaf nodes.

In the RDPG, each node is assigned a latent position vector $x_i \in \R^q$.
 For notation simplicity, we re-parameterize the original RDPG model in \citep{rdpg} as a random graph with
$$\lambda_{ij}^{RDPG} = q^{-1} |\langle x_i, x_j \rangle|.$$
For convenience, this scales by the latent dimension $q$. The absolute value ensures that $\lambda_{ij}\ge 0$.

Typically, in the RDPG, latent position vectors $x_1, \dots, x_n \in \R^q$ are considered to be independent and identically distributed.  However, in order to generate hierarchical dependence among the $n$ nodes in the graph, we make these random vectors dependent, and their dependence is specified by a GGM on $\T$. In the $\T$-Graphical RDPG proposed below, the vectors $x_1, \dots, x_n$ are \textit{dependent}, but the elements within each $x_i$ are \textit{independent}.

\begin{definition}
Given a tree $\T = (V, E)$ with $n$ leaf nodes.
Consider a Gaussian distribution $\mathcal{P}_X(\cdot)$ that is a GGM on $\T$, and let $\mathcal{P}_{X_\ell}(\cdot)$ be the marginal distribution of all leaf nodes. Generate $q$ independent length $n$ vectors \[X_1, \cdots, X_q \overset{i.i.d.}{\sim} \mathcal{P}_{X_\ell}(\cdot),\]
and place them into columns of $X\in \R^{n\times q}$. Then the \textbf{$\T$-Graphical RDPG} is an RDPG, with $x_i$ defined as the $i$th row of $X$. 

\end{definition} 

As the number of latent dimensions $q \rightarrow \infty$, this model converges, in a simple sense defined below, to a $\T$-Stochastic Graph. Similarly, for any $\T$-Stochastic Graph, there exists a $\T$-Graphical RDPG that converges to it as $q \rightarrow \infty$. 

\begin{theorem}\label{thm:tsg_tg-rdpg}
Given any $\T$-Graphical RDPG, there exists a $\T$-Stochastic Graph with $\lambda_{ij}\geq 0$ such that for any pair of nodes $i$ and $j$, 
\[\lambda_{ij}^{RDPG} \xrightarrow{q\rightarrow\infty} \lambda_{ij}.\]
And vice versa, for any $\T$-Stochastic Graph, there exist a $\T$-Graphical RDPG such that $\lambda_{ij}^{RDPG} \rightarrow \lambda_{ij}$ as the latent dimension $q\rightarrow \infty$.
\end{theorem}

Given Proposition \ref{prop:cov->dist} and \ref{prop:dist->cov}, the equivalences in Theorems \ref{thm:tsg_tg-rdpg} follow from the law of large numbers; as $q \rightarrow \infty$,  
$$\lambda_{ij}^{RDPG} = q^{-1} |\langle x_i, x_j \rangle| \rightarrow |\Sigma_{i,j}| = \exp(- d(i,j)) = \lambda_{ij}.$$

\subsection{$\T_r$-Overlapping Blockmodel}\label{sec:osb}

In a DCSBM, every node $i$ belongs to exactly one block. In an Overlapping SBM, node $i$ is allowed to belong to multiple blocks. In this way, ``clusters'' overlap with each other and that is where the name ``overlapping'' comes from. To distinguish from the DCSBM, we utilize the symbol $K$ to indicate the number of blocks and utilize a binary vector of length $K$, denoted by $Z_i \in \{0, 1\}^K$, to signify the block membership of node $i$.

\begin{definition}
The \textbf{Overlapping Stochastic Blockmodel} (OSBM) \footnote{The original paper on the overlapping Stochastic Blockmodel is not exactly this definition here because it includes a logistic link function $\pr(A_{ij} = 1) = logit(Z_iBZ_j^T)$, and we further include degree corrected parameters $\theta_i$.} is a random graph with
\[\lambda_{ij}^{OSBM} = \theta_i\theta_j Z_i^T BZ_j\]
for node specific degree parameters $\theta_i \in \R_+$, block assignments $Z_i \in \{0, 1\}^K$, and connectivity matrix $B\in \R_+^{K\times K}$ that is full rank.
\end{definition}

The definition of $\T_r$-Overlapping Blockmodel below uses ``rooted trees'' $\T_r$. The topology structure of $\T_r$ and $\T$ are the same, except that $\T_r$ picks some internal node $r$ as the root. In a rooted tree $\T_r$, we say node $u$ is an \textbf{ancestor} of $i$ if $u$ lies on the shortest path between node $i$ and root $r$, denote the set of ancestor nodes of $i$ as $\text{anc}(i)$. The $\T_r$-Overlapping Blockmodel is an OSBM with a hierarchical structure encoded in the block membership: every leaf node belongs to all of its ancestor blocks.

\begin{definition}
The graph $A$ follows the $\T_r$-\textbf{Overlapping Blockmodel} if $A$ is an OSBM such that 
\vspace{-0.06in}
\begin{enumerate}
\itemsep0em 
    \item nodes in graph $A$ are indexed by leaf nodes in $\T_r$,
    \item blocks in the OSBM are indexed by all internal nodes in $\T_r$,
    \item nodes $i$ belongs to block $u$ if internal node $u$ is an ancestor of leaf node $i$, i.e., \[Z_{iu} = 1 \Longleftrightarrow u\in \text{anc}(i).\]
\end{enumerate}
\end{definition}

The relationship between $\T$-Stochastic Graphs and $\T_r$-Overlapping Blockmodels are presented in Theorem \ref{thm:tsg->tob} and \ref{thm:tob->tsg} below, with proofs contained in Appendix \ref{appendix:tsg->tob} and \ref{appendix:tob->tsg}, accordingly.

\begin{theorem}\label{thm:tsg->tob}
Every $\T$-Stochastic Graph is a $\T_r$-Overlapping Blockmodel, where $r$ is any internal node in $\T$.
\end{theorem}

\begin{theorem}\label{thm:tob->tsg}
Any $\T_r$-Overlapping Blockmodel with diagonal connectivity matrix $B$ is a $\T$-Stochastic Graph, where $\T$ is the unrooted $\T_r$.
\end{theorem}

\subsection{Equivalent Stochastic Processes that generate edges}

This section introduces two Stochastic Processes that help generate $\T$-Stochastic Graphs efficiently. Each process parameterizes an edge generator $\mathcal{P}(\cdot)$ that samples one edge at a time from the set of all possible ${n \choose 2}$ edges.  Both edge generators, if given as input to Algorithm \ref{alg:graph_generator}, will generate graphs that are equivalent to $\T$-Stochastic Graphs with Poisson distributed elements.

\begin{algorithm}
	\caption{Random Graph Generator}
	\label{alg:graph_generator}
	\KwIn{sparsity parameter $\zeta$, edge generator $\mathcal{P}(\cdot)$}
	\KwOut{random graph $A\in \R^{n\times n}$}
 	Initialize $A$ to be a matrix with all zeros. Sample the number of edges $m\sim \text{Poisson}(\zeta)$\;
	\For{$\ell \gets1$ \KwTo $m$}{
 Sample an edge $(I,J)\sim \mathcal{P}(\cdot)$, then add it to the graph $A$ by incrementing $A_{I,J}$ by one
 \;}
\end{algorithm}

The edge generators both use the concept of a \textbf{Markov process}.  
In any Markov process $X_t$, 
we say $a$ is an \textbf{absorbing state} if $\pr(X_{t+1} = a \mid X_{t} = a) = 1$, and denote $\tau$ as the first time that $X_t$ enters an absorbing state. 

\begin{figure}[h]
    \centering
    \includegraphics[width=5in]{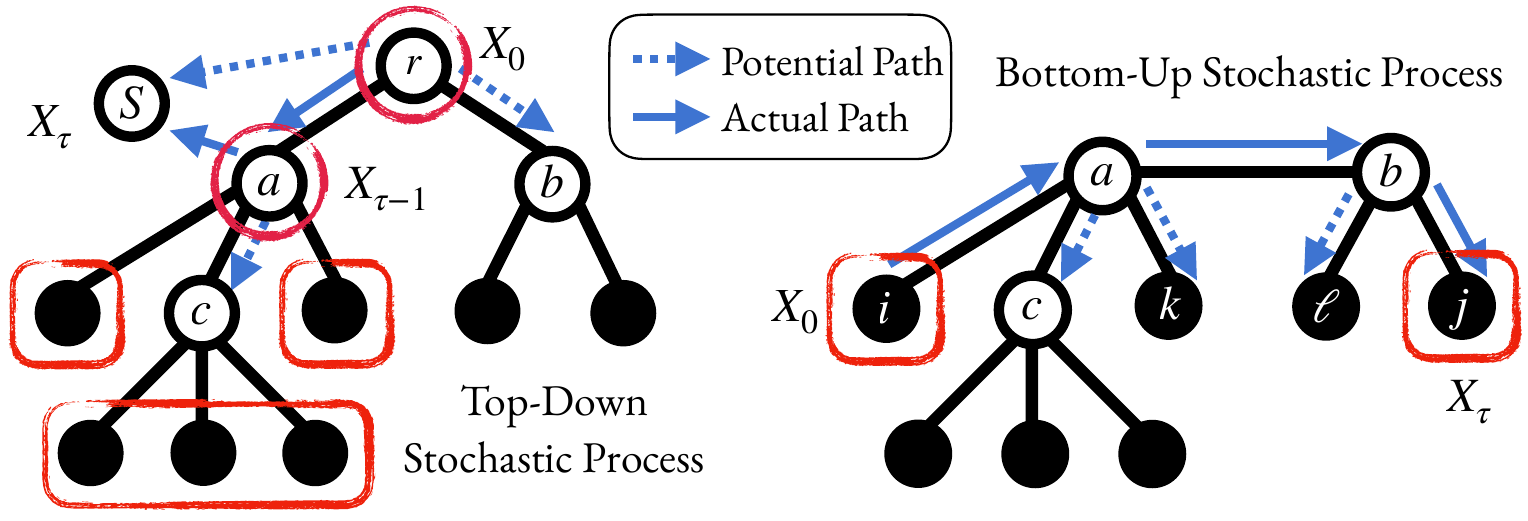}
    \caption{This plot illustrates realizations of the $\T$-Top Down Generator and the $\T$-Bottom Up Generator. 
    For the $\T_r$-Top Down Stochastic Process, $X_t$ starts from the root $r$, and can go to the left child $a$, the right child $b$, or the absorbing state $S$. In this realization, $X_t$ goes to node $a$, and has two options for the subsequent state: the internal node $c$, or the absorbing state $S$. Here, $X_t$ moves to the absorbing state $S$, resulting in $X_{\tau-1} = a$. Then the five leaf nodes marked in the red rectangle, which are descendants of node $a$, have a positive probability of being selected to form an edge. 
    To start the $\T$-Bottom Up Stochastic Process, sample a leaf node $i$ and define $X_0$ as the \textit{directed} edge away from it, $\langle i, a \rangle$. Starting from $\langle i, a \rangle$, $X_t$ can progress to any of the edges originating from node $a$, that is, edge $\langle a, c \rangle$, $\langle a, k \rangle$, and $\langle a, b \rangle$. In this realization, $X_t$ proceeds to edge $\langle a, b \rangle$, and has two options for the next state: the edge leading to leaf node $\ell$ or the edge leading to leaf node $j$. Here, $X_t$ goes to edge $\langle b, j \rangle$. As a result, $X_0 = \langle i, a \rangle, X_{\tau} = \langle b, j \rangle$, and the generator selects leaf nodes $i$ and $j$.}
    \label{fig:tds_bus}
\end{figure}

\subsubsection{$\T_r$-Top Down Stochastic Process}\label{sec:tds}

In a rooted tree $\T_r$, we say node $i$ is a \textbf{descendant} of node $u$ if $u$ is an ancestor of $i$, 
denote the set of all descendants of node $u$ as $\text{desc}(u)$.
We say node $i$ is a \textbf{child} of node $u$ if $u$ is the closest ancestor to $i$, denote the set of children of node $u$ as $\mathcal{C}(u)$, i.e. $\mathcal{C}(u) = \{i: (u, i)\in E,\ u\in \text{anc}(i)\}$. Note that $\mathcal{C}(u)\subseteq \text{desc}(u)$.

The $\T_r$-Top Down Generator samples an edge in $A$ through two steps. 
\begin{enumerate}
\itemsep0em 
    \item \textbf{Select an internal node $u \in V$.} 
A Markov process starts from the root node $r$ and walks along the tree in a top-down fashion. At each internal node, the process can either proceed to one of its non-leaf children or transition to the absorbing state (it cannot go backward). Once it reaches the absorbing state, the preceding state $X_{\tau-1}$ is selected as the internal node $u$ we sample. The generator needs to specify the transition probabilities. 
\item \textbf{Select two leaves that descend from $u$.} The generator selects two leaf nodes among all the descendants of $u$ with probabilities proportional to node-specific degree parameters $\theta_i$.
\end{enumerate} 
An example is provided in Figure \ref{fig:tds_bus}, and the detailed definitions 
can be found in Appendix \ref{appendix:tds_definition}. The proof of Theorem \ref{thm:tsg->tds} is contained in Appendix \ref{appendix:tds_proof}.

\begin{theorem}\label{thm:tsg->tds}
Given a $\T$-Stochastic Graph with $A_{ij}\overset{ind.}{\sim} \text{Poisson}(\lambda_{ij})$, there exists a sparsity parameter $\zeta>0$ and a $\T_r$-Top Down Generator such that they can generate graphs (using Algorithm \ref{alg:graph_generator}) with the same distribution as this $\T$-Stochastic Graph. And vice versa, 
for any $\T_r$-Top Down Generator and sparsity parameter $\zeta>0$, there exists a $\T$-Stochastic Graph with Poisson distributed $A_{ij}$ such that they have the same distribution.
\end{theorem}

\subsubsection{$\T$-Bottom Up Stochastic Process}\label{sec:bus}
This section considers directed graphs and adopts the notation $\langle u, v \rangle$ to indicate a directed edge that goes from node $u$ to node $v$. Given an undirected tree graph $\T = (V, E)$, the directed graph $\T_{dir} = (V, E_{dir})$ can be constructed by adding two directed edges, $\langle u, v \rangle$ and $\langle v, u \rangle$, to $E_{dir}$ for each undirected edge $(u, v) \in E$.

The $\T$-Bottom Up Generator samples an edge in $A$ by performing a non-backtracking\footnote{A non-backtracking random walk is not Markov if considering vertices as the state space but can be turned into a Markov chain by changing the state space from the vertices to the directed edges \citep{kempton2016non}, that's why we consider directed edges as the state space.} random walk $X_t$ on tree $\T$, and then selecting two leaf nodes based on the starting state, $X_0$, and the stopping state, $X_\tau$. This random walk is a Markov Process with its state space consisting of all directed edges in $\T_{dir}$. In particular, this process involves:

\begin{enumerate}
\itemsep0em 
    \item \textbf{Picking a leaf node to start:} To begin the process, a starting state $X_0$ is selected from all edges that originate from leaf nodes, with probabilities \(\{\pi_i : i \in V_\ell\}\).
    \item \textbf{Walking along the tree in a non-backtracking fashion:} Given $X_t = \langle u, v \rangle$, the process can move to any directed edge $\langle v, w \rangle$ that starts from node $v$, with a probability proportional to some edge parameter\footnote{We have $c_{uv} = c_{vu}$ here. Notice that edge parameter $c_{uv}$ is different from edge weight $w_{uv}$.} $c_{vw}>0$. However, 
$X_{t+1}$ cannot be the reverse edge $\langle v, u \rangle$ leading back to $u$, hence the term ``non-backtracking''.
    \item \textbf{Ending at another leaf node:} Any edge that leads to a leaf node is an absorbing state, the process stops once it reaches one of them.
\end{enumerate}
Appendix \ref{appendix:bus_definition} provides detailed definitions 
and an example is illustrated in Figure \ref{fig:tds_bus}.

The following definition of \textbf{symmetry} is helpful to parametrize the equivalence between $\T$-Bottom Up Generators and $\T$-Stochastic Graphs. To have full equivalence, it's also necessary to allow negative edge weights in $\T$-Stochastic Graphs (as in Remark \ref{rmk:negative_edge_weights}). 
Theorem \ref{thm:bus_negative} in Appendix \ref{appendix:negative_edge_bus}) provides an explicit equivalent condition for negative edge weights on specific edges.
 This requirement for negative edge weight also arises in the context of Hierarchical Stochastic Equivalence (Section \ref{sec:hse}), and additional discussion is available in Appendix \ref{appendix: negative_edge_weight}. The proof of Theorem \ref{thm:tsg->bus} and \ref{thm:bus->tsg}
 can be found in Appendix \ref{appendix:tsg->bus} and \ref{appendix:bus->tsg},
 accordingly.

\begin{definition}
    A $\T$-Bottom Up Stochastic Process is \textbf{symmetric} if for any pair of leaf nodes $i, j$, 
    \begin{equation}\label{eq:symmetric}
        \pr(X_0 = \langle i, p(i) \rangle, X_{\tau} = \langle p(j), j \rangle) = \pr(X_0 = \langle j, p(j) \rangle, X_{\tau} = \langle p(i), i \rangle).
    \end{equation}
\end{definition}

\begin{theorem}\label{thm:tsg->bus}
For any $\T$-Stochastic Graph with $A_{ij}\overset{ind.}{\sim} \text{Poisson}(\lambda_{ij})$, there exists a symmetric $\T$-Bottom Up Generator and a sparsity parameter $\zeta$ that can generate graphs (using Algorithm \ref{alg:graph_generator}) with the same distribution. 
\end{theorem}

\begin{theorem}\label{thm:bus->tsg}
Given a set of edge parameters $\{c_{uv}: (u, v)\in E\}$ that defines transition probabilities,
there exists a set of probabilities \(\{\pi_i: i \in V_\ell\}\) such that $\{c_{uv}: (u, v)\in E\}$ and \(\{\pi_i: i \in V_\ell\}\) together define a symmetric $\T$-Bottom Up Generator. Moreover, given this generator and a sparsity parameter $\zeta > 0$ to generate graphs (using Algorithm \ref{alg:graph_generator}), there exists a $\T$-Stochastic Graph, \textbf{with potentially negative edge weights} that has the same distribution as the generated graph. 
\end{theorem}

\subsection{Hierarchical Stochastic Equivalence}\label{sec:hse}

\begin{quote} 
\textit{Stochastic Equivalence} (SE) is the original axiomatic motivation for the popular Stochastic Blockmodel \citep{holland}.  %
This section proposes a notion of \textit{Hierarchical Stochastic Equivalence} (HSE) and provides theorems that show HSE graphs are nearly equivalent to $\T$-Stochastic Graphs.
\end{quote}

Consider a random graph $A$ with $\lambda_{ij} \ge 0$; two nodes $i$ and $j$ are \textbf{\textit{Stochastically Equivalent}} (SE) if for all other nodes $\ell$, $\lambda_{i\ell}/\lambda_{j\ell} = 1$ \citep{holland}. SE is an elementary property that characterizes the Stochastic Blockmodel; in the SBM, two nodes are SE if and only if they are in the same block. This idea naturally extends to the Degree-Corrected Stochastic Blockmodel, two nodes $i$ and $j$ are \textbf{\textit{Degree-Corrected Stochastically Equivalent}} (DC-SE) if $\lambda_{i\ell}/\lambda_{j\ell} = c_{ij}$ for all other nodes $\ell$ and some constant $c_{ij}$ that does not depend on $\ell$.

\textbf{\textit{Hierarchical  Stochastic Equivalence}} (HSE) is a stronger condition than previous SE and DC-SE notions.  In those previous notions, only node pairs in the same ``block'' are equivalent.  In HSE, every pair of nodes satisfies a set of equivalence relationships. The key idea of HSE relies on the concept of ``median node'',  which is illustrated in Figure \ref{fig:hse}.

\begin{definition}\label{def:median}
The \textbf{median} $m(i, j, \ell)$ of any three leaf nodes $i,j,\ell \in \T$ is the unique node in $\T$ that
 lies on the shortest path between each pair of vertices $(i,j), (\ell,i)$, and $(\ell,j)$ in $\T$.
\end{definition}
\vspace{-0.05in}

In Figure \ref{fig:hse}, $\ell_1, \ell_2, \dots  $ are leaves that descend from some subtree connected to $a$. In that illustration, the median of $i,j$, and one of these $\ell \in \{\ell_1, \ell_2, \cdots\}$ is the node marked $a$; $m(i, j, \ell) = a$.

\begin{figure}[htbp] %
   \centering
   \includegraphics[width=2in]{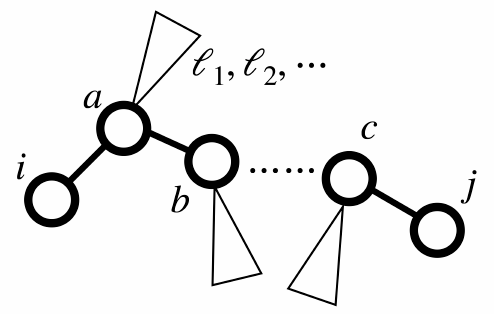} 
   \caption{A tree example that helps illustrate HSE. Triangles connected to nodes represent subtrees. Nodes $a, b, \cdots, c$ are internal nodes along the path between leaf nodes $i, j$. 
}
   \label{fig:hse}
\end{figure}

\vspace{-0.05in}

\begin{definition}\label{def:hse}
A random graph with $\lambda_{ij}\geq 0$ is $\T$-HSE if for any three leaf nodes $i, j, \ell$ in $\T$,
\vspace{-0.05in}
\begin{equation}\label{eq:hse}
    \lambda_{i, \ell}/\lambda_{j, \ell} = c_{ij}^{(m)},
    \vspace{-0.03in}
\end{equation}
where $m = m(i, j, \ell)$ is the median of $i,j,\ell$ in tree $\T$. So, the righthand side only depends on $\ell$ via its median node with $i$ and $j$.
\end{definition}
\vspace{-0.02in}

Take Figure \ref{fig:hse} as an example, in Equation \eqref{eq:hse}, $\lambda_{i\ell}/\lambda_{j\ell}$ stays the same for any node $\ell \in \{\ell_1, \ell_2, \cdots\}$ since they share the same median node $a$.

 To make $\T$-HSE graphs equivalent to $\T$-Stochastic Graphs, it is necessary to allow negative edge weights in $\T$-Stochastic Graphs (as in Remark \ref{rmk:negative_edge_weights}).
 Appendix \ref{appendix:neg_edge_hse} provides more discussion on why negative edge weights are inevitable in HSE and how this relates to the SE condition in SBMs. The proof of Theorem \ref{thm:tse->tsg} and \ref{thm:tsg->tse} can be found in Appendix \ref{appendix:tse->tsg} and \ref{appendix:tsg->tse}, accordingly.

\begin{theorem}\label{thm:tse->tsg}
Suppose $A$ is a random graph model with $\lambda_{ij} > 0$ for all node pairs $i,j$. If there exist a tree $\T$ such that $A$ is $\T$-HSE, then $A$ is a $\T$-Stochastic Graph \textbf{with potentially negative edge weights}.
\end{theorem}

\begin{theorem} \label{thm:tsg->tse}
Any $\T$-Stochastic Graph is $\T$-HSE.
\end{theorem}

\section{Estimating $\T$-Stochastic Graphs using $\T$-Graphical Blockmodel} \label{sec:estimation}

\subsection{\pps: three-step estimation utilizing the $\T$-Graphical Blockmodel}

Every $\T$-Stochastic Graph is a DCSBM where the connectivity matrix $B$ is the covariance matrix for a Gaussian Graphical Model (GGM). This parameterization is called a $\T$-Graphical Blockmodel in Section \ref{sec:tgb} and it provides a natural path for estimation: recover the DCSBM first, and then construct the hierarchy using tools from Graphical Models. 
Using the language of tree graphs, recovering the DCSBM identifies the parent node $p(i) \in V$ for each leaf node $i \in V$. This is because $p(i) = z(i)$ using the notation of Definition \ref{def:dcsbm}. Then the algorithm only needs to construct $\T_z$, a subtree of $\T$, which is the smallest tree in $\T$ that connects the parent nodes $p(i)$. See Figure \ref{fig:tsg_dcsbm_tree} for an example.

To estimate $z(\cdot)$ and $\T_z$, we propose \texttt{synthesis}, 
a three-step mechanism that consistently recovers the hierarchy under mild assumptions (Theorem \ref{thm:meta_theorem} in Section \ref{sec: pps_consistent}). The \texttt{vsp} algorithm in the first step is a consistent community detection algorithm from \citep{varimax_rohe}; more details can be found in Appendix \ref{appendix:vsp}. The \texttt{TSGdist} and \texttt{SparseNJ} algorithms in the following steps are proposed using techniques from GGMs and phylogeny studies, and are explained in detail in the subsequent sections.
To simplify notations, we define the block membership matrix $Z\in \R^{n\times k}$ as $Z_{ij} = \theta_{i}\mathds{1}\{z(i)=j\}$.

\begin{algorithm}
\caption{\pps (meta-algorithm)}
    \label{alg:pps}
	\KwIn{\\ \Indp \Indp 
        adjacency matrix $A\in \R^{n\times n}$; \\
        number of blocks $k$;\\ 
        regularization parameter $\epsilon$; see Remark \ref{rmk:pps_regularization} for the default value;\\ 
        cutoff value $\varphi$; see Remark \ref{rmk:pps_cutoff} for the default value.}
	\KwOut{\\ \Indp \Indp
        subtree $\widehat \T_z$;\\
        block membership matrix $\widehat Z\in \R^{n\times k}$.}
  \vspace{0.05in}
    Community detection in the DCSBM: $\widehat Z \leftarrow \texttt{vsp}\left(A, k\right)$;\\[-0.05in]
    Distance estimation between nodes in $V_z$: $\widehat D \leftarrow \texttt{TSGdist}\left(A, \widehat Z, \epsilon \right)$;\\[-0.05in]
    Subtree reconstruction: $\widehat\T_z \leftarrow \texttt{SparseNJ}\left(\widehat D, \varphi \right)$.
\end{algorithm}

\begin{remark}\label{rmk:pps_post}
    (Optional) post-processing step: the entire tree $\T$ can be reconstructed by attaching twigs on $\widehat \T_z$ based on $\widehat Z$. We output $\widehat \T_z$ because, this level of the hierarchy (relationships among communities instead of single nodes) is easier to visualize for large social networks. For example, Figure \ref{fig:journal_tree} displays $k = 100$ communities as leaf nodes.
\end{remark}

\begin{remark}\label{rmk:pps_regularization}
    (Default) regularization parameter $\epsilon$: compute the sample average degree as $\widehat \Delta = \sum_{i = 1}^n\sum_{j = 1}^n A_{ij}/n$ and set the default regularization parameter to $\epsilon = \gamma\widehat \Delta^{.5}/n$ with $\gamma = 0.01$. Theorem \ref{thm:meta_theorem} gives more details about conditions on $\epsilon$ to guarantee asymptotic consistency.
\end{remark}

\begin{remark}\label{rmk:pps_cutoff}

(Default) cutoff value $\varphi$: there are multiple choices available for $\varphi$, see Section \ref{sec:sparse_nj} and Theorem \ref{thm:meta_theorem} for some examples. Algorithm \ref{alg:pps} use a default cutoff value of 
$\varphi = 2 \max_{uv}\, \left\{\widehat\sigma_{uv}\right\}$,
where $\widehat\sigma_{uv}^2$ is an estimate for the variance of $\widehat D_{uv}$ and is computed using Equation \eqref{eq: var_d_uv}.
\end{remark}

\begin{remark}\label{rmk:pps_k}
Selecting $k$: like many other community detection algorithms, \texttt{vsp} requires the number of blocks as an input. Nonetheless, we observed that the output tree of \pps is robust to the selection of $k$, see Appendix \ref{appendix:robust_k} for more details.  
\end{remark}

\subsection{\texttt{TSGdist}: distance estimation with properties of GGMs}\label{sec:tsg_dist}

\texttt{TSGdist} estimates $D\in \R^{k \times k}$, the matrix of the pairwise distance (in $\T$) between the nodes in $V_z$. This matrix can also be interpreted as the distances between blocks in the DCSBM. The main idea behind this algorithm is derived from the connection between Gaussian Graphical Models and $\T$-Stochastic Graphs. According to Proposition \ref{prop:cov->dist}, 
\begin{equation}\label{eq:neg_log_transform}
    D_{uv} = -\log\left(B_{uv}\right)
\end{equation} for any pair of nodes $u, v\in V_z$. Therefore, the first step is to estimate $B$, followed by constructing a plug-in estimator for $D$. 

To ensure the validity of the log transformation in Equation \eqref{eq:neg_log_transform} and to obtain positive distances, $\widehat B$ must have all diagonal entries equal to one and all non-diagonal entries bounded between $(0, 1]$.  Algorithm \ref{algo:tsg_dist} achieves this by utilizing two matrix clipping operations, denoted as $[\, \cdot\, ]_+$ and $[\, \cdot\, ]_{\leq 1}$. For any matrix $M$, 
\(\left[M_+\right]_{ij} = \max\{0,M_{ij}\}\) and \( \left[M_{\leq 1}\right]_{ij} = \min\{1, M_{ij}\}\).

\begin{algorithm}
\caption{\texttt{TSGdist} (step 2)}
\label{algo:tsg_dist}
	\KwIn{ \\ \Indp \Indp
        adjacency matrix $A\in \R^{n\times n}$;\\ 
        estimated block membership matrix $\widehat Z\in \R^{n\times k}$;\\
        regularization parameter $\epsilon$.}
	\KwOut{\\ \Indp \Indp
        pairwise distances $\widehat D \in \R^{k\times k}$ between nodes in $V_z$.} 
    Nonnegative Transformation: let $\widehat B^{nn} = \frac{1}{n^2}\left[\widehat Z^T A \widehat Z \right]_+$;\\[-0.05in]
	Regularization: set $\widehat B^{nn}_{ij} \leftarrow \widehat B^{nn}_{ij}+\epsilon$;\\[-0.05in]
	Scaling: let $\widehat S$ be a diagonal matrix with $\widehat S_{ii} = \sqrt{\widehat B^{nn}_{ii}}$, let $\widehat B = \left[\widehat S^{-1}\widehat B^{nn}\widehat S^{-1}\right]_{\leq 1}$;\\[-0.05in]
    Negative Log Transformation: output $\widehat D$ with $\widehat D_{uv} = -\log\left(\widehat B_{uv}\right)$.
\end{algorithm}

\begin{remark}\label{rmk:bnn_estimator}
Sometimes a nonnegative transformation of $\widehat Z$ provides a more reasonable interpretation. In that case, we recommend estimating $\widehat B^{nn}$ as $\widehat B^{nn} = \left(\widehat Z_+^T\widehat Z_+\right)^{-1} \widehat Z_+^T A \widehat Z_+ \left(\widehat Z_+^T \widehat Z_+\right)^{-1}$ or $\widehat B^{nn} = \frac{1}{n^2}\widehat Z_+^T A \widehat Z_+$. Theorem \ref{thm:tsgdist} and \ref{thm:meta_theorem} still hold for these two estimators.
\end{remark}

Notably, the column ordering of $Z$ is not identifiable as one could relabel nodes in $V_z$ without changing the tree structure or $\lambda_{ij}$. In the context of DCSBMs recovery, this is the unidentifiability of communities/blocks labels, and it is generally unavoidable \citep{zhao2012consistency, varimax_rohe}. As a result, the row and column ordering of $\widehat B$ and $\widehat D$ are also not identifiable. To address this, we define the set of reordering operations as $\mathscr{P}(k) = \{P \in \R^{k\times k}: P_{ij} \in \{0, 1\}, \ P^TP = PP^T = I_k\}$.

The following theorem shows that, up to row and column reorderings, \texttt{TSGdist} gives consistent estimation for $D$ as long as $\widehat Z$ is a consistent estimator for $Z$ up to column reorderings and scalings. The convergence rate for this $\widehat D$ can be found in Theorem \ref{thm:meta_theorem}. The proof of Theorem \ref{thm:tsgdist} can be found in Appendix \ref{appendix:tsgdist}

\begin{theorem}\label{thm:tsgdist}
Suppose $A$ is generated from a $\T$-Graphical Blockmodel with $\mathscr{A} = \rho_n ZBZ^T$, where $Z$ follows a bounded distribution and $B$ is a fixed matrix. Consider any $\widehat Z$ that is a consistent estimator for $Z$ up to column reorderings and scalings, that is, there exists a series of reordering matrix $P_n\in \mathscr{P}(k)$ and a fixed diagonal scaling matrix $C\in \R^{k\times k}$ such that 
\(\left\|\widehat{Z} - ZCP_n\right\|_{2\rightarrow \infty} = O_p(e_n),\)
where $\|\cdot\|_{2\rightarrow \infty}$ is the two to infinity norm \citep{cape2019two} and $e_n$ is a series of constants converging to zero as $n$ goes to infinity. If \(\left \| A - \mathscr{A} \right\|_2 = O_p\left( n \rho_n e_n \right),\) and the regularization parameter $\epsilon = O_p(e_n\rho_n)$, we have
\[\left\|\widehat{D} - P_n^TDP_n\right\|_{F} = O_p\left(e_n\right),\]
where $\widehat D$ is the output of Algorithm \ref{algo:tsg_dist}.
\end{theorem}

\subsection{\texttt{SparseNJ}: tree estimation with NJ and thresholding}\label{sec:sparse_nj}

To estimate the subtree $\T_z$ in the $\T$-Stochastic Graph, this section proposes \texttt{SparseNJ}, an extension to the popular Neighbor-Joining (NJ) algorithm. 
Previously, \cite{atteson1997performance} showed that NJ is consistent in the following sense:  
Let the tree $\T$ have a pairwise distance matrix $D$ among its leaf nodes and let  $\eta$ be the length of the shortest edge in $\T$.  If we observe $\widehat D$ that satisfies $\max_{u,v} \left|\widehat D_{uv} - D_{uv}\right| < \eta/2$, then NJ with $\widehat D$ will reconstruct a tree $\widehat \T$ that is \textbf{equivalent} to $\T$, denoted as $\widehat \T \sim \T$. 
Importantly, the above result does not require $\T$ to be ultrametric. However, it does presume  two conditions that are violated when we seek to recover $\T_z$ using the estimated pairwise distances between nodes in $V_z$. 
First, it presumes that $\T$ is binary.  Second, it presumes that only leaf nodes appear in $\widehat D$.

This section shows that, the problem of non-binary trees and non-leaf nodes distances can be resolved by applying an edge thresholding step after the traditional NJ algorithm; we refer to this new algorithm as \texttt{SparseNJ} (see Algorithm \ref{alg:sparse_nj} in Appendix \ref{appendix:sparse_nj_nj_algorithm} for details). Edge thresholding is a common post-processing step in empirical analysis (e.g. edge contraction in \citep{choi2011learning}). We contribute to the existing literature by providing a theoretical guarantee for it. Additionally, Theorem \ref{thm:cutoff} provides insight into how to select an appropriate cutoff value. 

Denote $V_o$ as the set of nodes that appear in $\widehat D$. When a tree $\T = (V, E)$ is estimated with $\widehat D$ and NJ, the resulting estimate $\widehat \T = \left(\widehat V, \widehat E\right)$ is always a binary tree with all nodes in $V_o$ labeled as leaf nodes. Therefore, if $\T$ is a non-binary tree or if our initial distance matrix includes some non-leaf nodes, $\widehat \T$ cannot be exactly equivalent to $\T$. However, $\widehat \T$ can still ``partially'' reconstruct $\T$ in the sense that all edges in $E$ are correctly reconstructed by some edges in $\widehat E$. Specifically, \cite{mihaescu2009neighbor} shows that an edge $e$ can be correctly reconstructed as long as $\max_{u,v}\left|\widehat D_{uv}-D_{uv}\right|<d(e)/4$, where $d(e)$ is the length of edge $e$. While this result focuses on the topology structure,
we further prove the existence of a cutoff value in the estimated edge length. Importantly, this cutoff value separates all the correctly reconstructed edges from those that are not. Definitions on the ``equivalence'' between two trees and the ``correctness'' of edge reconstructions follow from \citep{atteson1997performance, bandelt1986reconstructing} and can be found in Appendix \ref{appendix:nj}. The proof of Theorem \ref{thm:cutoff} and Corollary \ref{coro:cutoff} can be found in Appendix \ref{appendix:cutoff} and \ref{appendix:cutoff_coro}, respectively.

\begin{theorem}\label{thm:cutoff}
Given any tree $\T = (V, E)$ and the estimated pairwise distances $\widehat D$ between nodes in $V_o\subseteq V$,
let $\delta = \max_{uv} \left|\widehat D_{uv} - D_{uv}\right|$. If \(\{i\in V : deg(i) < 3\} \subseteq V_o\) and 
\[\min_{e\in E} d(e)> 4\delta, \] 
then NJ correctly reconstructs all edges in $\T$. Moreover, for any $e$ that is correctly reconstructed and any edge $e_0$ that is not, 
\[\widehat d(e_0)\leq 2\delta < \widehat d(e),\]
where $\widehat d(\cdot)$ is the estimated edge length.

\end{theorem}

\begin{corollary}\label{coro:cutoff}
If conditions in Theorem \ref{thm:cutoff} are satisfied, then the \texttt{SparseNJ} algorithm with the cutoff value  $\varphi = 2\delta$ outputs a tree $\widehat \T$ that is equivalent to $\T$.
\end{corollary}

While \texttt{SparseNJ} is simple to implement and has a theoretical guarantee, its success depends on the choice of an appropriate cutoff value, and it may erroneously shrink short edges to zero when $\widehat d(e)$ is not accurately estimated. In empirical data analysis, the main goal is to obtain a clear picture of the underlying structures and understand the relationship between clusters in the network. In this regard, a basic neighbor-joining algorithm can produce trustworthy results since the incorrectly reconstructed edges are short (according to Theorem \ref{thm:cutoff}), and therefore should not significantly affect the visualization result.

If the precise topology structure is of interest, a simple approach is to sort and plot $\widehat d(e)$ for all edges $e\in \widehat E$, and determine the cutoff value based on the gap between consecutive $\widehat d(e)$ values. Alternatively, if $A_{ij}$ is Poisson distributed or follows a sparse Bernoulli distribution, we recommend estimating the variance of $\widehat D_{uv}$ and using it to approximate $\delta$. In particular, let 

\begin{equation}\label{eq: var_d_uv}
   \widehat \sigma^2_{uv} = \mathds{1}\left\{\left[\widehat Z^T A\widehat Z\right]_{uv} >0\right\} \ \left. \widehat S_{uu}^2 \widehat S_{vv}^2  \sum_{i = 1}^n\sum_{j = 1}^n  \widehat Z^2_{iu} \widehat Z^2_{vj}A_{ij}\middle/\widehat B_{uv}^2 \right.,
\end{equation}
where $\widehat S$, $\widehat Z$, and $\widehat B$ are estimations defined in Algorithm \ref{algo:tsg_dist}. Then we could estimate $\widehat  \delta = \max_{uv} \, \left\{\widehat\sigma_{uv}\right\}$ and use $\varphi = 2\widehat \delta$ as the cutoff value. This is the default cutoff value in Algorithm \ref{alg:pps} (see Remark \ref{rmk:pps_cutoff}) and is also used in simulations (see Appendix \ref{sec:simu_multifurcating}).

\subsection{Consistency guarantee of \pps}\label{sec: pps_consistent}

To allow for sparseness in $A$, this section assumes the population adjacency matrix scales with $\rho_n\in \R$, i.e. $\mathscr{A} = \E(A|Z) = \rho_n ZBZ^T$, where $B$ is a fixed covariance matrix (of nodes in $V_z$) and rows in $Z$ follow an 
asymptotically fixed distribution. If $\rho_n \rightarrow 0$, then $A$ contains mostly zeros. In the context of $\T$-Stochastic Graphs, $\rho_n \rightarrow 0$ implies that the distances between nodes in $V_z$ are fixed while the distances from leaf nodes to their parents go to infinity. This section examines the asymptotic performance of \pps with a focus on the output of step 2 and step 3. We refer readers to \citep{varimax_rohe} for the consistency of $\widehat Z$.

The following assumptions are needed to guarantee the consistency of the \pps algorithm. Assumption \ref{assum:z_dist} is a standard distribution assumption for the block membership matrix $Z$ in DCSBMs \citep{varimax_rohe, zhao2012consistency}. In the context of $\T$-Stochastic Graphs, this means that fixing tree $\T_z$, each leaf node independently selects its parent out of $k$ nodes in $V_z$ through a Multinomial distribution. Then, the weight between leaf nodes and their parents, after a negative exponential transformation, follows a bounded positive distribution. Assumption \ref{assum:sparse} controls the sparsity level of $A$. It requires the average expected degree to grow as fast as $O(\log^{11.1} n)$. Assumption \ref{assum:A_dist} governs the tail behavior of $A_{ij}$. It is worth noting that this assumption is more inclusive than sub-Gaussian as both Poisson and Bernoulli distributions satisfy it. We refer readers to \citep{varimax_rohe} for more discussions on this assumption.

\begin{assumption}\label{assum:z_dist}
There exists a probability distribution $\pi$ on $[k]$ and a bounded positive distribution $f(\cdot)$ such that $z(i)\stackrel{i.i.d.}{\sim} Multinomial(\pi)$ and  $\theta_{i}\stackrel{i.i.d.}{\sim} f(\cdot)$.
\end{assumption}

\begin{assumption}\label{assum:sparse}
$\Delta_n \succeq \log^{11.1}n$, where $\Delta_{n} = n\rho_n$.
\end{assumption}

\begin{assumption}\label{assum:A_dist}
Let $\bar{\rho}_n = \max_{i, j} |\mathscr{A}_{ij}|$, then for all $m\geq 2$
\[\E \left[(A_{ij}-\mathscr{A}_{ij})^m\right]\leq (m-1)!\max(\bar{\rho}_n^{m/2}, \bar{\rho}_n),\]
where this expectation is conditional on $Z$.
\end{assumption}

Since node labels in $V_z$ are not identifiable (as discussed in Section \ref{sec:tsg_dist}), we examine the consistency of $\widehat \T_z$ up to node relabeling.
Specifically, consider tree $\T_z$ with distance matrix $D\in \R^{k\times k}$, and let $P\in \mathscr{P}(k)$ be a reordering matrix. We define $P(\T_z)$ as the tree obtained by relabelling nodes in $V_z$ such that the distance matrix is reordered as $P^TDP$. 
For a set of random variables $X_n$ and constants $a_n$, notation $X_n = \omega_p(a_n)$ means that $\lim_{n\rightarrow \infty}\pr\left(|X_n/a_n| \le \epsilon \right)=0$ for any $\epsilon>0$. The proof of Theorem \ref{thm:meta_theorem} can be found in Appendix \ref{appendix:meta_theorem}.

\begin{theorem}\label{thm:meta_theorem}
Suppose $A$ is generated from a $\T$-Graphical Blockmodel with $\mathscr{A} = \rho_n ZBZ^T$. Under assumption \ref{assum:z_dist},\ref{assum:sparse},\ref{assum:A_dist}, for any regularization parameter $\epsilon = O_p(\Delta_n^{.76} n^{-1}\log^{2.75})$ and cut off value $\varphi$ that converges to zero with $\varphi = \omega_p \left(\Delta^{-{.24}}\log^{2.75} n\right)$
, there exists a series of reordering matrix $P_n\in \mathscr{P}(k)$ such that \[\left\|\widehat{D} - P_n^TDP_n\right\|_{F} = O_p\left(\Delta_n^{-.24}\log^{2.75} n\right),\]
\[\pr\left(\widehat \T_z \sim P_n\left(\T_z\right)\right)\xrightarrow{n\rightarrow \infty} 1, \quad \pr\left(\widehat \T \sim \T\right)\xrightarrow{n\rightarrow \infty} 1,\]
where $\widehat D$ is the estimate produced by step 2 in algorithm \ref{alg:pps}, $\|\cdot\|_F$ is the Frobenius norm, $\widehat \T_z$ is the output of algorithm \ref{alg:pps}, and $\widehat \T$ is the output of the optional post-processing step in Remark \ref{rmk:pps_post}.
\end{theorem}

\section{Wikipedia hyperlink network: recovered hierarchy reveals a world map}\label{sec:real_data}

This section studies the English Wikipedia hyperlink network from \citep{snapnets}. This network contains over 4 million Wikipedia pages and approximately one hundred million hyperlinks. Each page is a vertex in the network and $A_{ij} = 1$ indicates a hyperlink from page $i$ to page $j$. For simplicity, we discard pages without any English letters in the title and take the 3-incore component, i.e., the largest subgraph where every page receives links from at least 3 other pages. After this preprocessing, there are $2,464,429$ pages and $76,229,780$ links left.

Due to its asymmetricity, $A$ holds information about both the linking pattern and the being-linked pattern in this network. We investigate both of them by applying \pps to $AA^T$ and $A^TA$, respectively.
More details about how to deal with asymmetric adjacency matrices and interpret the results can be found in Appendix \ref{appendix:asymmetric}. Both trees are constructed with $k = 50$ in the first step, \(\widehat B^{nn} = \widehat Z_+^TA\widehat Z_+\) in the second step, and the NJ algorithm in the third step. Each leaf node is named by the title of the page with the largest element in the corresponding column of $\widehat Z$. Additional discussions about leaf node labeling can be found in Appendix \ref{appendix:bff}.

\begin{figure}[h]
    \centering
	\includegraphics[width=4in]{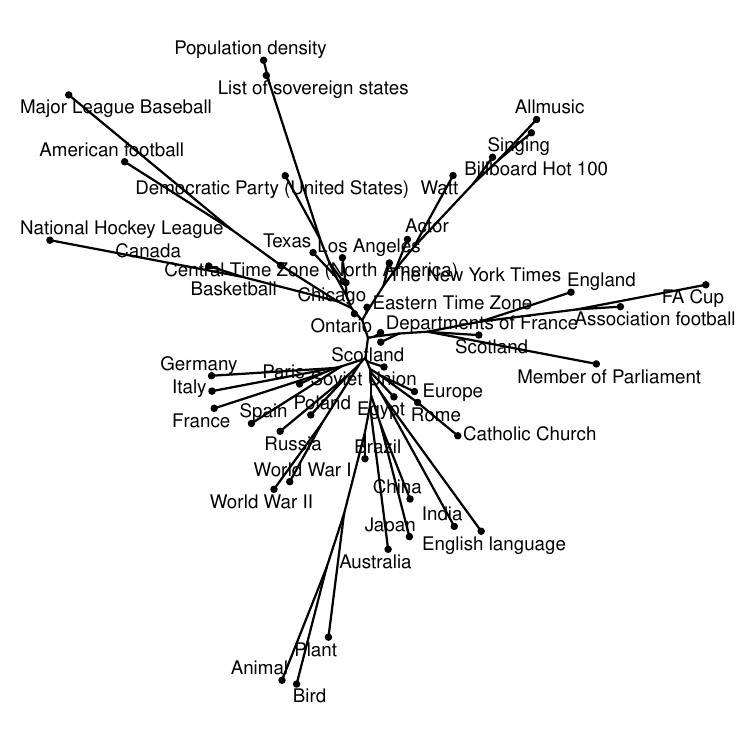}
 \vspace{-0.3in}
   \caption{Tree recovered from the being linked pattern of Wikipedia pages. On the bottom of the tree is a cluster of animals and plants. As we go up, there are clusters of Asian countries (China, Japan); European countries (Germany, Italy, France), the United Kingdom has its own cluster on the right (England, FA Cup); the top half of the tree is about the United States and Canada (North American countries), with metropolis (Chicago, Ontario) near the center of the tree.}
   \label{fig:wiki_beinglinked}
\end{figure}

Figure \ref{fig:wiki_beinglinked} displays the hierarchy structure recovered from the being linked pattern. From bottom to top, there are clusters of animals and plants, Asian countries, European countries, and North American countries. The majority of the tree is occupied by the clusters of European and North American countries, which is expected given that the dataset is derived from English Wikipedia pages. Apart from countries, the tree also exhibits representative cultures and events. For example, on the right side, there are ``Association football'' and ``FA Cup'' near the United Kingdom cluster; on the top, we can observe ``American football'', ``Major League Baseball'', ``Billboard Hot 100'', etc, which are all mainstream sports and entertainment events in American. In summary, this tree serves as a Wikipedia world map that reveals the relationship between countries and cultures across the world.
Unlike the being linked tree, where all leaf nodes are labeled with meaningful page names, the linking pattern tree (Figure \ref{fig:wiki_Z} in Appendix \ref{appendix:wiki_Z}) labels numerous nodes as ``List of something''. This is because pages titled ``List of something''  typically have a high number of outgoing links and are therefore considered to be important hubs in their cluster. Despite the discrepancy in node labeling, comparable clustering patterns can be observed in the linking tree.

\subsection{An informal diagnosis on tree topology}

Even if a data source has no hierarchical structure, \pps will still output a tree.  So, it is important to ask how well the \pps tree represents the data. 
This section provides an informal check for whether tree topology \textit{fits} the dataset. Specifically, we compare $\widehat D$ estimated by \texttt{TSGdist} in the second step with the distances re-estimated by NJ. Notice that the former one can be any matrix with positive elements 
and the latter one must be an additive distance on the tree. The diagnostic results for the Wikipedia hyperlink network are presented in Figure \ref{fig:diagnosis}, where the two matrices are quite similar. The second matrix appears to be a smoothed version of the first one, suggesting that a tree structure fits the data well, with idiosyncratic disturbances.

\begin{figure}[h]
    \centering
    \textbf{Comparing $\widehat D$ estimated by the \texttt{TSGdist} and the \texttt{NJ} algorithm performs an informal diagnosis of whether a $\T$-Stochastic Graph captures the latent structure well.}\\
    \vspace{0.05in}
	\includegraphics[width=3.5in]{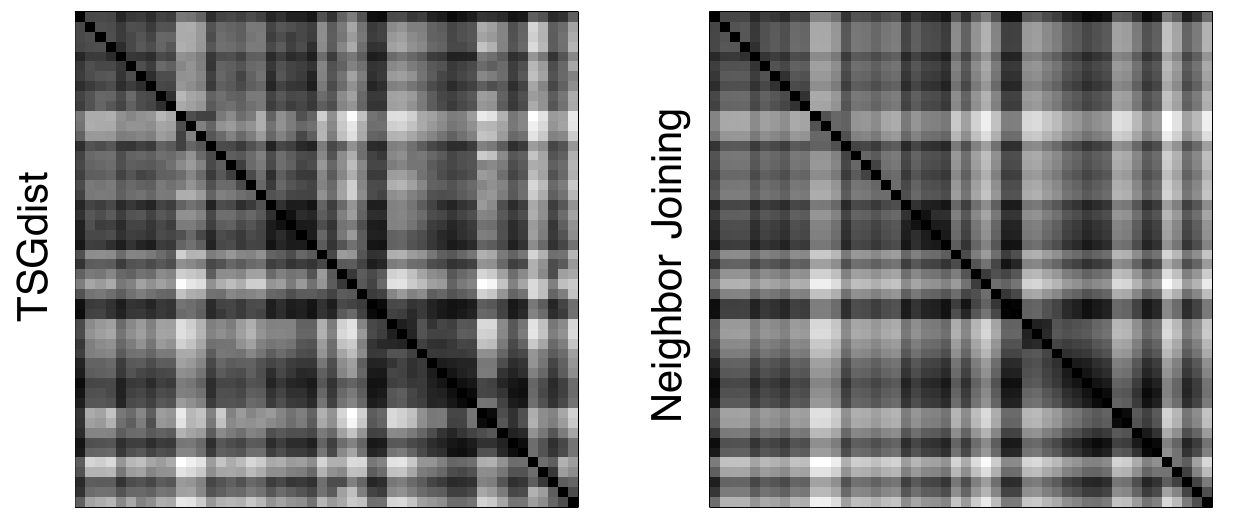}
 \vspace{-0.1in}
    \caption{Diagonsis plot of the being linked pattern in the Wikipedia hyperlink network, with rows and columns arranged following the leaf nodes ordering in the estimated tree (in a circular fashion).} 
   \label{fig:diagnosis}
\end{figure}

\section{Discussion}

Dimension reduction is a fundamental and crucial technique in modern multivariate analysis. However, the estimated dimensions may be hard to interpret especially when there are hundreds of them, which is usually true in large real datasets. Statistical tools are needed to make sense of these dimensions. While PCA provides an ordering of components based on eigenvalues or variance explained, there is no similar technique for factor analysis. The same issue holds for latent structure recovery problems in network analysis, where large networks are assumed to have low-rank structures, and techniques like spectral clustering are used to identify those dimensions. After latent dimensions or factors are estimated, the relationship between them is left unidentified. In many empirical settings, these dimensions appear to have some hierarchical structure.

In this paper, 
we propose a novel parameterization of latent hierarchies in social networks through an exponential transformation of the pairwise negative distance between leaf nodes in a tree graph. Despite its simplicity, this model can generalize to many previous common model classes and hierarchical models. Six alternative models are also proposed and discussed, they understand hierarchies from different angles, yet, are all equivalent to this simple model. Based on those equivalences, a three-step bottom up recovery procedure is proposed: the first step identifies clusters in the network, the second step serves as an aggregation tool that reduces error caused by instabilities of network observations; the third step recovers the underlying structure from this aggregation. We provide consistency guarantees for this three-step recovery procedure. Empirically, it reveals meaningful structure in large scale networks even with $k=100$ communities, a number far larger than one might be willing to interpret without a hierarchical interpretation.

\paragraph{Acknowledgements:}
Thank you to Cécile Ané, Sebastien Roch, and Bret Larget for teaching us phylogeny. Thank you to Keith Levin, Joshua Cape, Tiago P. Peixoto, Vince Lyzinski, and Richard Nerland for inspiring discussions regarding social networks. Thank you to Muzhe Zeng, Fan Chen, Alex Hayes, Auden Krauska, and Jitian Zhao for helpful discussions. This research is supported in part by NSF Grants DMS-1916378 and DMS-2023239 (TRIPODS).

\clearpage

\bibliographystyle{abbrvnat}
\bibliography{references}

\begin{thebibliography}{55}
\providecommand{\natexlab}[1]{#1}
\providecommand{\url}[1]{\texttt{#1}}
\expandafter\ifx\csname urlstyle\endcsname\relax
  \providecommand{\doi}[1]{doi: #1}\else
  \providecommand{\doi}{doi: \begingroup \urlstyle{rm}\Url}\fi

\bibitem[Aizenbud et~al.(2021)Aizenbud, Jaffe, Wang, Hu, Amsel, Nadler, Chang,
  and Kluger]{aizenbud2021spectral}
Y.~Aizenbud, A.~Jaffe, M.~Wang, A.~Hu, N.~Amsel, B.~Nadler, J.~T. Chang, and
  Y.~Kluger.
\newblock Spectral top-down recovery of latent tree models.
\newblock \emph{arXiv preprint arXiv:2102.13276}, 2021.

\bibitem[Ameijeiras-Alonso et~al.(2021)Ameijeiras-Alonso, Crujeiras, and
  Rodriguez-Casal]{JSSv097i09}
J.~Ameijeiras-Alonso, R.~M. Crujeiras, and A.~Rodriguez-Casal.
\newblock multimode: An r package for mode assessment.
\newblock \emph{Journal of Statistical Software}, 97\penalty0 (9):\penalty0
  1–32, 2021.
\newblock \doi{10.18637/jss.v097.i09}.
\newblock URL
  \url{https://www.jstatsoft.org/index.php/jss/article/view/v097i09}.

\bibitem[Ammar et~al.(2018)Ammar, Groeneveld, Bhagavatula, Beltagy, Crawford,
  Downey, Dunkelberger, Elgohary, Feldman, Ha, Kinney, Kohlmeier, Lo, Murray,
  Ooi, Peters, Power, Skjonsberg, Wang, Wilhelm, Yuan, van Zuylen, and
  Etzioni]{ammar-etal-2018-construction}
W.~Ammar, D.~Groeneveld, C.~Bhagavatula, I.~Beltagy, M.~Crawford, D.~Downey,
  J.~Dunkelberger, A.~Elgohary, S.~Feldman, V.~Ha, R.~Kinney, S.~Kohlmeier,
  K.~Lo, T.~Murray, H.-H. Ooi, M.~Peters, J.~Power, S.~Skjonsberg, L.~Wang,
  C.~Wilhelm, Z.~Yuan, M.~van Zuylen, and O.~Etzioni.
\newblock Construction of the literature graph in semantic scholar.
\newblock In \emph{Proceedings of the 2018 Conference of the North {A}merican
  Chapter of the Association for Computational Linguistics: Human Language
  Technologies, Volume 3 (Industry Papers)}, pages 84--91, New Orleans -
  Louisiana, June 2018. Association for Computational Linguistics.
\newblock \doi{10.18653/v1/N18-3011}.
\newblock URL \url{https://aclanthology.org/N18-3011}.

\bibitem[Atteson(1997)]{atteson1997performance}
K.~Atteson.
\newblock The performance of neighbor-joining algorithms of phylogeny
  reconstruction.
\newblock In \emph{International Computing and Combinatorics Conference}, pages
  101--110. Springer, 1997.

\bibitem[Bandelt and Dress(1986)]{bandelt1986reconstructing}
H.-J. Bandelt and A.~Dress.
\newblock Reconstructing the shape of a tree from observed dissimilarity data.
\newblock \emph{Advances in applied mathematics}, 7\penalty0 (3):\penalty0
  309--343, 1986.

\bibitem[Bartlett(1951)]{bartlett1951inverse}
M.~S. Bartlett.
\newblock An inverse matrix adjustment arising in discriminant analysis.
\newblock \emph{The Annals of Mathematical Statistics}, 22\penalty0
  (1):\penalty0 107--111, 1951.

\bibitem[Cape et~al.(2019)Cape, Tang, and Priebe]{cape2019two}
J.~Cape, M.~Tang, and C.~E. Priebe.
\newblock The two-to-infinity norm and singular subspace geometry with
  applications to high-dimensional statistics.
\newblock \emph{The Annals of Statistics}, 47\penalty0 (5):\penalty0
  2405--2439, 2019.

\bibitem[Cattell(1966)]{cattell1966scree}
R.~B. Cattell.
\newblock The scree test for the number of factors.
\newblock \emph{Multivariate behavioral research}, 1\penalty0 (2):\penalty0
  245--276, 1966.

\bibitem[Chen et~al.(2021)Chen, Roch, Rohe, and Yu]{fanpvalues}
F.~Chen, S.~Roch, K.~Rohe, and S.~Yu.
\newblock Estimating graph dimension with cross-validated eigenvalues.
\newblock \emph{arXiv preprint arXiv:2108.03336}, 2021.

\bibitem[Choi et~al.(2011)Choi, Tan, Anandkumar, and Willsky]{choi2011learning}
M.~J. Choi, V.~Y. Tan, A.~Anandkumar, and A.~S. Willsky.
\newblock Learning latent tree graphical models.
\newblock \emph{Journal of Machine Learning Research}, 12:\penalty0 1771--1812,
  2011.

\bibitem[Chung and Lu(2002)]{chung2002connected}
F.~Chung and L.~Lu.
\newblock Connected components in random graphs with given expected degree
  sequences.
\newblock \emph{Annals of combinatorics}, 6\penalty0 (2):\penalty0 125--145,
  2002.

\bibitem[Clauset et~al.(2008)Clauset, Moore, and
  Newman]{clauset2008hierarchical}
A.~Clauset, C.~Moore, and M.~E. Newman.
\newblock Hierarchical structure and the prediction of missing links in
  networks.
\newblock \emph{Nature}, 453\penalty0 (7191):\penalty0 98--101, 2008.

\bibitem[Dasgupta et~al.(2006)Dasgupta, Hopcroft, Kannan, and
  Mitra]{dasgupta2006spectral}
A.~Dasgupta, J.~Hopcroft, R.~Kannan, and P.~Mitra.
\newblock Spectral clustering by recursive partitioning.
\newblock In \emph{European Symposium on Algorithms}, pages 256--267. Springer,
  2006.

\bibitem[El~Ghaoui(2002)]{el2002inversion}
L.~El~Ghaoui.
\newblock Inversion error, condition number, and approximate inverses of
  uncertain matrices.
\newblock \emph{Linear algebra and its applications}, 343:\penalty0 171--193,
  2002.

\bibitem[Erd{\H{o}}s et~al.(1999)Erd{\H{o}}s, Steel, Sz{\'e}kely, and
  Warnow]{erdHos1999few}
P.~L. Erd{\H{o}}s, M.~A. Steel, L.~A. Sz{\'e}kely, and T.~J. Warnow.
\newblock A few logs suffice to build (almost) all trees (i).
\newblock \emph{Random Structures \& Algorithms}, 14\penalty0 (2):\penalty0
  153--184, 1999.

\bibitem[Felsenstein(1984)]{felsenstein1984distance}
J.~Felsenstein.
\newblock Distance methods for inferring phylogenies: a justification.
\newblock \emph{Evolution}, pages 16--24, 1984.

\bibitem[Gillespie(1984)]{gillespie1984molecular}
J.~H. Gillespie.
\newblock The molecular clock may be an episodic clock.
\newblock \emph{Proceedings of the National Academy of Sciences}, 81\penalty0
  (24):\penalty0 8009--8013, 1984.

\bibitem[Holland et~al.(1983)Holland, Laskey, and Leinhardt]{holland}
P.~Holland, K.~Laskey, and S.~Leinhardt.
\newblock Stochastic blockmodels: First steps.
\newblock \emph{Social Networks}, 5\penalty0 (2):\penalty0 109--137, 1983.

\bibitem[Karrer and Newman(2011)]{karrer2011stochastic}
B.~Karrer and M.~Newman.
\newblock Stochastic blockmodels and community structure in networks.
\newblock \emph{Physical Review E}, 83\penalty0 (1):\penalty0 016107, 2011.

\bibitem[Kempton(2016)]{kempton2016non}
M.~Kempton.
\newblock Non-backtracking random walks and a weighted ihara's theorem.
\newblock \emph{arXiv preprint arXiv:1603.05553}, 2016.

\bibitem[Kumar(2005)]{kumar2005molecular}
S.~Kumar.
\newblock Molecular clocks: four decades of evolution.
\newblock \emph{Nature Reviews Genetics}, 6\penalty0 (8):\penalty0 654--662,
  2005.

\bibitem[Lancichinetti et~al.(2009)Lancichinetti, Fortunato, and
  Kert{\'e}sz]{lancichinetti2009detecting}
A.~Lancichinetti, S.~Fortunato, and J.~Kert{\'e}sz.
\newblock Detecting the overlapping and hierarchical community structure in
  complex networks.
\newblock \emph{New journal of physics}, 11\penalty0 (3):\penalty0 033015,
  2009.

\bibitem[Latouche et~al.(2011)Latouche, Birmel{\'e}, and
  Ambroise]{latouche2011overlapping}
P.~Latouche, E.~Birmel{\'e}, and C.~Ambroise.
\newblock Overlapping stochastic block models with application to the french
  political blogosphere.
\newblock 2011.

\bibitem[Lauritzen(1996)]{lauritzen1996graphical}
S.~L. Lauritzen.
\newblock \emph{Graphical models}, volume~17.
\newblock Clarendon Press, 1996.

\bibitem[Lei et~al.(2020)Lei, Li, and Lou]{lei}
L.~Lei, X.~Li, and X.~Lou.
\newblock Consistency of spectral clustering on hierarchical stochastic block
  models.
\newblock \emph{arXiv preprint arXiv:2004.14531}, 2020.

\bibitem[Leskovec and Krevl(2014)]{snapnets}
J.~Leskovec and A.~Krevl.
\newblock {SNAP Datasets}: {Stanford} large network dataset collection.
\newblock \url{http://snap.stanford.edu/data}, June 2014.

\bibitem[Leskovec and Mcauley(2012)]{leskovec2012learning}
J.~Leskovec and J.~Mcauley.
\newblock Learning to discover social circles in ego networks.
\newblock \emph{Advances in neural information processing systems}, 25, 2012.

\bibitem[Leskovec et~al.(2009)Leskovec, Lang, Dasgupta, and
  Mahoney]{leskovec2009community}
J.~Leskovec, K.~J. Lang, A.~Dasgupta, and M.~W. Mahoney.
\newblock Community structure in large networks: Natural cluster sizes and the
  absence of large well-defined clusters.
\newblock \emph{Internet Mathematics}, 6\penalty0 (1):\penalty0 29--123, 2009.

\bibitem[Leskovec et~al.(2010)Leskovec, Huttenlocher, and
  Kleinberg]{leskovec2010signed}
J.~Leskovec, D.~Huttenlocher, and J.~Kleinberg.
\newblock Signed networks in social media.
\newblock In \emph{Proceedings of the SIGCHI conference on human factors in
  computing systems}, pages 1361--1370, 2010.

\bibitem[Li et~al.(2020)Li, Lei, Bhattacharyya, Van~den Berge, Sarkar, Bickel,
  and Levina]{li}
T.~Li, L.~Lei, S.~Bhattacharyya, K.~Van~den Berge, P.~Sarkar, P.~J. Bickel, and
  E.~Levina.
\newblock Hierarchical community detection by recursive partitioning.
\newblock \emph{Journal of the American Statistical Association}, pages 1--18,
  2020.

\bibitem[Lyzinski et~al.(2016)Lyzinski, Tang, Athreya, Park, and Priebe]{hrdpg}
V.~Lyzinski, M.~Tang, A.~Athreya, Y.~Park, and C.~E. Priebe.
\newblock Community detection and classification in hierarchical stochastic
  blockmodels.
\newblock \emph{IEEE Transactions on Network Science and Engineering},
  4\penalty0 (1):\penalty0 13--26, 2016.

\bibitem[Margoliash(1963)]{margoliash1963primary}
E.~Margoliash.
\newblock Primary structure and evolution of cytochrome c.
\newblock \emph{Proceedings of the National Academy of Sciences of the United
  States of America}, 50\penalty0 (4):\penalty0 672, 1963.

\bibitem[Mihaescu et~al.(2009)Mihaescu, Levy, and
  Pachter]{mihaescu2009neighbor}
R.~Mihaescu, D.~Levy, and L.~Pachter.
\newblock Why neighbor-joining works.
\newblock \emph{Algorithmica}, 54\penalty0 (1):\penalty0 1--24, 2009.

\bibitem[Paradis(2012)]{paradis2012analysis}
E.~Paradis.
\newblock \emph{Analysis of Phylogenetics and Evolution with R}, volume~2.
\newblock Springer, 2012.

\bibitem[Pearl(1988)]{pearl1988probabilistic}
J.~Pearl.
\newblock \emph{Probabilistic reasoning in intelligent systems: networks of
  plausible inference}.
\newblock Morgan kaufmann, 1988.

\bibitem[Peixoto(2014)]{peixoto2014hierarchical}
T.~P. Peixoto.
\newblock Hierarchical block structures and high-resolution model selection in
  large networks.
\newblock \emph{Physical Review X}, 4\penalty0 (1):\penalty0 011047, 2014.

\bibitem[Pulquerio and Nichols(2007)]{pulquerio2007dates}
M.~J. Pulquerio and R.~A. Nichols.
\newblock Dates from the molecular clock: how wrong can we be?
\newblock \emph{Trends in Ecology \& Evolution}, 22\penalty0 (4):\penalty0
  180--184, 2007.

\bibitem[Ravasz et~al.(2002)Ravasz, Somera, Mongru, Oltvai, and
  Barab{\'a}si]{ravasz2002hierarchical}
E.~Ravasz, A.~L. Somera, D.~A. Mongru, Z.~N. Oltvai, and A.-L. Barab{\'a}si.
\newblock Hierarchical organization of modularity in metabolic networks.
\newblock \emph{science}, 297\penalty0 (5586):\penalty0 1551--1555, 2002.

\bibitem[Richardson et~al.(2003)Richardson, Agrawal, and
  Domingos]{richardson2003trust}
M.~Richardson, R.~Agrawal, and P.~Domingos.
\newblock Trust management for the semantic web.
\newblock In \emph{The Semantic Web-ISWC 2003}, pages 351--368. Springer, 2003.

\bibitem[Robinson and Foulds(1981)]{robinson1981comparison}
D.~F. Robinson and L.~R. Foulds.
\newblock Comparison of phylogenetic trees.
\newblock \emph{Mathematical biosciences}, 53\penalty0 (1-2):\penalty0
  131--147, 1981.

\bibitem[Rohe and Zeng(2022)]{varimax_rohe}
K.~Rohe and M.~Zeng.
\newblock Vintage factor analysis with varimax performs statistical inference.
\newblock \emph{Accepted as a discussion paper in JRSS-B; arXiv preprint
  arXiv:2004.05387}, 2022.

\bibitem[Rohe et~al.(2018)Rohe, Tao, Han, and Binkiewicz]{rohe2018note}
K.~Rohe, J.~Tao, X.~Han, and N.~Binkiewicz.
\newblock A note on quickly sampling a sparse matrix with low rank expectation.
\newblock \emph{The Journal of Machine Learning Research}, 19\penalty0
  (1):\penalty0 3040--3052, 2018.

\bibitem[Rozemberczki and Sarkar(2020)]{feather}
B.~Rozemberczki and R.~Sarkar.
\newblock {Characteristic Functions on Graphs: Birds of a Feather, from
  Statistical Descriptors to Parametric Models}.
\newblock In \emph{Proceedings of the 29th ACM International Conference on
  Information and Knowledge Management (CIKM '20)}, page 1325–1334. ACM,
  2020.

\bibitem[Rozemberczki et~al.(2019)Rozemberczki, Allen, and
  Sarkar]{rozemberczki2019multiscale}
B.~Rozemberczki, C.~Allen, and R.~Sarkar.
\newblock Multi-scale attributed node embedding, 2019.

\bibitem[Saitou and Nei(1987)]{saitou1987neighbor}
N.~Saitou and M.~Nei.
\newblock The neighbor-joining method: a new method for reconstructing
  phylogenetic trees.
\newblock \emph{Molecular biology and evolution}, 4\penalty0 (4):\penalty0
  406--425, 1987.

\bibitem[Shen et~al.(2009)Shen, Cheng, Cai, and Hu]{SHEN20091706}
H.~Shen, X.~Cheng, K.~Cai, and M.-B. Hu.
\newblock Detect overlapping and hierarchical community structure in networks.
\newblock \emph{Physica A: Statistical Mechanics and its Applications},
  388\penalty0 (8):\penalty0 1706--1712, 2009.
\newblock ISSN 0378-4371.
\newblock \doi{https://doi.org/10.1016/j.physa.2008.12.021}.
\newblock URL
  \url{https://www.sciencedirect.com/science/article/pii/S0378437108010376}.

\bibitem[Silverman(1981)]{silverman1981using}
B.~W. Silverman.
\newblock Using kernel density estimates to investigate multimodality.
\newblock \emph{Journal of the Royal Statistical Society: Series B
  (Methodological)}, 43\penalty0 (1):\penalty0 97--99, 1981.

\bibitem[Sokal(1958)]{sokal1958statistical}
R.~R. Sokal.
\newblock A statistical method for evaluating systematic relationships.
\newblock \emph{Univ. Kansas, Sci. Bull.}, 38:\penalty0 1409--1438, 1958.

\bibitem[Tantrum et~al.(2003)Tantrum, Murua, and
  Stuetzle]{tantrum2003assessment}
J.~Tantrum, A.~Murua, and W.~Stuetzle.
\newblock Assessment and pruning of hierarchical model based clustering.
\newblock In \emph{Proceedings of the ninth ACM SIGKDD international conference
  on Knowledge discovery and data mining}, pages 197--205, 2003.

\bibitem[Tropp(2012)]{tropp2012user}
J.~A. Tropp.
\newblock User-friendly tail bounds for sums of random matrices.
\newblock \emph{Foundations of Computational Mathematics}, 12\penalty0
  (4):\penalty0 389--434, 2012.

\bibitem[Wang and Rohe(2016)]{10.1214/16-AOAS977}
S.~Wang and K.~Rohe.
\newblock {Discussion of ``Coauthorship and citation networks for
  statisticians''}.
\newblock \emph{The Annals of Applied Statistics}, 10\penalty0 (4):\penalty0
  1820 -- 1826, 2016.
\newblock \doi{10.1214/16-AOAS977}.
\newblock URL \url{https://doi.org/10.1214/16-AOAS977}.

\bibitem[Young and Scheinerman(2007)]{rdpg}
S.~Young and E.~Scheinerman.
\newblock Random dot product graph models for social networks.
\newblock \emph{Algorithms and models for the web-graph}, pages 138--149, 2007.

\bibitem[Zhang et~al.(2022)Zhang, Chen, and Rohe]{zhang2018attention}
Y.~Zhang, F.~Chen, and K.~Rohe.
\newblock Social media public opinion as flocks in a murmuration:
  Conceptualizing and measuring opinion expression on social media.
\newblock \emph{Accepted in Journal of Computer Mediated Communication}, 2022.

\bibitem[Zhao et~al.(2012)Zhao, Levina, Zhu, et~al.]{zhao2012consistency}
Y.~Zhao, E.~Levina, J.~Zhu, et~al.
\newblock Consistency of community detection in networks under degree-corrected
  stochastic block models.
\newblock \emph{The Annals of Statistics}, 40\penalty0 (4):\penalty0
  2266--2292, 2012.

\bibitem[Zuckerkandl(1962)]{zuckerkandl1962molecular}
E.~Zuckerkandl.
\newblock Molecular disease, evolution, and genetic heterogeneity.
\newblock \emph{Horizons in biochemistry}, pages 189--225, 1962.

\end{thebibliography}

\newpage
\appendix

\section{Parametric bootstrap for large networks}\label{appendix:parametric_bootstrap}

This section examines the parametric bootstrap of Li's model and our $\T$-Stochastic Graph. 
Specifically, observing the adjacency matrix of a graph, we can fit two different models.  First, we use the top-down bi-partition algorithm to estimate Li's model. Then, we use the \pps algorithm defined in Section \ref{sec:estimation} to estimate the more general $\T$-Stochastic Graph model.  
With these two models, we can simulate artificial networks and then calculate their splitting vectors. 

\begin{figure}[!ht] %
   \centering

\textbf{Diagnostic plots of ``parametric bootstraps'' from a non-ultrametric model (left) and an ultrametric model (right). Though the top-down bi-partition algorithm allows $\widehat \T_r$ to be imbalanced, the splitting vector still does not align with the real dataset.}
\vspace{0.05in}
   \includegraphics[width=6in]{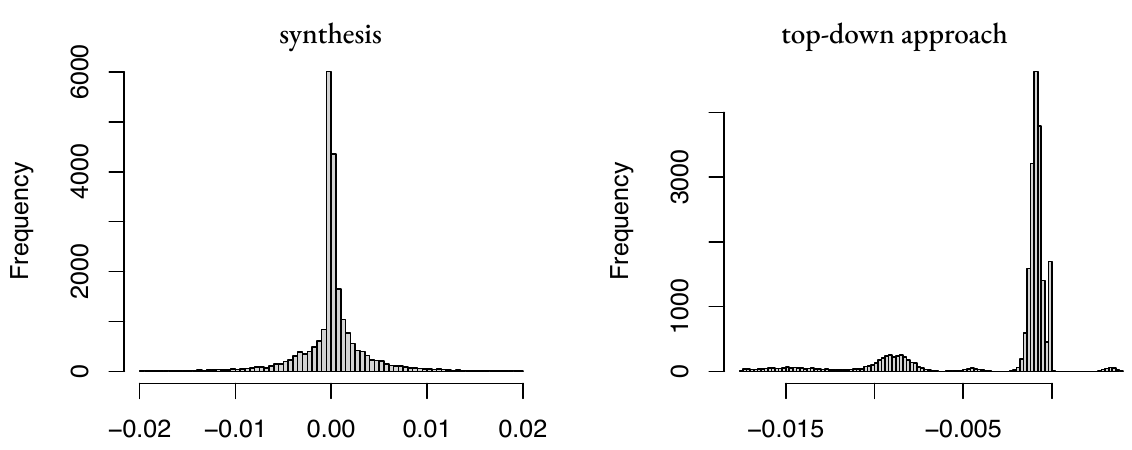} 
   \caption{These diagnostic plots are generated from networks that were simulated from two different statistical models.  Both statistical models were estimated from the citation network in Figure \ref{fig:journal_tree}, with one important difference: the model on the left was estimated with the \pps algorithm proposed in this paper and the model on the right was estimated with the \texttt{HCD-Spec} algorithm proposed by \cite{li}.
   The \pps network on the left recreates the single mode displayed in the original data's diagnostic plot (Figure \ref{fig:secondeighist}). The ultrametric procedure on the right fails to recreate this pattern. 
   Appendix \ref{sec:simulation} gives more details of this experiment. } \label{fig:split_vector_generated}
   
\end{figure}

Figure \ref{fig:split_vector_generated} displays the diagnostic plots of  parametric bootstrap networks fitted using the journal citation graph (same graph in Figure \ref{fig:journal_tree}). The diagnostic plot from the $\T$-Stochastic Graph network (left) has a similar single-mode structure compared to real data, while the diagnostic plot for Li's model (right) presents multiple modes. This multi-mode phenomenon appears in many other large networks as well (Figure \ref{fig:parametric_bootstrap})
One might notice that Theorem \ref{thm:ultrametric_split} implies two symmetric modes around zero, but the right panel displays asymmetric modes. This is because the top-down bi-partition algorithm allows the estimated tree to be unbalanced.
However, this tolerance for imbalance does not alleviate the fundamental problem: there are still multiple modes.

Unsurprisingly, the network that is simulated from an ultrametric model generates a diagnostic plot with multiple modes. These modes would enable efficient splitting. Unfortunately, the real data does not generate these multiple modes; \textit{they are an artifact of modeling assumptions}. 

\begin{figure}[!ht]
    \centering
    \includegraphics[angle = 270, width = 6in]{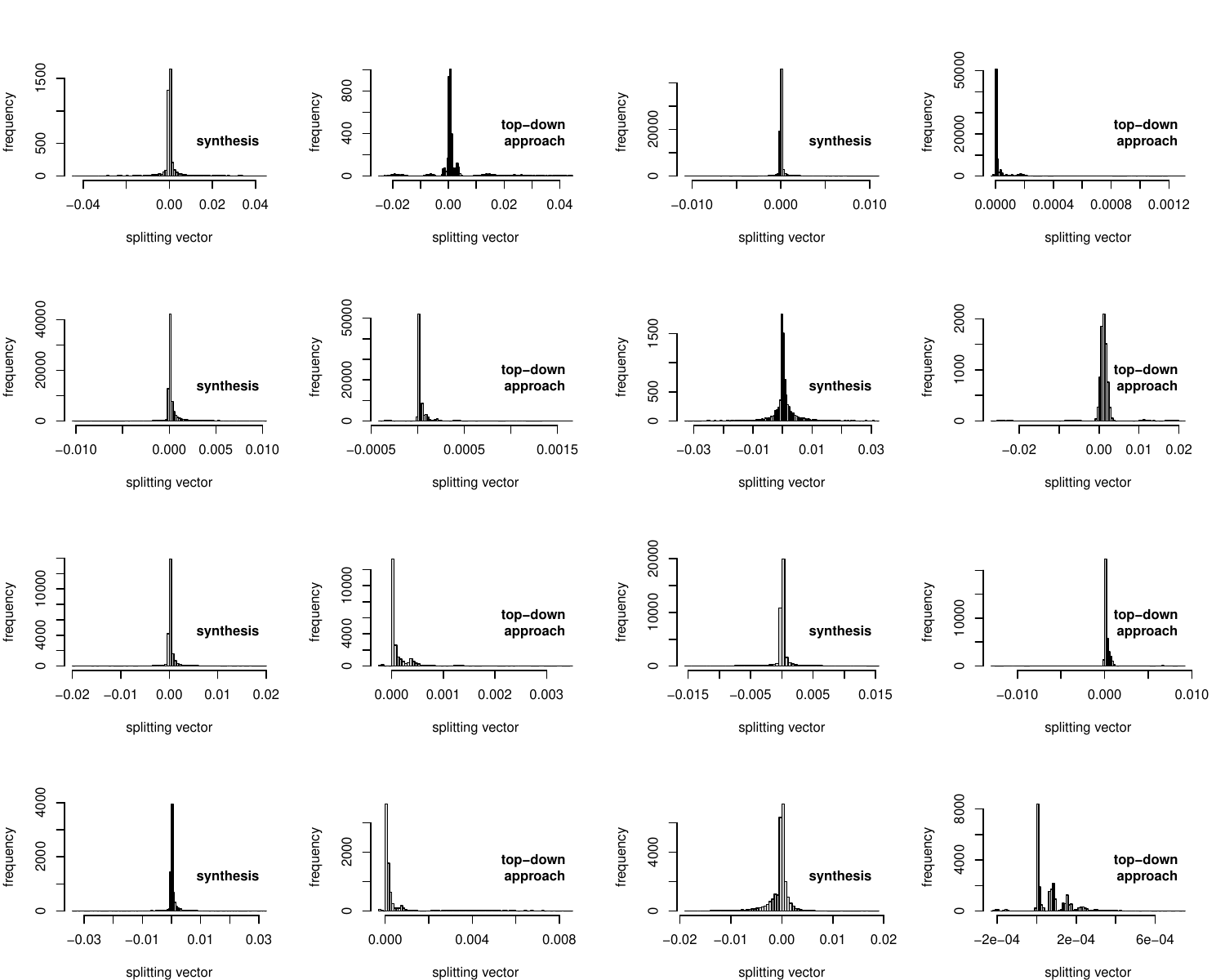}
    \caption{Parametric bootstrap plots for social networks in Table \ref{tab:splitvecs}. All plots are created in the same way as Figure \ref{fig:split_vector_generated}. For the Wikipedia page network, a random subsample of 24644 nodes (1\% of the original nodes) is taken since the original network is too large for the bi-partition algorithm to work.}
    \label{fig:parametric_bootstrap}
\end{figure}

\clearpage
\section{Relationship with Other hierarchical models}\label{appendix:other_models}

In modern literature, the idea of encoding hierarchies in social networks with latent tree graphs can be traced back to the Hierarchical Random Graph (HRG) proposed by \cite{clauset2008hierarchical}. 
In the HRG, each edge is a Bernoulli random variable where $\lambda_{ij}$ depends upon the lowest common ancestor of $i$ and $j$. If the pair of nodes $(u, v)$ have the same lowest common ancestor as nodes $(i, j)$, then $\lambda_{ij} = \lambda_{uv}$. 
Following this parameterization in HRGs, 
\cite{lei} and \cite{li} proposed the Binary Tree Stochastic Blockmodel (BTSBM) and the generalized Binary Tree Stochastic Blockmodel (GBTSBM) by combining the HRG and the SBM.

Although these models define $\lambda_{ij}$ in a different manner, \textbf{HRGs, BTSBMs, and GBTSBMs are all $\T$-Stochastic Graphs with certain restrictions}. To discuss how the $\T$-Stochastic Graph is more general than these three models, we introduce the following conditions. 
We say an unrooted tree $\T$ with distance $d(\cdot, \cdot)$ satisfies some of the following conditions if there exists a way to root $\T$ such that the rooted tree $\T_r$ with distance $d(\cdot, \cdot)$ satisfies those conditions. 

\begin{enumerate}
    \item[(C1)] \textbf{Ultrametric.} All leaf nodes are equidistant to the root.
    \item[(C2)] \textbf{Multilevel-Ultrametric.} Internal nodes at the same level\footnote{In a rooted tree $\T_r$, the level of node $u$ is the number of edges on the path from the root $r$ to node $u$.} are equidistant to the root.
    \item[(C3)] \textbf{Binary-Until-Leaves.} Every internal node has exactly two children, except for the parents of leaf nodes.
\end{enumerate}

Let $\mathscr{F}_{\text{TSG}}$ denote the set of distributions of all $\T$-Stochastic Graphs. Specifically, each element $f_{\T, d}$ in $\mathscr{F}_{\text{TSG}}$ represents the distribution of a random graph $A$ that is a $\T$-Stochastic Graph with distance $d(\cdot, \cdot)$. Similarly, let $\mathscr{F}_{\text{HRG}}$, $\mathscr{F}_{\text{BTSBM}}$, and $\mathscr{F}_{\text{GBTSBM}}$ denote sets of distributions for the HRG, the BTSBM, and the GBTSBM, respectively. 

Certain models impose specific constraints to establish consistency proofs for their proposed algorithms, such as the assortative assumption in the BTSBM. In the context of $\T$-Stochastic Graphs, most assumptions can be interpreted as whether negative edge weights (as in Remark \ref{rmk:negative_edge_weights})  are allowed. 
To make things neat, Theorem \ref{thm:model_relationship} assumes that all models may include negative edge weights. The formal definition of specific models and assumptions are available in Appendix \ref{appendix:other_models_definition} and the proofs of Theorem \ref{thm:model_relationship} can be found in Appendix \ref{appendix:other_models_equivalence}. In addition, Appendix \ref{appendix:other_models_comprehensive_equivalence} contains a more comprehensive theorem that examines model equivalence in light of the presence or absence of negative edge weights.

\begin{figure}[!ht] %
   \centering
   \includegraphics[width=5.5in]{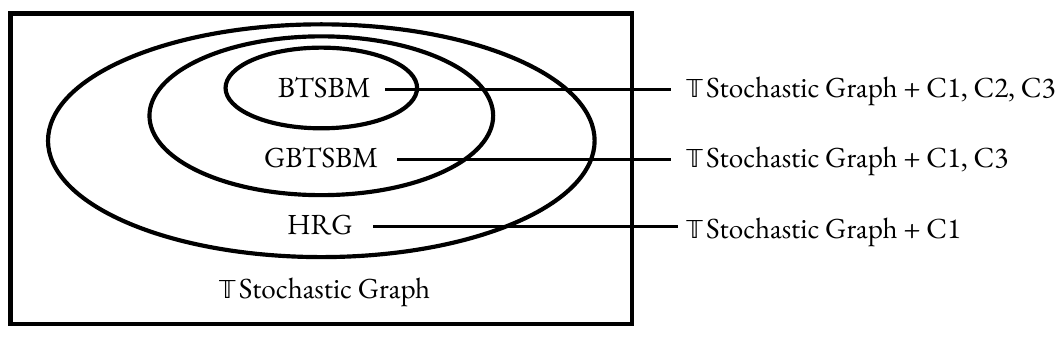} 
   \caption{This figure summarizes the result in Theorem \ref{thm:model_relationship}. Starting with a $\T$-Stochastic Graph, one can get an HRG by imposing the ultrametric condition, a generalized BTSBM by further forcing the binary condition, and a BTSBM if the Multilevel-Ultrametric condition is also satisfied.}
   \label{fig:venn_plot}
\end{figure}

\begin{theorem}\label{thm:model_relationship}
\footnote{
\cite{clauset2008hierarchical, li, lei} all assume $A_{ij}$ to be Bernoulli distributed, and we discuss those models under a more general semi-parametric setting where only $\E(A_{ij})$ are defined. If assuming Bernoulli distribution, then the equivalence still holds by reducing the $\T$-stochastic graph to be Bernoulli distributed. }
\leavevmode

$\mathscr{F}_{\text{HRG}} = \{f_{\T, d}\in \mathscr{F}_{\text{TSG}}: \T \mbox{ with distance $d(\cdot, \cdot)$ satisfies C1}\}, $

$\mathscr{F}_{\text{GBTSBM}} = \{f_{\T, d}\in \mathscr{F}_{\text{TSG}}: \T \mbox{ with distance $d(\cdot, \cdot)$  satisfies C1, C3}\}, $

$\mathscr{F}_{\text{BTSBM}} = \{f_{\T, d}\in \mathscr{F}_{\text{TSG}}: \T \mbox{ with distance $d(\cdot, \cdot)$  satisfies C1, C2, C3}\}.$

\end{theorem}

Out of the three conditions C1, C2, and C3, the most crucial one that reveals the essential distinction between the $\T$-Stochastic Graphs and the other three models is the ultrametric condition. It's helpful to have some sense of it. The concept of ultrametric trees originates from research on phylogeny and evolution; it employs the molecular clock hypothesis \citep{zuckerkandl1962molecular, margoliash1963primary, kumar2005molecular} which assumes that all species evolve at the same rate, i.e. the substitution rate at a given site is the same across all species in the evolution tree. This assumption is not upheld in numerous real-world cases and could result in systematic biases in phylogenetic tree reconstructions \citep{gillespie1984molecular, pulquerio2007dates}. Consequently, the ultrametric assumption is no longer popular in modern phylogenetic studies. Another interpretation for ultrametricity is the dendrogram.  Dendrograms are usually produced by hierarchical clustering methods. Given a distance or dissimilarity matrix, clusters are combined recursively by certain linkage criteria. The well-known UPGMA (unweighted pair group method with arithmetic mean) \citep{sokal1958statistical} is such an agglomerative hierarchical clustering method that was used in phylogenetic tree reconstruction. This method is no longer commonly used because it heavily relies on the ultrametric assumption \citep{felsenstein1984distance}. One way to create ultrametric trees is cutting a dendrogram at the same height across all branches. Alternatively, cutting different branches at different heights leads to a non-ultrametric tree. In a decision tree or random forest, this is analogous to pruning. Previously, \cite{tantrum2003assessment} discussed the necessity of pruning in hierarchical clustering and proposed a non-ultrametric pruning process.

Apart from the HRG, the BTSBM, and the GBTSBM, there are some other hierarchical models that describe the latent structure from different perspectives. The Hierarchical Stochastic Blockmodel (HSBM) \citep{hrdpg} encodes hierarchies by extending the RDPG representation of the SBM to have a multi-layer mixture of point mass distributions. The formal definition can be found in Appendix \ref{appendix:other_models_hsbm_counter_example}. Although there are cases where certain HSBMs can be considered as $\T$-Stochastic Graphs, and vice versa, neither model completely encompasses the other (see Appendix \ref{appendix:other_models_hsbm_counter_example} for two counter-examples). 
There is another research direction for modeling hierarchical structures in social networks based on the Bayesian approach. In particular, \cite{peixoto2014hierarchical} presents a nested generative model where blocks in the previous level are considered as nodes in the next level, and edge counts are considered as edge multiplicities between each node pair. There are no elegant relationships between the $\T$-Stochastic Graph and this model as they employ different notions of hierarchies.

\subsection{Definitions of HRGs, GBTSBMs, and BTSBMs}\label{appendix:other_models_definition}

The formal definition of the HRG, the GBTSBM, and the BTSBM are given in Definitions \ref{def:hrg}, \ref{def:gbtsbm}, and \ref{def:btsbm}, respectively, using the notations established in this paper. All three models incorporate hierarchical structures represented by a rooted tree $\T_r = (V, E)$ and define $\lambda_{ij}$ based on a set of probabilities\footnote{the word probabilities does not imply $\sum_{u\in V\setminus V_\ell}\pr(u) = 1$, $\pr(u)$ is referred to as probabilities because \citep{clauset2008hierarchical, li, lei} all consider a Bernoulli distribution for $A_{ij}$ and thus $\lambda_{ij}$ is the probability that $A_{ij} = 1$.} $\{\pr(u): u\in V\setminus V_\ell\}$ for all internal nodes. In a rooted tree, we use $lvl(u)$ to denote the \textbf{level} of node $u$, $lvl(u) = k$ if there are $k$ edges on the path from the root to $u$. The root node $r$ is a level-0 node, i.e., $lvl(r) = 0$. Additionally, we use the notation $i\wedge j$ to represent the nearest common ancestor of nodes $i$ and $j$, as defined in Appendix \ref{appendix:tob}.

\begin{definition}\label{def:hrg}
The graph $A\in \R^{n\times n}$ is a $\T_r$-HRG as defined in \citep{clauset2008hierarchical} if
\begin{enumerate}
    \item rows and columns in $A$ are indexed by leaf nodes in $\T_r$,
    \item for any pair of nodes $i, j \in V_{\ell}$,
    \[\lambda_{ij}^{HRG} = \pr(i\wedge j).\]
\end{enumerate}
\end{definition}

\begin{definition}\label{def:gbtsbm}
The graph $A\in \R^{n\times n}$ is a $\T_r$-GBTSBM as defined in \citep{li} if $\T_r$ satisfies the Binary-Until Leaves condition (C3 in Appendix \ref{appendix:other_models}) and $A$ is an SBM with block assignments $z(\cdot)$ and connectivity matrix $B\in \R_+^{k\times k}$ such that 
such that

\begin{enumerate}
\item nodes in graph $A$ are indexed by nodes in $V_\ell$, the set of leaf nodes,
\item rows and columns of $B$ are indexed nodes in $V_z$, the set of parent nodes of leaf nodes,
\item node $i$ belongs to block $u$ if $u$ is the parent of $i$, i.e. $z(i) = u\Longleftrightarrow p(i) = u$,
\item for any pair of nodes $u, v \in V_z$, 
\[B_{uv} = \pr(u\wedge v),\]specifically, when $u=v$, the nearest common ancestor of $u$ and $v$ is the node itself, that is, $B_{uu} = \pr(u)$.
\end{enumerate}
\end{definition}

\begin{definition}\label{def:btsbm}
The graph $A$ is a $\T_r$-BTSBM as defined in \citep{lei} if $A$ is a GBTSBM and node probabilities $\pr(\cdot)$ further satisfy
\begin{equation}\label{eq:other_models_btsbm}
    \pr(m) = \pr(n)
\end{equation}
for any two nodes $m, n$ that are at the same level, i.e., $lvl(m) = lvl(n)$.
\end{definition}

\begin{remark}
(Binary Assumption) The HRG model also assumes a binary tree since it uses a dendrogram to represent $\T_r$. However, we do not categorize HRGs as binary structures for two reasons: Firstly, HRGs can represent non-binary structures since they allow $\pr(u) = \pr(v)$ for any neighbor nodes $u$ and $v$. Secondly, the algorithm employed in \citep{clauset2008hierarchical} can produce non-binary trees. Unlike HRGs, the BTSBMs and GBTSBMs heavily rely on the binary assumption since both \citep{li, lei} require $\pr(u) \neq \pr(v)$ to ensure distinct eigenvalues, which ensures the consistency of the recursive bi-partition algorithm.
\end{remark}

\citep{li, lei} examine the BTSBMs and GBTSBMs under specific assumptions, which are itemized as follows. The GBTSBMs are discussed in relation to the weak assortative assumption, while the BTSBMs are discussed with regard to both the assortative and disassortative assumptions.

\begin{enumerate}
    \item[(A1)] \textbf{Weak assortative.} $\pr(u) < \pr(v)$ when $u$ is the parent of $v$
    \item[(A2)] \textbf{Assortative.} $\pr(u) < \pr(v)$ when $lvl(u) < lvl(v)$
    \item[(A3)] \textbf{Disassortative.} $\pr(u) > \pr(v)$ when $lvl(u) < lvl(v)$
\end{enumerate}

\subsection{HRGs, BTSBMs, and generalized BTSBMs are all $\T$-Stochastic Graph (proof of Theorem \ref{thm:model_relationship})}\label{appendix:other_models_equivalence}

Theorem \ref{thm:model_relationship} is an outcome derived directly from the three lemmas presented below. Notably, the results outlined in Theorem \ref{thm:model_relationship} are established under the cases when
\begin{enumerate}
    \item None of the assumption (A1)-(A3) is assumed, that is, HRGs, GBTSBMs, and BTSBMs are defined as in Definitions \ref{def:hrg}, \ref{def:gbtsbm}, and \ref{def:btsbm}, respectively.
    \item $\T$-Stochastic Graphs allow for negative internal edge weights.
\end{enumerate}

\begin{lemma}\label{lemma:hrg->tsg}
Any $\T_r$-HRG is a $\T$-stochastic graph with $\T$ satisfies C1, and vice versa. 
\end{lemma}
\begin{proof}

We first show that any $\T_r$-HRG is a $\T$-stochastic graph with $\T$ satisfies (C1). This can be proved by assigning $\theta_{uv}$ to all edges in $\T_r$, and then defining distance as $d(u, v) = -\log(\theta_{uv})$. 
For any two nodes $u, v$ that are neighbors, without loss of generality, we assume $u$ is the parent of $v$ and define 
\begin{equation}\label{eq:other_models_define_theta}
    \theta_{uv} = \sqrt{\frac{\pr(u)}{\pr(v)}},
\end{equation}
notice that $\pr(\cdot)$ are defined only for internal node, if $v$ is a leaf node, we assume $\pr(v) = 1$. Then its easy to verify that 
\[ \prod_{(u, v)\in E(i, j)} \theta_{uv}= \prod_{(u, v)\in E(i, i\wedge j)} \theta_{uv} \prod_{(u, v)\in E(j, i\wedge j)} \theta_{uv} = \sqrt{\pr(i\wedge j)} \sqrt{\pr(i\wedge j)}  = \pr(i\wedge j).\]
To verify C1, just notice that for any leaf node $i$, $\prod_{(u, v)\in E(r, i)} \theta_{uv} = \sqrt{\pr(r)}$.

Now we show that any $\T$-stochastic graph with $\T$ satisfies (C1) is a $\T_r$-HRG. This is proved by assigning $\pr(u)$ to all internal nodes $u$.
Since $\T$ satisfies (C1), there exists a way to root $\T$ as $\T_r$ such that for all leaf nodes $i$, \[d(r, i)\equiv c,\] where $c$ is a constant. Define $\pr(u)$ as 
\begin{equation}\label{eq:other_models_define_P}
   \pr(u) = \exp(-2c+2d(u, r))
\end{equation}
Now consider any two leaf nodes $i, j$, the following equalities complete the proof
\begin{align*}
    \lambda^{TSG}_{ij} &= \exp(-d(i, j)) = \exp(-d(i, i\wedge j) - d(j, i\wedge j)) \\
    &= \exp(-d(i, r) + (r, i\wedge j) -d(j, r) + (r, i\wedge j))\\
    & = \exp(-2c + 2d(i\wedge j, r))\\
    & = \pr(i\wedge j)
\end{align*}

\end{proof}
\begin{lemma}\label{lemma:gbtsbm->hrg}
Any $\T_r$-GBTSBM is a $\T_r$-HRG with $\T_r$ satisfies C3, and vice versa.
\end{lemma}
\begin{proof}
As the nearest common ancestor of any pair of leaf nodes $i$ and $j$ is identical to the nearest common ancestor of $p(i)$ and $p(j)$, we obtain \[\lambda_{ij}^{GBTSBM} = B_{z(i)z(j)} = B_{p(i)p(j)} = \pr(p(i)\wedge p(j)) = \pr(i\wedge j) = \lambda_{ij}^{HTG}\]
\end{proof}

\begin{lemma}\label{lemma:btsbm->hrg}
Any $\T_r$-BTSBM is a $\T_r$-HRG with $\T_r$ satisfies C3 and C2, and vice versa.
\end{lemma}
\begin{proof}
By Lemma \ref{lemma:hrg->tsg} and Lemma \ref{lemma:gbtsbm->hrg}, it suffices to demonstrate the following two things
\begin{enumerate}
    \item (C2) holds for the $\T$-Stochastic Graph obtained from the $\theta_{uv}$'s constructed using Equation \eqref{eq:other_models_define_theta}. Consider any two nodes $m, n$ at the same level. By definition, we have $\pr(n) = \pr(m)$, and thus
\[\prod_{(u, v)\in E(r, m)} \theta_{uv} = \sqrt{\frac{\pr(r)}{\pr(m)}} = \sqrt{\frac{\pr(r)}{\pr(n)}} = \prod_{(u, v)\in E(r, n)} \theta_{uv}.\]
    \item For any $\T$-Stochastic Graph with $\T$ satisfies (C2), probabilities $\pr(u)$ defined through Equation \eqref{eq:other_models_define_P} satisfies Equation \eqref{eq:other_models_btsbm}. To show this, consider any two nodes $m, n$ such that $lvl(m) = lvl(n)$, since $\T$ satisfies (C2), we have $d(m, r) = d(n, r)$ and thus 
    \[\pr(m) = \exp(-2c+2d(m, r)) = \exp(-2c+2d(n, r)) = \pr(n).\]
\end{enumerate}
\end{proof}

\subsection{Comprehensive Equivalence Theorem}\label{appendix:other_models_comprehensive_equivalence}

Let $\mathscr{F}_{TSG(+)}$ and $\mathscr{F}_{TSG(-)}$ represent the collection of distributions of $\T$-Stochastic Graphs with all internal edges being positive and negative, accordingly. It is important to note that the union of $\mathscr{F}_{TSG(+)}$ and $\mathscr{F}_{TSG(-)}$ is not $\mathscr{F}_{TSG}$ since some $\T$-Stochastic Graphs have both positive and negative edge weights. Similarly, $\mathscr{F}_{GBTSBM(A1)}$, $\mathscr{F}_{BTSBM(A2)}$, and $\mathscr{F}_{BTSBM(A3)}$ refer to the sets of distributions of GBTSBMs with the weak assortative assumption, BTSBMs with the assortative assumption, and BTSBMs with the disassortative assumption, respectively.

\begin{theorem}
\leavevmode

$\mathscr{F}_{GBTSBM(A1)} = \{f_{\T, d}\in \mathscr{F}_{TSG(+)}: \T \mbox{ with distance $d(\cdot, \cdot)$ satisfies C1, C3}\}, $

$\mathscr{F}_{BTSBM(A2)} = \{f_{\T, d}\in \mathscr{F}_{TSG(+)}: \T \mbox{ with distance $d(\cdot, \cdot)$ satisfies C1, C2, C3}\}, $

$\mathscr{F}_{BTSBM(A3)} = \{f_{\T, d}\in \mathscr{F}_{TSG(-)}: \T \mbox{ with distance $d(\cdot, \cdot)$ satisfies C1, C2, C3}\}, $

\end{theorem}
\begin{proof}
    First notice that when Equation \eqref{eq:other_models_btsbm} is satisfied, (A1) and (A2) are the same. Therefore to show the first two equivalence, we just need to show that (A1) is equivalent to positive edge weight. 
    
    Consider any two neighbor nodes $u, v$, without loss of generality, assume $u$ is the parent of $v$. When (A1) is satisfied, 
    \[\pr(u) < \pr(v),\]
    by Equation \eqref{eq:other_models_define_theta}, we have \[\theta_{uv}<1,\] and thus \[d(u, v) = -\log(\theta_{uv})>0.\] 
    When all edge weights are positive, \[d(r, u) = d(r, v) - d(u, v) < d(r, v),\] 
    by Equation \eqref{eq:other_models_define_P}, we have \[\pr(u)<\pr(v).\] 
    
    To show the third equivalence, it suffices to prove (A3) is equivalent to all edge weights being negative, using Equation \eqref{eq:other_models_define_theta} and \eqref{eq:other_models_define_P}, this is evident by the same logic as above.

\end{proof}

\begin{figure}[h] %
   \centering
   \includegraphics[width=6in]{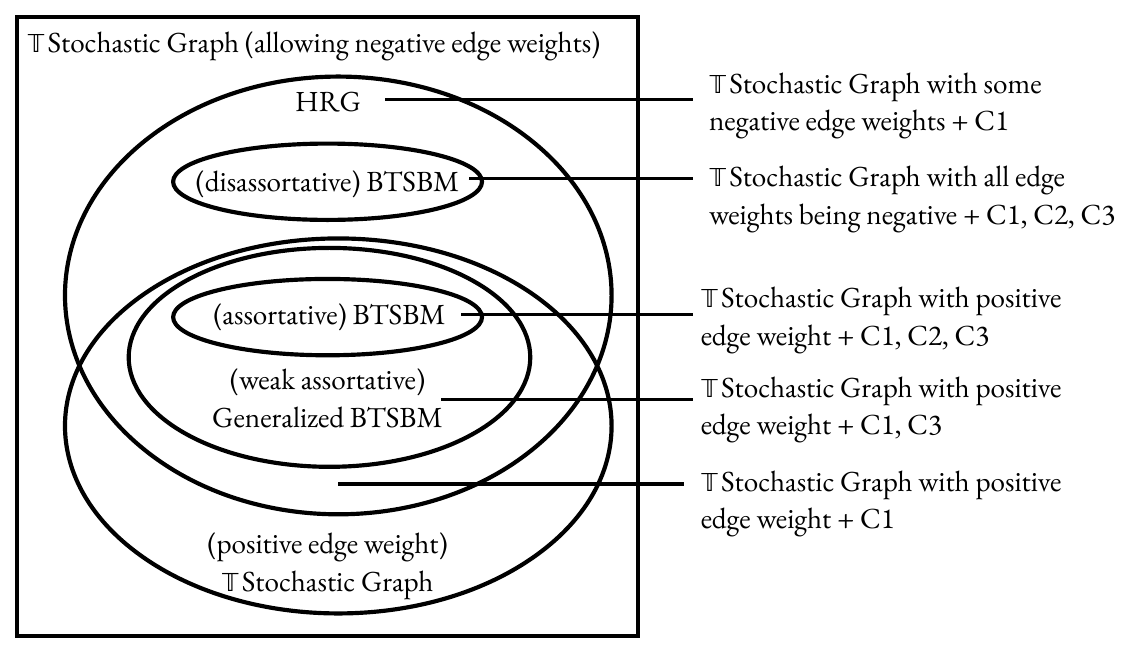} 
   \caption{This figure summarizes model relationships considering negative edge weights.}
   \label{fig:venn_plot_negative}
\end{figure}

\subsection{HSBMs: Definitions and Counter examples}\label{appendix:other_models_hsbm_counter_example}

Inspired by the RDPG representation of the SBM, the Hierarchical Stochastic Blockmodel (HSBM) \citep{hrdpg} describes hierarchies by extending the RDPG to have a multi-layer mixture of point mass distributions. To distinguish the HSBM from general SBM, \cite{hrdpg} requires that the SBM satisfies an ``Affinity Assumption'' defined by a hierarchy. Definition \ref{def:hsbm} formally defines the HSBM using notations in this paper.

\begin{definition}\label{def:hsbm}
The graph $A$ is an $\T_r$-HSBM as defined in \citep{hrdpg} if $A$ is an SBM with connectivity matrix $B$ such that 
\begin{enumerate}
    \item rows and columns in $B$ are indexed by leaf nodes in $\T_r$
    \item the elements in $B$ satisfies that 
    \begin{equation}\label{eq:hsbm}
        \max_{lvl(u\wedge v) = l_1} B_{uv} < \min_{lvl(u'\wedge v') = l_2} B_{u'v'}, \quad \forall \ l_1 < l_2,
    \end{equation}
    where $lvl(u)$ is the level of node $u$, as defined in Appendix \ref{appendix:other_models_definition}.
\end{enumerate}
\end{definition}

The HSBM does not directly specify probabilities using the tree structure; instead, it imposes constraints on the relative scales. It is noteworthy that a level $\ell$ HSBM is always a level $\ell-1$ HSBM, as decreasing the level simply reduces the constraints, and all SBMs are level-1 HSBMs. In order to avoid the degeneration of HSBMs into general SBMs, we concentrate on HSBMs with levels $>1$. While there are obvious distinctions between HSBMs and $\T$-Stochastic Graphs, such as the former being an SBM and the latter a DCSBM, this difference is not definitive. Below are two counterexamples illustrating that, even if we solely focus on the connectivity matrix $B$, the HSBM and the $\T$-Stochastic Graph are not interchangeable.

\begin{exmp} Consider a SBM with the connectivity matrix $B$ defined as $B_{uv} = \exp(D_{uv})$.

\begin{enumerate}
    \item ($\T$-Stochastic Graph is not HSBM) If $D_{uv} = d(u, v)$, where $d(\cdot, \cdot)$ is the additive distances constructed by edge weights in Figure \ref{fig:counterexample}, then this SBM is a $\T$-stochastic Graph and is not an HSBM (there is no way to root $\T$ such that Equation \eqref{eq:hsbm} is satisfied).
    \begin{figure}[htbp]
    \centering
	\includegraphics[width=2.5in]{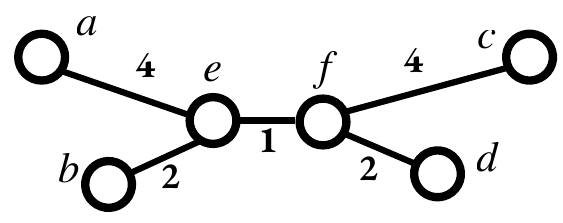}
   \caption{An example of a $\T$-Stochastic Graph that is not a HSBM}
   \label{fig:counterexample}
\end{figure}
    \item (HSBM is not $\T$-Stochastic Graph) If $D$ is defined as the following matrix, then this SBM is a HSBM and is not a $\T$-stochastic Graph.
\[D = \left(
\begin{matrix}
0 & 1 & 3 & 3 \\
1 & 0 & 3 & 4\\
3 & 3 & 0 & 2 \\
3 & 4 & 2 & 0
\end{matrix}\right)
\]
\end{enumerate}
\end{exmp}

\subsection{Relationship between the assortativity, the affinity, and the positive edge weights assumption}\label{appendix: other_models_assortativity_affinity}

In a multi-level HSBM, the graph can be partitioned into multiple subgraphs, and the affinity assumption imposes the condition that the density of connections within each subgraph is higher than the density of connections between subgraphs. This condition resembles the assortative assumption for BTSBMs and the weak assortative assumption for GBTSBMs. To illustrate the differences between these assumptions, we employ Figure \ref{fig:populationA} as an example.

\begin{figure}[h]
    \centering
	\includegraphics[width=2in]{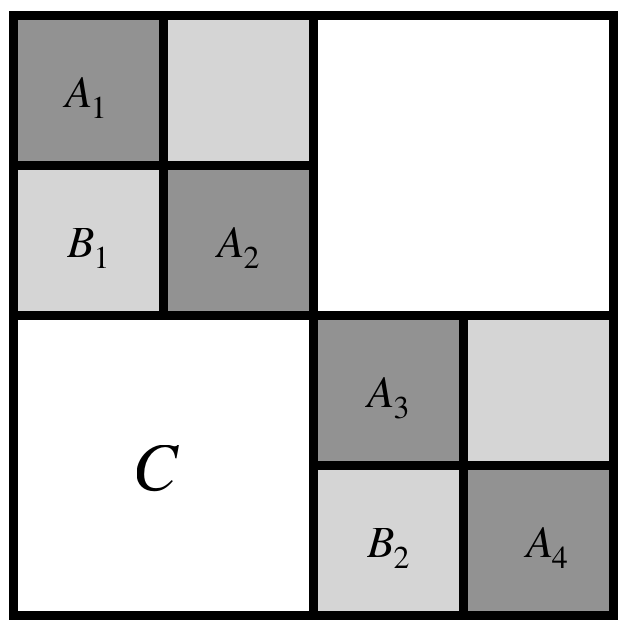}
   \caption{This figure presents a natural intuition for hierarchical models. Generally, in the population adjacency matrix ($\mathscr A$), elements in diagonal block matrix tend to be larger than off-diagonal block matrix, that is, $A_i>B_j>C$ for $i \in [4], j \in [2]$.}
   \label{fig:populationA}
\end{figure}

Among the models discussed, the BTSBM under the assortative assumption is the most restrictive model as it assumes equal probabilities for nodes at the same level and a strict increasing trend from the root to the leaves. Specifically, it assumes $A_1 = A_2 = A_3 = A_4$, $B_1 = B_2$, and $A_1>B_1>C$, where the notation ``$M>N$'' indicates that any element in matrix $M$ is greater than all elements in matrix $N$. In contrast, while the HSBM does not require blocks on the same level to be identical, it still imposes the affinity assumption. Specifically, in an HSBM, $A_1, A_2, A_3$, and $A_4$ are not necessarily the same, but they are all greater than $B_1$ and $B_2$. The generalized BTSBM is more flexible, as it only requires the weak assortative assumption, which demands the probability assigned to any node to be strictly greater than its parent. In particular, it only requires $B_1$ to be smaller than $A_1$ and $A_2$, but not necessarily smaller than $A_3$ and $A_4$.

On the other hand, the $\T$-Stochastic Graph does not impose any such assumptions but incorporates the increasing trend in its definition. Notably, $\lambda_{ij} = \exp(-d(i, j))$, which implies that, as long as the distance is positive, the more edges on the path from $i$ to $j$, the less likely that $i$ and $j$ are connected.

\clearpage
\section{Numerical Results}\label{sec:simulation}

This section investigates the empirical performance of \pps under several different tree structures by comparing it with two recursive bi-partition algorithms proposed by \citep{li, lei}. We refer to them as the Adjacency recursive bi-partition algorithm and the Laplacian recursive bi-partition algorithm, accordingly, based on how they calculate the splitting vector.
Recursive bi-partition is a simple and representative algorithm that recovers hierarchical structures in a top-down fashion. Studies by \cite{li} and \cite{lei} have demonstrated its consistency under the BTSBM and the generalized BTSBM. Additionally, it has displayed competitive performance compared to the regularized spectral clustering, Louvain's modularity method, and \cite{dasgupta2006spectral}'s recursive partition method \citep{li}.

For all subsequent simulations, $A_{ij}$ is generated from a Poisson distribution. The column ordering of $\widehat Z$ is determined by matching the columns of $\widehat Z$ to the true $Z$ using the Gale-Shapley algorithm.
All three algorithms are provided with the true $k$. For recursive bi-partition algorithms, we repeatedly split the community with the largest second eigenvalue until there are $k$ leaf nodes in the tree.

The following four measures of accuracy are used to compare the three algorithms, and the performance of each algorithm is averaged over 100 independent replications.

1. Robinson-Foulds Metric \citep{robinson1981comparison}: This is a popular tree comparison metric in phylogenetic studies. It measures the topology differences between two trees
by computing the number of non-trivial splits that differ between them. To eliminate the influence of tree size on the metric, we adopted the normalized Robinson-Foulds distance, which always falls within the range of $[0, 1]$.

2. Correct Recovery Rate: This is the ratio of instances in which the algorithm successfully recovers the true latent tree structure. In other words, it represents the frequency at which the Robinson-Foulds Metric is exactly zero. It's important to note that this is a strict criterion because any incorrectly inferred edge leads to an unsuccessful reconstruction.

3. Error in $\widehat B$: While the above two measures focus on the topology structure and ignore edge weights, this measure calculates the Frobenius norm of the difference between $\widehat B$ and $B$. It provides a sense of how well the edge lengths are estimated. For \pps, $\widehat B$ are re-estimated after the tree reconstruction. For the recursive bi-partition algorithm, the probabilities are estimated by averaging the number of edges between (or within) blocks, and we normalize $\widehat B$ to make sure all diagonal elements are equal to one.

4. Block Membership Error Rate: This is the percentage of the vertexes that are wrongly clustered. Since \texttt{vsp} gives a continuous estimate of the block membership via $\widehat Z$, we assign node $i$ to the block corresponding to the largest elements in row $i$, i.e., assign node $i$ to block $j = \underset{j}{\mathrm{argmax}} \ \widehat Z_{ij}$.

\subsection{Binary Tree}\label{sec:simu_binary}

We start with $\T_z$ being binary trees. In this section, we keep $k$ fixed at 32 and $n$ at 6400 while varying the expected degrees from $10$ to $70$. The \pps algorithm estimates $\widehat B^{nn} = \frac{1}{n^2}\left[Z^TAZ\right]_+$ in the second step and applies a neighbor-joining algorithm in the third step.

Both ultrametric (balanced binary tree) and non-ultrametric (random binary tree) settings are explored. The balanced binary trees have edge lengths of 0.5, resulting in a 0.36 between-blocks to within-blocks edge ratio. All random binary trees are generated through a branching process with edge lengths sampled independently from the uniform distribution. This is a commonly used phylogenetic tree generating procedure and further details can be found in \citep{paradis2012analysis}. To avoid nearly zero edge lengths, we add 0.1 to all generated lengths.

Figure \ref{fig:simu_binary_tree} presents the simulation results with respect to all four metrics. For a balanced binary tree, Adjacency recursive bi-partition method outperforms the other two methods, achieving high accuracy even when the expected degree is low. This is expected since the Adjacency recursive bi-partition algorithm is designed for this simple homogeneous case. However, as the expected degree increases, \pps exhibits compatible performance. 

\begin{figure}[!ht]
    \centering
	\includegraphics[width=6.5in]{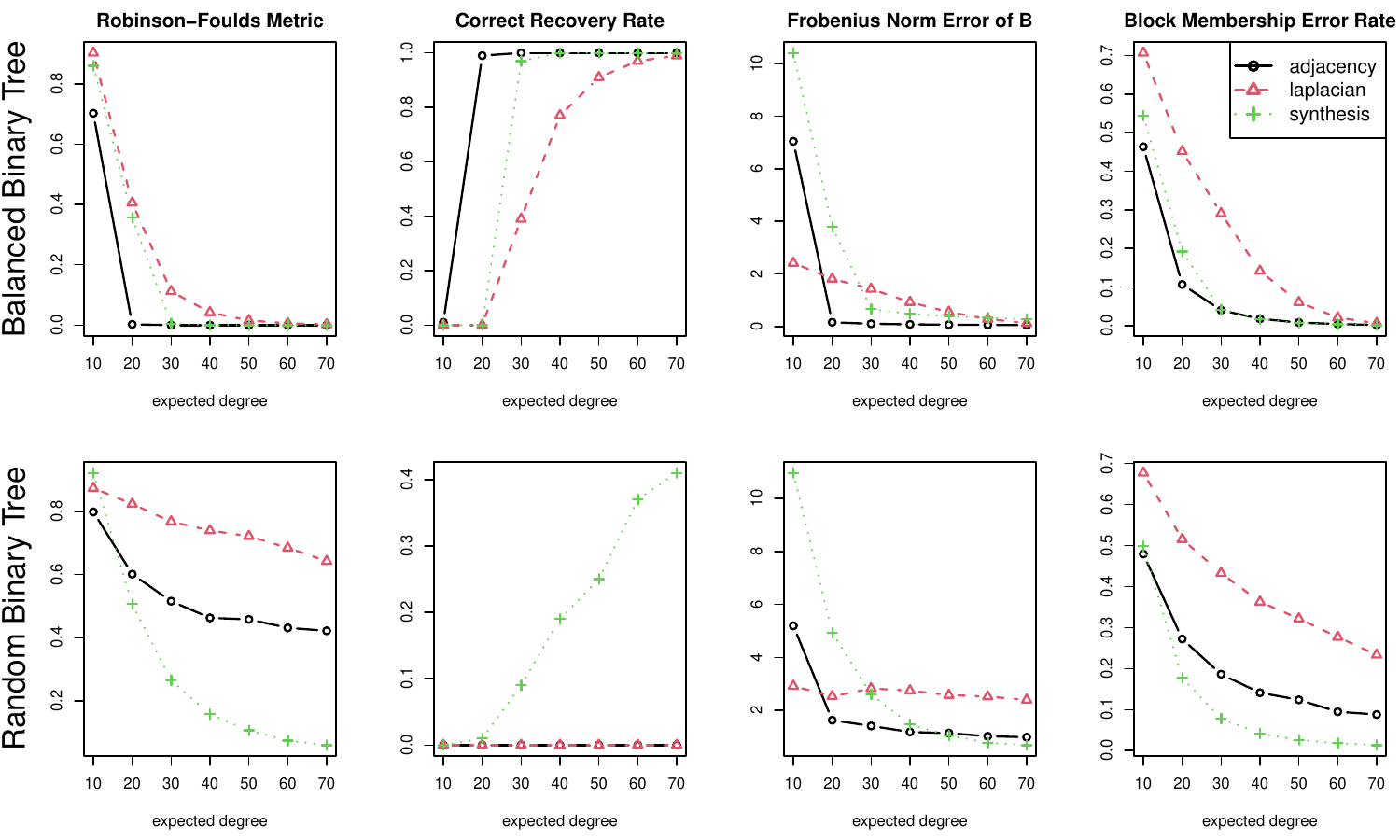}
   \caption{Performance of the Adjacency recursive bi-partition, the Laplacian recursive bi-partition, and the \pps under balanced binary trees and random binary trees with varying average degrees.}
   \label{fig:simu_binary_tree}
\end{figure}

For random binary trees, \pps achieves better performance under all four measures. Both recursive bi-partition algorithms suffer from systematic bias caused by the ultrametric assumption, failing to recover the correct tree structure even when the graph is dense. It is worth noting that the block membership error rate for recursive bi-partition algorithms decreases as the expected degree goes higher. This phenomenon can be explained by \citep{aizenbud2021spectral}, which presents the population results of splitting vectors in binary trees (without the need for ultrametricity). This study shows that the population splitting vector can correctly divide the tree, but it is not guaranteed to split the tree at the correct ``root''. In the context of $\T$-Stochastic Graphs, it implies that the recursive bi-partition algorithms could get the block membership correct, but are not able to recover the tree topology. To obtain a correct tree structure, one plausible approach is to merge subtrees after identifying block structures. This leads to an algorithm that is similar to our three-step recovery procedure. Instead of \texttt{vsp}, the first step applies a recursive bi-partition algorithm to identify block membership; the second step estimates the distance matrix $D$ or the connectivity matrix $B$; then the third step applies NJ on $\widehat D$ or other techniques such as the spectral merging step in \citep{aizenbud2021spectral} on $\widehat B$ to merge subtrees recursively.

\subsection{Multifurcating Tree}\label{sec:simu_multifurcating}
This section delves into scenarios where $\T_z$ is allowed to be non-binary. The most basic and representative case of multifurcating trees is a star-shaped tree with equal edge length. This is also an ultrametric, but non-binary tree. We examine the performance of three methods under the star-shaped tree as $k$ increases. The expected degrees are set to be 50 to make sure the instability is caused by multifurcation instead of sparsity. The edge lengths are set to be 1.5 so there is a significant difference between within-block edge density and between-block edge density. Since the Frobenius norm is affected by the dimension of $B$, we divide this metric by $k^2$ to make sure it's a fair comparison across different $k$ values. Apart from a star-shaped tree, we also explore the more general cases: random multifurcating trees. This is a challenging problem, we fix $k = 16$ and vary the expected degree from $20$ to $200$. The tree-generating process includes two steps. Firstly, a random binary tree is generated using the same algorithm described in Appendix \ref{sec:simu_binary}, then branches shorter than 0.5 are shrunk to zero. 

In this section, the number of expected nodes for each block is fixed to be 200, i.e. $n = 200k$, \pps estimates $\widehat B^{nn} = \frac{1}{n^2}\left[Z^TAZ\right]_+$ in the second step and then applies \texttt{SparseNJ} with a threshold suggested in Section \ref{sec:sparse_nj}.

\begin{figure}[!ht]
    \centering
	\includegraphics[width=6.5in]{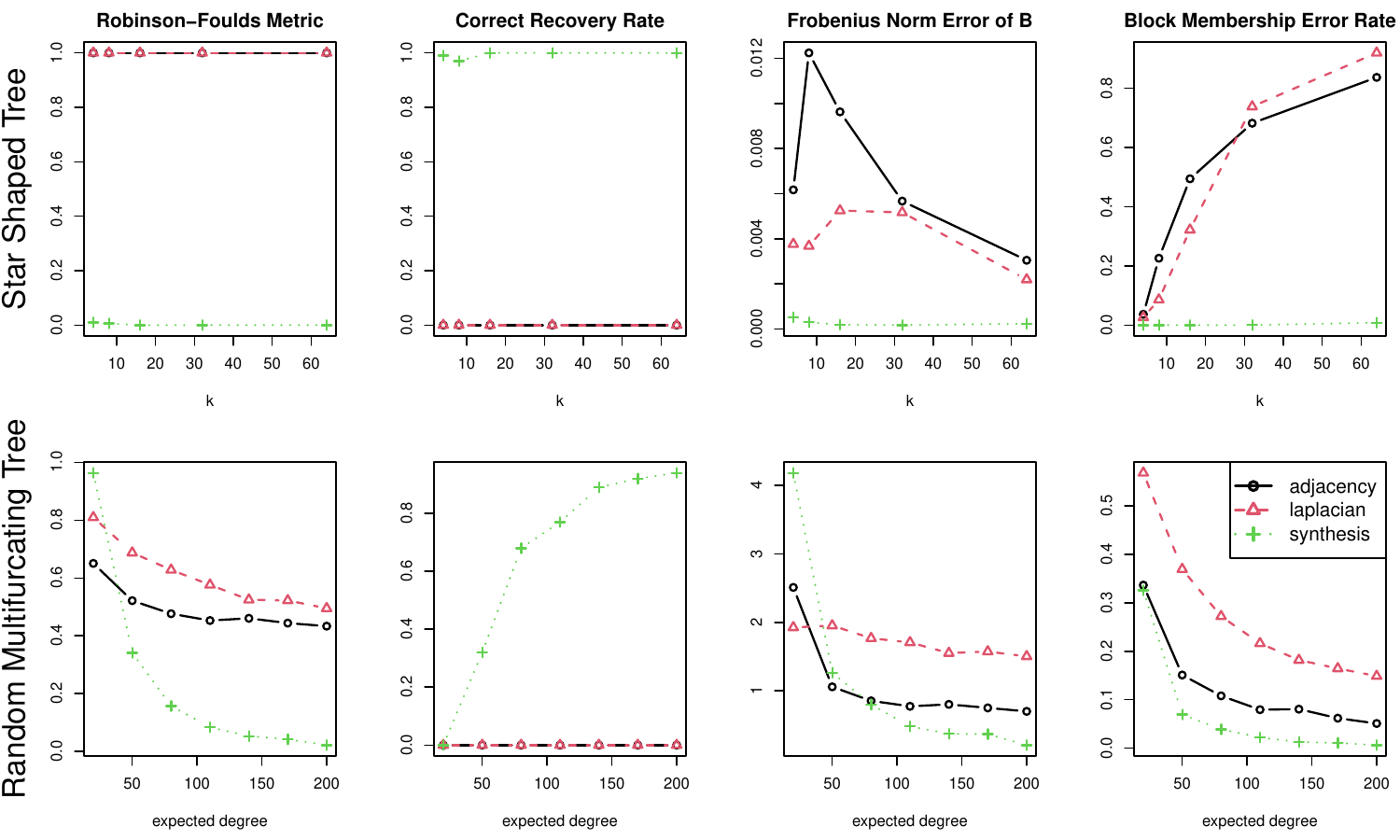}
   \caption{Performance of the Adjacency recursive bi-partition, the Laplacian  recursive bi-partition, and the \pps under star-shaped trees and random multifurcating trees.}
   \label{fig:simu_non_binary}
\end{figure}

Figure \ref{fig:simu_non_binary} presents the simulation results for both star-shaped trees and random multifurcating trees. Since recursive bi-partition algorithms output only binary trees, their correct recovery rate for these non-binary tree structures is always zero. Nevertheless, we report this measure to observe the performance of the \pps algorithm. 

For star-shaped trees, the \pps algorithm correctly recovers the tree structure in almost every replication. Importantly, its performance remains stable as $k$ increases. In contrast, both bi-partition algorithms present increasing bias as $k$ increases, especially for the block membership estimation. \cite{li} also presents simulations for this structure, but with a different interpretation. Li understands star-shaped trees as networks with no hierarchical communities because all between-community probabilities are the same. With the formulation of $\T$-Stochastic Graphs and its equivalence models, we provide a different understanding of this simple structure: a $\T$-Graphical Block model with $\T_z$ being a star-shaped tree. The number of blocks is set to be $k = 4, 8, 16, 32, 64$ intentionally to match the simulation in \citep{li}. Li points out that all methods fail completely with $k > 16$; whereas we show that \pps exhibit satisfying performance even when $k$ reaches 64. 

Regarding random multifurcating trees, \pps shows improving accuracy as the expected degree increases and can almost recover all trees correctly when the expected degree reaches 200. It produces a low Robinson-Foulds distance and block membership error rates even when the expected degree is low. Both bi-partition algorithms fail in this case as they are designed for ultrametric and binary structures. As the tree becomes more heterogeneous and multifurcates more, the \pps algorithm's advantages become more significant.

\newpage
\section{Dealing with Real-world Data}

\subsection{How to deal with asymmetric adjacency matrix}\label{appendix:asymmetric}
Real-world networks have different formats and only a small percentage of them are undirected with a symmetric adjacency matrix. As a matter of fact, an asymmetric adjacency matrix contains more information than a symmetric one, the question is how to apply \pps to it.

Due to the exchangeability of distance, i.e. $d(i, j) = d(j, i)$, \pps is discussed under the symmetric setting. However, it's not hard to adapt it for asymmetric cases. Let $A$ be the original asymmetric adjacency matrix, the most straightforward way is to ignore the direction of edges, and create a new adjacency $A^{undir} = A + A^T$. 

If one wants to distinguish between incoming edges and outgoing edges, another way is to define $A^{left} = AA^T$ and $A^{right} = A^TA$. Applying \pps to $A^{left}$ and $A^{right}$ can reveal two different latent structures in the network. Take the Wikipedia hyperlink data as an example, where $A_{ij} = 1$ implies a hyperlink from page $i$ to page $j$. In this case, $A^{left}_{ij}$ is the number of pages that both page $i$ and page $j$ link to. This presents the linking pattern of pages, two pages are considered to have similar patterns if they link to the same page. On the other hand, $A^{right}$ shows the being linked pattern, the similarity between two pages is measured by the number of pages that link to both of them. 

\subsection{How to name leaf nodes}\label{appendix:bff}

Reconstructing a tree without proper leaf node labels is meaningless if the goal is to interpret and understand the dataset. Since each leaf node corresponds to one column in $\widehat Z$, labels for leaf nodes can be obtained by interpreting factors reconstructed in  $\widehat Z$.

Let $x$ be a vector of text, where $x_i$ is some text that describes the vertex corresponding to row $i$ in $A$\footnote{Here $A$ denotes the input matrix of \pps in general, it could be $A^{undir}$, $A^{left}$, or $A^{right}$ as described in Appendix \ref{appendix:asymmetric}}. Then $\widehat Z_{ij}$ can be considered as a measure of how much $x_i$ can describe factor $j$.  

The simplest way to name leaf node $j$ is to choose $x_i$ such that $i = \underset{i}{\mathrm{argmax}} \ \widehat Z_{ij}$. When text in $x$ is succinct, concise, and representative, this strategy is efficient and creates meaningful labels. When the text is long, such as a sentence or even a paragraph, a better way is to extract the most important word. Let $X$ be a document-term matrix formed by $x$ such that $X_{il} = 1$ if $x_i$ contains word $l$. The naive way to measure the importance of word $l$ for factor $j$ is to calculate the correlation between the $l$th column of $X$ and the $j$th column of $\widehat Z$. Due to the sparse and heterogeneous nature of $\widehat Z$ and $X$, simple correlations are unstable. The ``best feature function'' (\texttt{bff}) \citep{10.1214/16-AOAS977} shows better performance in real data analysis. It measures the importance of word $l$ for factor $j$ as

\[\text{\bff}(j, l) = \sqrt{\frac{\sum_{i\in \text{in}(j)} \widehat Z_{ij}X_{ij}}{\sum_{i\in \text{in}(j)} \widehat Z_{ij}}} - \sqrt{\frac{\sum_{i\in \text{out}(j)} \widehat Z_{ij}X_{ij}}{\left|\text{out(j)}\right|}},\]
where sets $\text{in}(j)$ and $\text{out}(j)$ for block $j$ are defined as $\text{in}(j) = \{i: \widehat Z_{ij}\geq 0\}$, $\text{out}(j) = \{i: \widehat Z_{ij}<0\}$.

\subsection{Linking Pattern of the Wikipedia page link network}\label{appendix:wiki_Z}

Figure \ref{fig:wiki_Z} presents the linking pattern of the wikipedia page link network. This tree presents similar clustering structure as the being linked pattern tree (Figure \ref{fig:wiki_beinglinked}).

\begin{figure}[h]
    \centering
    \includegraphics[width = 5in]{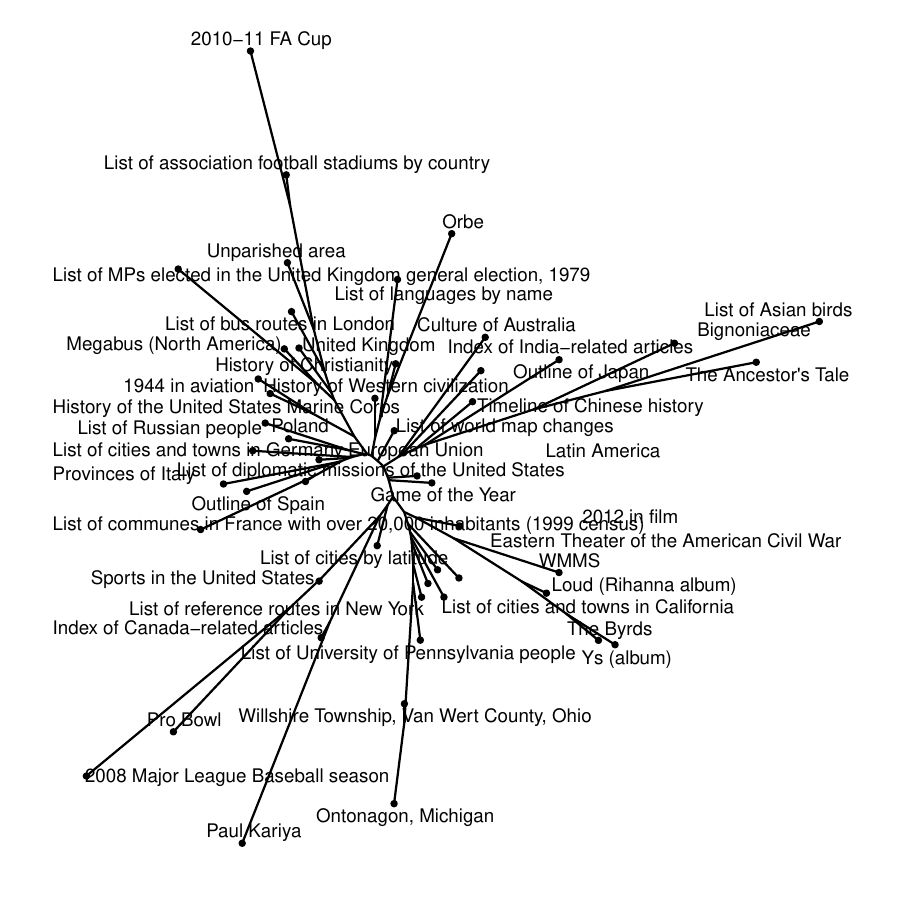}
    \vspace{-0.3in}
   \caption{Tree recovered from the linking pattern of Wikipedia pages. There are many ``list of something'' pages in the tree. Though the structure is not as clear as the being linked tree, there are still similar patterns. The bottom half of the tree is the North American cluster and the upper half of the tree is the European cluster, the Asian cluster, and the animal/plant cluster are mixed together on the right of the tree.}

   \label{fig:wiki_Z}
\end{figure}

\clearpage

\section{Robustness with respect to the selection of $k$}\label{appendix:robust_k}

The consistency of step 1 (\texttt{vsp}) in \pps relies on two things: $k$ remains constant as $n$ increases and the correct $k$ is chosen as input. However, \pps exhibits satisfying performance in both simulation and real data examples when these two conditions are violated. 

\subsection{Robustness to $k$ in simulated data: balanced binary tree and non-balanced tree}
This section briefly discusses the robustness of \pps by conducting simulations of $\T$ being a binary balanced tree and non-balanced tree in the  $\T$-Stochastic Graph.

Consider a $\T$-Stochastic Graph with $\T$ being a $n$ leaf node binary balanced tree, then in the equivalent $\T$-Graphical Blockmodel, the number of blocks $k = n/2$, and there are only two nodes in each block. This is an example of $k$ increasing at the same rate as $n$. In this situation, it is not plausible to choose the true $k$ as the input. However, \pps with an underestimated $k$ as input can still recover some top-level structures of the tree.

\begin{figure}[!ht]
    \centering
	\includegraphics[width=6in]{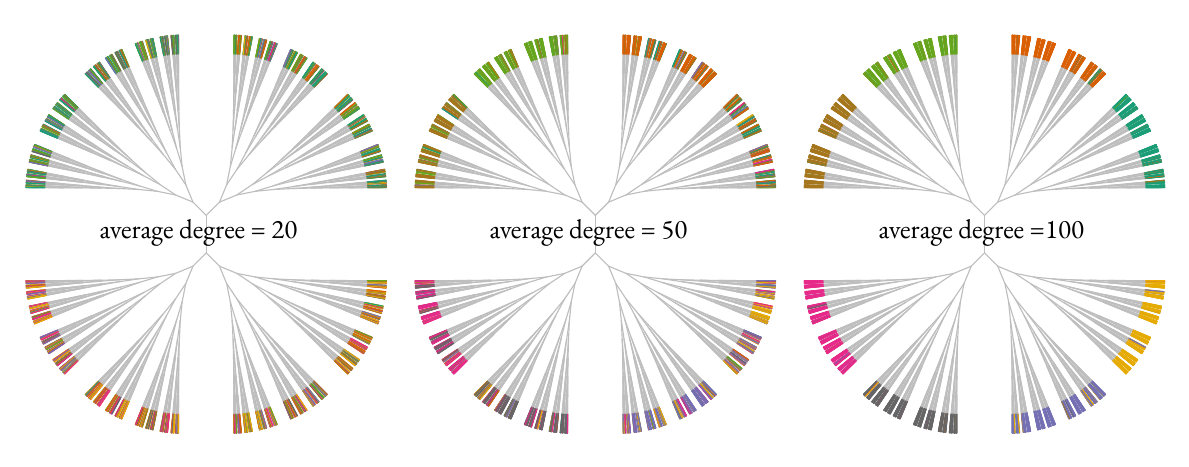}
   \caption{Performance of the first step algorithm on a balanced binary tree with $n = 2048$ leaf nodes. The true $k = 1024$ in this case and we test our algorithm with $k = 8$, hoping to recover the top structure of the tree. Leaf nodes are colored by their estimated clusters. As the average degree goes higher, the leaf nodes get finer clusters}
   \label{fig:robust_k_balance}
\end{figure}

Figure \ref{fig:robust_k_balance} plots the clustering result of the first step algorithm when $n = 2048$ (true $k = 1024$), and we choose $k = 8$ as the input. The true underlying tree structure (balanced binary tree) is plotted and leaf nodes are colored by the clustering result. Since we choose $k = 8$, we would hope the algorithm to recover the first 3 levels of the tree. That is, descendants of the same nodes in the 3rd level should be clustered into the same block. To explore how the performance is affected by the average degree, we plot the results of average degree = 20, 50, and 100 from left to right. We can observe that the first-step clustering results are still trustworthy with an underestimated $k$, and as the average degree increases, the clusters become more refined.

This tree example is relatively easy and clear due to its balanced and symmetric properties. Under more general cases where $\T$ is a random tree, \pps still gives reasonable clustering results. 

\begin{figure}[!ht]
\centering
\subfloat[]{%
  \includegraphics[clip,width=5.5in]{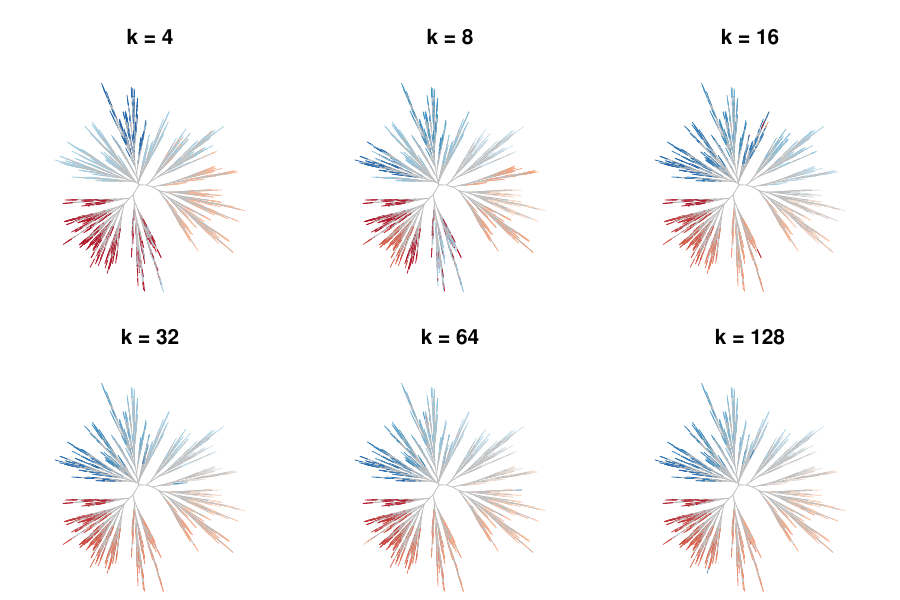}%
}

\subfloat[]{%
  \includegraphics[clip,width=5.5in]{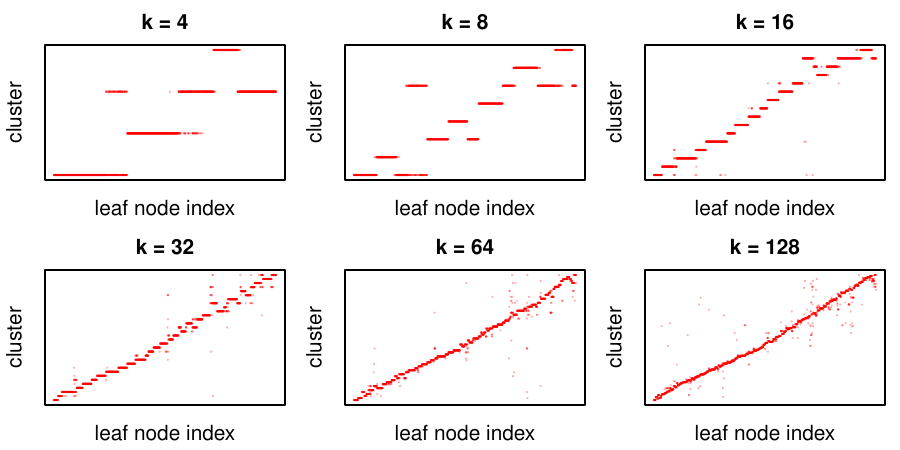}%
}
\caption{Performance of the first step algorithm on an unbalanced binary tree with $n = 2048$ leaf nodes. We test our algorithm with $k = 4, 8, 16, 32, 64, 128$. In Subfigure (a), leaf nodes are colored by their estimated clusters, In Subfigure (b), leaf nodes are plotted against their estimated clusters. As $k$ increases, the clusters get more refined.}
\label{fig:robustness_nonbalanced}
\end{figure}

Figure \ref{fig:robustness_nonbalanced} presents the result of a randomly generated tree with $n = 2048$ leaf nodes. In Subfigure (a), the true underlying tree shape is plotted and leaf nodes are colored by their cluster. A good estimation result should present similar color for leaf nodes that are close. In Subfigure (b), the estimated clusters are plotted versus the leaf node index. Still, a good result should cluster nodes with close indexes together. When $k = 4$, there is some inconsistency in the tree presented in Subfigure (a). In the upper half of the tree, the light blue color is cut into two parts by the dark blue color; in the lower half of the tree, the red is mixed with the light blue color. This light blue color corresponds to the third cluster in Subfigure (b), and we can observe that this cluster contains three discontinuous parts of the tree. Similar inconsistency can be observed when $k = 8$, and the result gets better as $k$ increases.

\subsection{Robustness to $k$ in real data: Journal Citation network 
}\label{sec:real_data_journal}

In this section, we examine a journal citation network derived from the Semantic Scholar dataset \citep{ammar-etal-2018-construction}, which comprises over 200 million academic publications. To analyze the journal citation patterns, we aggregate the paper citation network by journals and identify unique journal names by converting all letters to lowercase and removing punctuation. To simplify computation, we randomly sampled 5\% of the paper citations and set $A_{ij} = 1$ if there were more than five citations (in this random sample) from papers in journal $i$ to papers in journal $j$. After sampling and thresholding, the original dataset containing approximately 100,000 journals was reduced to 22,688.

While $A$ is a square matrix, it is not a symmetric matrix. We deal with this by applying \pps to $A^{undir} = A^T + A$. More details about dealing with asymmetric matrices can be found in Appendix \ref{appendix:asymmetric}. Due to the heterogeneous row and column sums, we also apply the degree-normalization and renormalization steps proposed in \citep{varimax_rohe}. All trees in this section are estimated with \(\widehat B^{nn} = \widehat Z_+^TA\widehat Z_+\) in the first step and a simple NJ algorithm in the second step. Leaf nodes are labeled with the most important word measured by \bff (see Appendix \ref{appendix:bff} for more details about \bff).

An important question here, and also in any other real data analysis, is how do we choose $k$? In applied multivariate statistics, estimating the number of latent dimensions is a fundamental problem. Most solutions try to find $k$ by locating a gap between eigenvalues, the widely used scree plot technique \citep{cattell1966scree} is a representative example. However, the real dataset hardly presents an eigengap. Figure \ref{fig:eigenvale} presents the eigenvalues of the degree normalized adjacency matrix $L$ for this dataset.

\begin{figure}[!ht]
    \centering
	\includegraphics[width=3in]{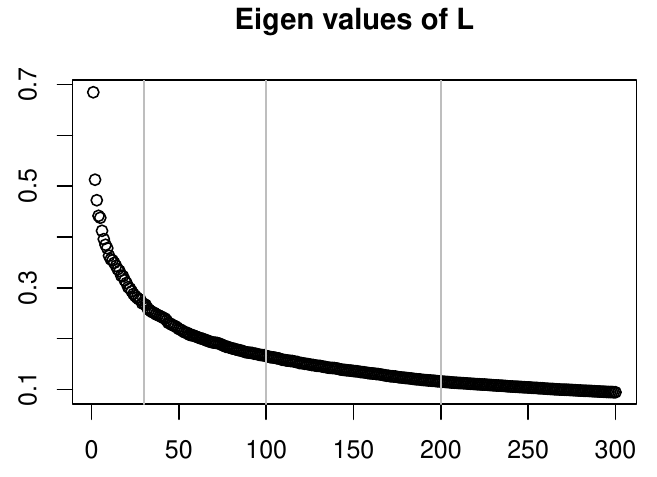}
   \caption{The first 300 eigenvalues of the normalized adjacency matrix for the journal citation network, with lines at $k = 30, 100, 200$. No clear eigengap can be observed.}
   \label{fig:eigenvale}
\end{figure}

With the $\T$-Stochastic graph and the \pps algorithm proposed, another way to think about this problem is to ask: are tree structures robust with respect to different choices of $k$? An ideal case is that as $k$ increases, the algorithm outputs more refined clusters and detailed tree structures. Though with some disturbance, we find that the general structure is stable across different selections of $k$ in this journal citation network. 

\begin{figure}[!ht]
    \centering
	\includegraphics[width=4.5in]{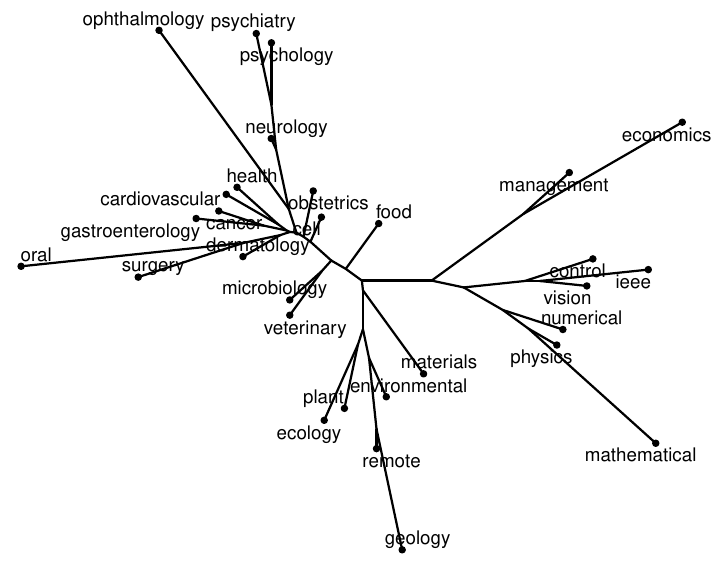}
   \caption{Journal citation tree with $k = 30$. This tree presents similar clustering structures as the journal citation tree with $k = 100$ (Figure \ref{fig:journal_tree}).}
   \label{fig:journal_k=30}
\end{figure}

Figure \ref{fig:journal_tree} plots the tree recovered from this dataset with $k = 100$. While there are more detailed group structures and connections, in general, the tree can be summarized as large clusters about \textit{social science (economics, management)}, \textit{math}, \textit{engineering}, \textit{earth science}, \textit{psychology}, and \textit{medical science}. Figure \ref{fig:journal_k=30} plots the tree recovered with $k = 30$, this smaller tree presents similar clusters and structures. The same pattern can be observed in journal citation trees with $k = 70$ and $k = 200$, presented in Figure \ref{fig:journal_k = 200} and \ref{fig:journal_k=70}, accordingly. Additionally, Figure \ref{fig:journal_k=70} also presents the detection of potentially non-tree-like structures.
\begin{figure}[!ht]
    \centering
   \textbf{Comparing the estimated tree distance matrix and the \texttt{TSGdist} estimation is helpful in detecting non-tree like structures}
   \vspace{.2in}
	\includegraphics[width=6in]{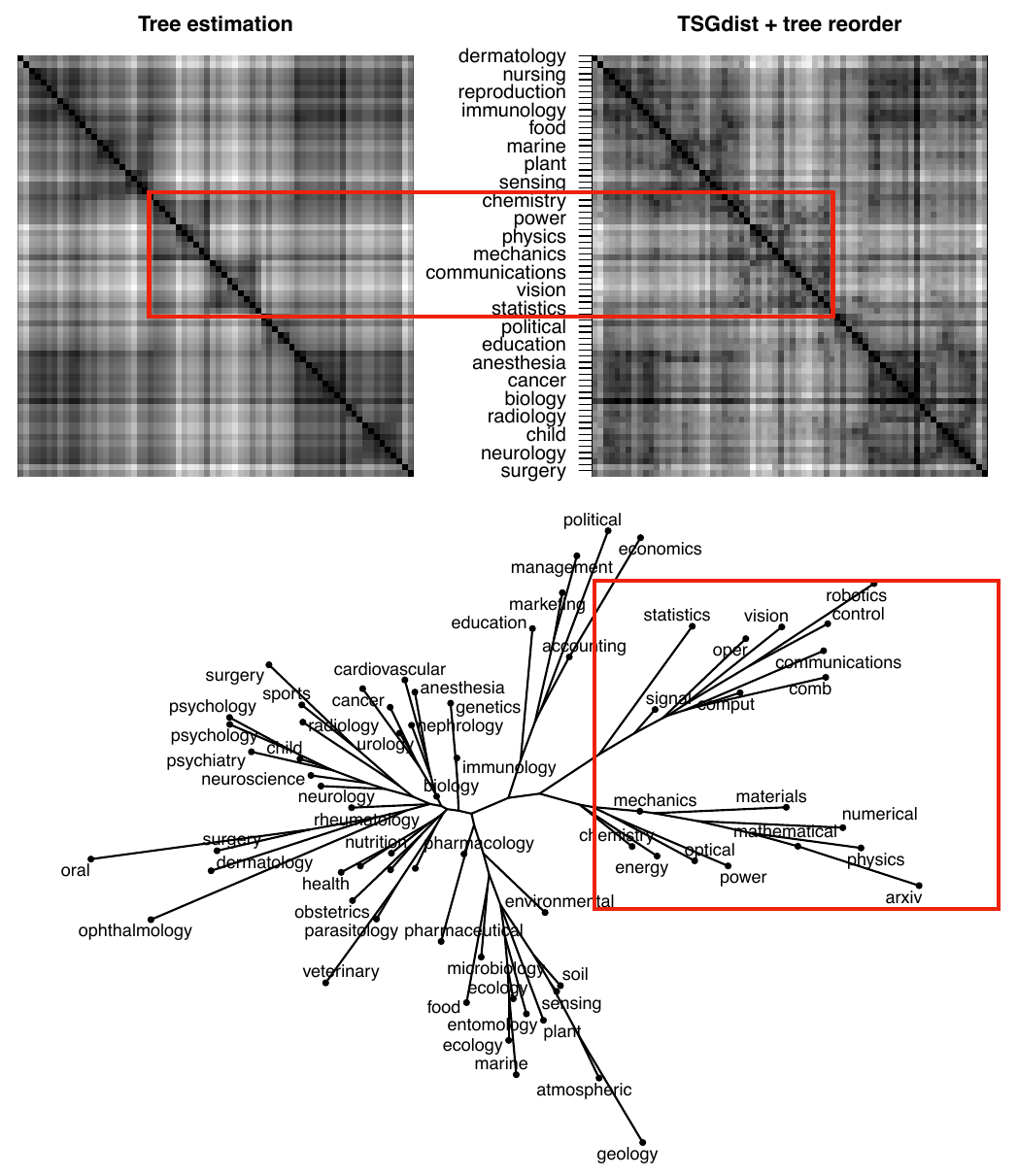}
	\vspace{-.2in}
   \caption{This is the Journal citation tree with $k = 70$ and its diagnosis. There are notable differences between the tree estimation and the \texttt{TSGdist} estimation around the \textit{engineering} cluster and the \textit{physics/material} cluster. Such differences usually indicate the existence of some non-tree-like structures that the neighbor-joining algorithm can't reconstruct. We conjecture that this is caused by the wide cross-disciplinary application of \textit{math/statistics}. The \textit{statistics/combinatory} factors are joined together with the \textit{engineering} cluster whereas the \textit{mathematical/arxiv} factors are combined with the \textit{physics/material} clusters.}
   \label{fig:journal_k=70}
\end{figure}

\begin{figure}[!ht]
    \centering
	\includegraphics[width=8in,angle=270]{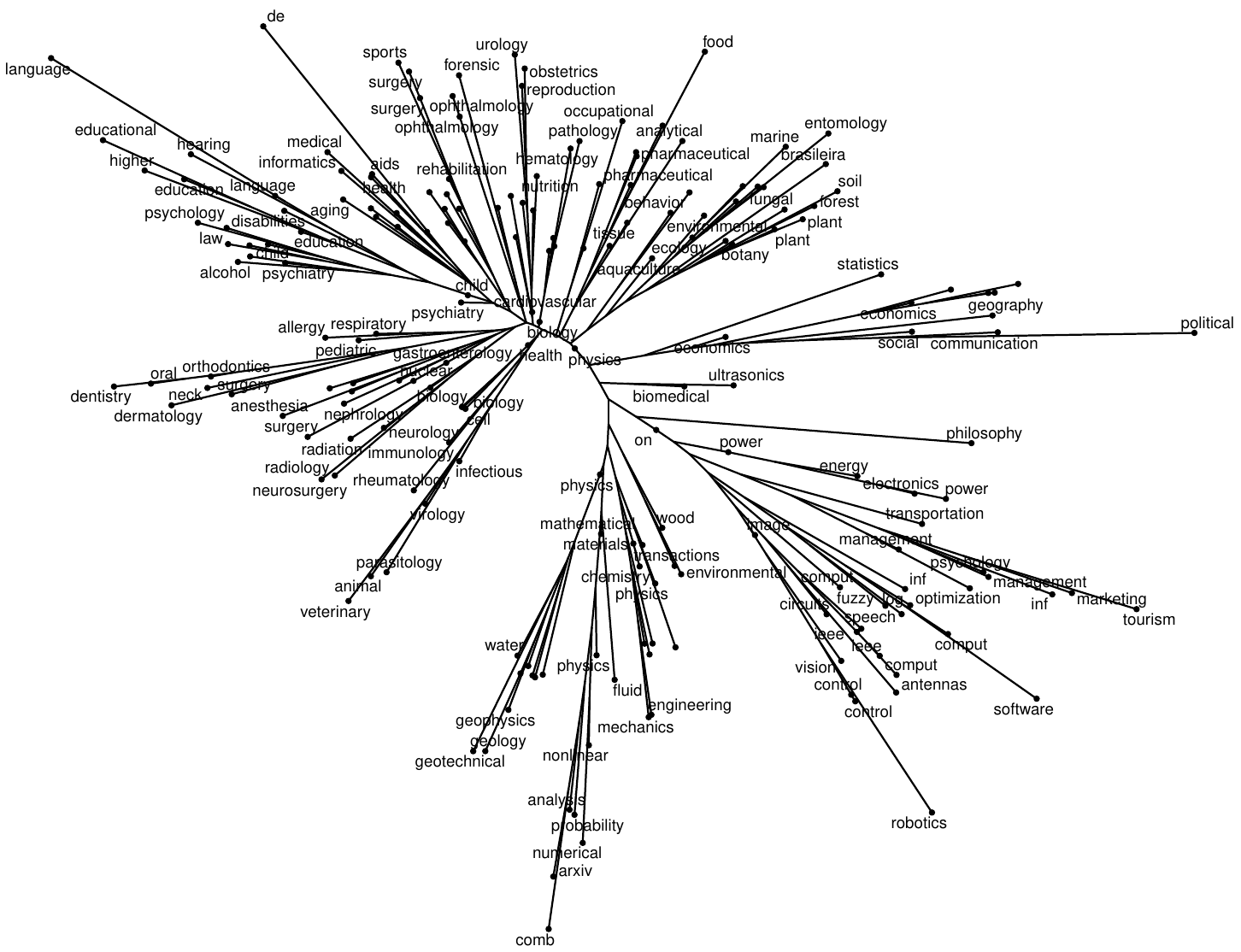}
   \caption{This figure presents the journal citation tree with $k = 200$. By comparing to Figure \ref{fig:journal_k=30}, \ref{fig:journal_tree}, and \ref{fig:journal_k = 200}, we can observe that recovered journal citation trees with different $k$ have similar structures.}
   \label{fig:journal_k = 200}
\end{figure}

\clearpage

\section{Proof for equivalence between six alternative models}

In a tree graph $\T = (V,E)$, if we take any two nodes $p, q \in V$, then there is only one path that connects them (without backtracking). Denote the set of nodes on this path, including nodes $p$ and $q$, as $V(p, q)$, and the set of edges along this path as $E(p, q)$.

Though the following sections, we use \(\lambda_{ij}^{TSG}, \ \lambda_{ij}^{GB}, \ \lambda_{ij}^{OB}\) to represent $\lambda_{ij}$'s in $\T$-Stochastic Graphs, $\T$-Graphical Blockmodels, and $\T_r$-Overlapping Blockmodel respectively. And denote the corresponding population adjacency matrix as \(\mathscr{A}^{TSG}, \mathscr{A}^{GB}, \mathscr{A}^{OB}\).

\subsection{$\T$-Graphical Blockmodel: Proof of Theorem \ref{thm:tsg_tgb}, Proposition \ref{prop:cov->dist} and \ref{prop:dist->cov} }\label{appendix:tsg_tgb}

The proof of Theorem \ref{thm:tsg_tgb} involves two directions:
\begin{enumerate}
    \item proving every $\T$-Graphical Blockmodel is a $\T$-Stochastic Graph,
    \item proving every $\T$-Stochastic Graph is a $\T$-Graphical Blockmodel.
\end{enumerate}
These two parts are established separately in the following sections (Appendix \ref{appendix:tgb->tsg} and Appendix \ref{appendix:tsg->tgb}). To demonstrate the first part, Proposition \ref{prop:cov->dist} is utilized, while the second part is shown using Proposition \ref{prop:dist->cov}. The proof of Proposition \ref{prop:cov->dist} and Proposition \ref{prop:dist->cov} can be found in Appendix \ref{appendix:cov->dist} and Appendix \ref{appendix:dist->cov}, respectively.

\subsubsection{Every $\T$-Graphical Blockmodel is a $\T$-Stochastic Graph}\label{appendix:tgb->tsg}

Given any $\T$-Graphical Blockmodel, construct a $\T$-Stochastic Graph with edge weight $w_{uv}$ such that 
\begin{enumerate}
    \item $w_{uv} = -\log(\vert \Sigma_{uv}\vert)$ for any pair of internal nodes $u, v$,
    \item $w_{i, p(i)} = -\log(\theta_i)$ for any leaf node $i\in V_\ell$ and its parent $p(i)$.
\end{enumerate}
It's easy to check that $w_{uv}$ is positive for all internal edges $(u, v)$ as $\vert \Sigma_{uv} \vert \in (0, 1)$. The proof is complete by
\begin{align*}
 \lambda^{TSG}_{ij} &= \exp(-d(i, j))  \\
 & = \exp\left(-w_{i, p(i)}\right)\exp\left(-w_{i, p(i)}\right) \exp(-d(p(i), p(j))) \\
 & = \theta_i \theta_j \exp(-d(p(i), p(j)))\\
 & = \theta_i \theta_j \exp\left(\sum_{(u, v)\in E(p(i), p(j))}-\log\left(\vert \Sigma_{uv} \vert\right)\right)  \quad \text{(by the construction above)}  \\
 & = \theta_i \theta_j \vert \Sigma_{p(i), p(j)}\vert \quad \text{(by Proposition \ref{prop:cov->dist})}\\
 & = \theta_i \theta_j \vert \Sigma_{z(i), z(j)}\vert = \theta_i \theta_j B_{z(i), z(j)}  \quad \text{(by Definition \ref{def:tgb})}\\
 & = \lambda^{GB}_{ij}
\end{align*}

\begin{remark}\label{appendix_rmk:multiply_c}
In the above proof for Theorem \ref{thm:tsg_tgb}, we construct edge weight on $\T$ such that $w_{i, p(i)} = log(1/\theta_i)$ for leaf node $i$. This means that when $\theta_i > 1$, the edge weight between leaf nodes $i$ and its parent is negative. In fact, one can multiply all $\theta_i$ by a constant $\sqrt{c}$ such that $$\lambda_{ij}^{GB} = c\,\theta_i \theta_j B_{z(i), z(j)} = c\exp(-d(i, j)),$$
that explains why considering $\lambda_{ij}^{TSG} = c \exp(-d(i, j))$ is equivalent to allowing $w_{i, p(i)}<0$ for any leaf node $i$.
\end{remark}

\subsubsection{Every $\T$-Stochastic Graph is a $\T$-Graphical Blockmodel}\label{appendix:tsg->tgb}

This part also serves as a complete proof for Theorem \ref{thm:tsg->dcsbm} as every $\T$-Graphical Blockmodel is a DCSBM by definition. Specifically, the matrix $B$ defined in Section \ref{sec:tsg->sbm} is a principal submatrix of the matrix $\Sigma$ defined in Proposition \ref{prop:cov->dist} and therefore is full rank.

To show that every $\T$-Stochastic Graph is a $\T$-Graphical Blockmodel, consider any $\T$-Stochastic Graph with additive distance $d(\cdot, \cdot)$. Define $\Sigma$ as
\begin{equation}\label{eq:construct_sigma}
    \Sigma_{ij} = \exp(-d(i, j)).
\end{equation}
Then by Proposition \ref{prop:dist->cov}, $\Sigma$ is positive definite, and the multivariate Gaussian distribution with covariance matrix $\Sigma$ is a GGM on $\T$. Therefore, one can construct a $\T$-Graphical Blockmodel with $\Sigma$, following Definition \ref{def:tgb}. Furthermore, define $\theta_i = \exp(-d(i, p(i)))$ for any leaf node $i$ and its parent $p(i)$.

The proof is completed by 
\begin{align*}
\lambda_{ij}^{GB} &= \theta_i \theta_j B_{z(i), z(j)} = \theta_i \theta_j B_{p(i), p(j)} =  \theta_i \theta_j \left\vert \Sigma_{p(i), p(j)} \right\vert \quad \text{(by Definition \ref{def:tgb})}\\
& = \theta_i \theta_j \Sigma_{p(i), p(j)} \quad \text{(since $\Sigma_{p(i), p(j)}$ is positive by Equation \eqref{eq:construct_sigma})}\\
& = \theta_i \theta_j \exp(-d(p(i), p(j)))  \quad \text{(by Proposition \ref{prop:dist->cov})}\\
& = \exp(-d(i, p(i))) \exp(-d(j, p(j))) \exp(-d(p(i), p(j))) \quad \text{(by the construction above)}\\
& = \lambda_{ij}^{TSG}
\end{align*}

\subsubsection{Proof of Proposition \ref{prop:cov->dist}}\label{appendix:cov->dist}

Proposition \ref{prop:cov->dist} is a direct result of Lemma \ref{lemma:additive}, which is a rephrase of Proposition 3 in \citep{choi2011learning} and is first discussed in \citep{erdHos1999few}. Lemma \ref{lemma:additive} shows that under a Gaussian Graphical Model, a negative log transformation of correlations builds an additive distance. Since Definition \ref{def:ggm} assumes $\Sigma_{ii} = 1$, correlations are equal to covariance and thus the proof is complete. Discussions about non-unit variance ($\Sigma_{ii} \neq 1$) can be found in Appendix \ref{appendix:covariance}

\begin{lemma}\label{lemma:additive}
	(\cite{erdHos1999few}) If a Multivariate Gaussian distribution $\mathcal P(\cdot)$ is a GGM on tree $\T = (V,E)$, denote the correlation between any two nodes (not necessarily neighbors) $p, q \in V$ as $\rho_{pq}$ and define \[d(p, q) = -\log \left(\lvert \rho_{pq}\rvert\right),\]
 then $d(\cdot, \cdot)$ is an additive distance on $\T$ in the sense that:
	\begin{equation}\label{eq:additive}
	    d(p, q) = \sum_{(u,v)\in E(p, q)}d(u, v) , \quad \forall p,q \in V,
	\end{equation}
 where $E(p, q)$ is the set of all edges along the path between nodes $k$ and $l$.
\end{lemma}

\begin{remark}\label{rmk:differnt_def_additive}
    Although the definitions of additive distance in our paper and in \citep{erdHos1999few} are presented differently, they are actually equivalent. It's easy to check that distances define in our approach satisfy Equation \eqref{eq:additive}, and thus is an additive distance in \citep{erdHos1999few}. On the other had, given any $d(\cdot, \cdot)$ that is additive as defined in \citep{erdHos1999few} (satisfies Equation \eqref{eq:additive}), let edge weights $w_{uv} = d(u, v)$ for all edges $(u, v)\in E$, and define distance $d_w(\cdot, \cdot)$ using $\{w_{uv}\}$, then $d_w(p, q) = d(p, q)$ for any pair of nodes $p, q\in V$.
\end{remark}

\subsubsection{Proof of Proposition \ref{prop:dist->cov}}\label{appendix:dist->cov}

The proof of Proposition \ref{prop:dist->cov} requires the following two lemmas.

\begin{lemma}\label{lemma:semi_positive}
    Consider a tree $\T = (V, E)$ with $n$ leaf nodes and a set of non-negative edge weights $\{w_{uv} \geq 0\}$ (we allow $w_{uv} = 0$ in this lemma), define $\Sigma^\ell \in \R^{n\times n}$ as 
    \[\Sigma^\ell_{ij} = \exp(-d(i, j)),\]
    for any pair of leaf nodes $i, j \in V_\ell$. Then $\Sigma^\ell$ is positive semidefinite.
\end{lemma}

\begin{lemma}\label{lemma:additive-reverse}
A multivariate Gaussian distribution $N(0, \Sigma)$ is local Markov\footnote{The definition of local Markov can be found in Appendix \ref{appendix:additive-reverse}} on tree $\T = (V, E)$ if the correlations between any pair of nodes $p,q \in V$ are multiplicative:
\begin{equation}\label{eq:multiplicative}
    \rho_{pq}  = \prod_{(u,v)\in E(p, q)} \rho_{uv}, \quad \forall p,q \in V.
\end{equation}
\end{lemma}

To prove Proposition \ref{prop:dist->cov}, consider an additive distance $d(\cdot, \cdot)$ on $\T = (V, E)$. Without loss of generality, assume $V = [N]$ and the set of leaf nodes $V_\ell = [n]$. Define the matrix $\Sigma \in \R^{N\times N}$ as 
\[\Sigma_{ij} = \exp(-d(i, j))\]
for any pair of nodes $i, j\in V$. Construct a new tree $\widetilde \T = \left(\widetilde V, \widetilde E\right)$ by attaching twigs with zero edge weight to all internal nodes in $V$, that is
\[\widetilde V = V \cup \{N+1, N+2, \cdots, 2N - n\}, \ \widetilde E = E \cup \{(n+i, N+i): i\in [N - n]\}\]
\[\widetilde w_{u, v} = 
\begin{cases}
    w_{u, v} & \text{if $u, v\in [N]$}\\
    0 & \text{otherwise}
\end{cases}.
\]
Define the matrix $\widetilde \Sigma^\ell \in \R^{N\times N}$ as 
\[\widetilde \Sigma^\ell_{ij} = \exp\left(-\widetilde d(i, j)\right)\] for any pair of leaf nodes $i, j\in \widetilde V_{\ell}$, then $\widetilde \Sigma^\ell = \Sigma$. By Lemma \ref{lemma:semi_positive}, $\widetilde \Sigma^\ell$ is positive semidefinite, therefore $\Sigma$ is positive semidefinite.

Since $\Sigma$ is positive semidefinite, we can define a multivariate Gaussian distribution $\mathcal{P}(\cdot)$ with mean zero and variance $\Sigma$. Then by Lemma \ref{lemma:additive-reverse}, $\mathcal{P}(\cdot)$ is local Markov on $\T$. 

What is left is to show that $\Sigma$ is invertible. This is because when $\Sigma$ is invertible, the joint probability density function is continuous and positive, thus the local Markov property implies the global Markov property. Together with the fact that zero entry in the inverse covariance matrix of a joint Gaussian distribution implies conditional independence, the proof is completed.

The proof of $\Sigma$ being invertible is by contradiction. Since $\Sigma$ is the covariance matrix of a Gaussian distribution, if $\Sigma$ is not full rank, then there exists a random variable $X$ that can be represented as a linear combination of others. Given the local Markov property, $X$ is independent of all other nodes given its neighbor nodes. Therefore $X$ must be a linear combination of its neighbor nodes. Denote the neighbors of $X$ as $Z_1, Z_2, \cdots, Z_m$, then there exists $c = (c_1, \cdots, c_m)^T$ such that \begin{equation}\label{eq:variables_linear_relationship}
    X = c_1Z_1 + c_2Z_2\cdots + c_m Z_m,
\end{equation}
Denote correlations $\rho_{X, Z_i}$ as $z_i$, then 
\begin{equation}\label{eq:sigma_representation}
    Cov(X, Z_1, \cdots, Z_m) = 
    \begin{bmatrix}
    1 & z_1 & z_2 & \cdots & z_m\\
    z_1 & 1 & z_1z_2 &\cdots & z_1z_m\\
    z_2 & z_1z_2& 1 & \cdots & z_2z_m\\
    \vdots & \vdots &\vdots & \ddots & \vdots \\
    z_m & z_1z_m & z_2z_m &\cdots &1
    \end{bmatrix} = \begin{bmatrix}
        1 & z^T \\
        z & \Gamma \\
    \end{bmatrix}
\end{equation}
where $z = (z_1, \cdots, z_m)^T$, $\Gamma = Cov(Z_1, \cdots, Z_m)$. 
Notice that 
\begin{equation}\label{eq:gamma_decompose}
    \Gamma = {\underbrace {\textstyle \text{diag}(1-z_1, \cdots, 1-z_m)}_{\displaystyle \mathclap{D}}} + zz^T,
\end{equation}
 where $\text{diag}(x)$ is a diagonal matrix with its $i$th diagonal elements being $x_i$. 

By the Sherman Morrison formula \citep{bartlett1951inverse} and the decomposition in Equation \eqref{eq:gamma_decompose}, $\Gamma$ is invertible and 
\begin{equation}\label{eq:gamma_inverse}
   \Gamma^{-1} = D^{-1} - \frac{D^{-1}zz^TD^{-1}}{1+z^T D^{-1} z}. 
\end{equation}
Other than this, Equation \eqref{eq:variables_linear_relationship} and \eqref{eq:sigma_representation} implies 
\begin{equation}\label{eq:linear_in_sigma}
    \begin{cases} 
    z = \Gamma c, \\
    1 = z^Tc
    \end{cases}.
\end{equation}
By Equation \eqref{eq:gamma_inverse} and \eqref{eq:linear_in_sigma}, 
\begin{equation}\label{eq:contradiction}
    1 = z^T \Gamma^{-1} z.
\end{equation}
It's easy to calculate 
\begin{equation}
    \left[\Gamma^{-1}\right]_{ij} = 
   \begin{dcases}
    \frac{1}{1-z_i^2} - \left. z_i^2 \middle/ \left(\left(1-z_i^2\right)^2\left(1+\sum_{i = 1}^m \frac{z_i^2}{1-z_i^2}\right)\right) \right.,  & i = j\\
    - \left. z_iz_j \middle/ \left(\left(1-z_i^2\right) \left(1-z_j^2\right)\left(1+\sum_{i = 1}^m \frac{z_i^2}{1-z_i^2}\right)\right) \right., & i\neq j\\
    \end{dcases},
\end{equation}
let $C_i = \dfrac{z_i^2}{1-z_i^2}$, then the right side of Equation \eqref{eq:contradiction} does not equal to its left side:
\begin{align*}
z^T \Gamma^{-1} z & = \sum_{i = 1}^m \left (C_i - \dfrac{C_i^2}{1+\sum_{\ell = 1}^mC_\ell}\right) - 2\sum_{i\neq j}\dfrac{C_iC_j}{1+\sum_{\ell = 1}^m C_\ell} \\
& = \dfrac{1}{1+\sum_{\ell = 1}^m C_\ell}\left(\left(1+\sum_{i = 1}^mC_i\right)\sum_{i = 1}^mC_i -\sum_{i = 1}^m C_i^2 - 2\sum_{i\neq j}C_iC_j\right)\\
& = \dfrac{\sum_{\ell = 1}^m C_\ell}{1+\sum_{\ell = 1}^m C_\ell}\neq 1
\end{align*}

\subsubsection{Proof of Lemma \ref{lemma:additive-reverse}}\label{appendix:additive-reverse}

\begin{definition}\label{def:local_markov}
A probability distribution $\mathcal{P}(\cdot)$ is local Markov with respect to $\T$ if the conditional distribution of any node $X_i$ given its neighbors is independent of the remaining nodes, i.e. $X_i\indep X_{V/\{i\cup nbd(i)\}}| X_{nbd(i)}$. 
\end{definition}

 Consider $X$ be any node in tree $\T$, let $Y$ be one of its non-neighbor nodes, and let $Z_0,Z_1,\cdots,Z_m$ be the neighbor nodes of $X$. By definition, the shortest path between $X$ and $Y$ must go through some $Z_i$, w.l.o.g. assume the path between $X$ and $Y$ go through $Z_0$. Denote $\rho_{X,Z_0} = x, \rho_{Y,Z_0} = y, \rho_{X,Z_i} = z_i, i = 1,\cdots,m$.

We first consider the case where $X,Y,Z_0,\cdots,Z_m$ all have unit variance, then  
\begin{equation*}
    Cov(Z_0,Y,X,Z_1,\cdots,Z_m) = \begin{bmatrix}
    1 & y & x & xz_1 & \cdots & xz_m\\
    y & 1 & xy & xyz_1 & \cdots & xyz_m\\
    x & xy & 1 & z_1 &\cdots & z_m\\
    xz_1 & xyz_1 & z_1& 1 & \cdots & z_1z_m\\
    \vdots & \vdots & \vdots &\vdots & \ddots & \vdots \\
    xz_m & xyz_m & z_m & z_1z_m &\cdots &1
    \end{bmatrix}
\end{equation*}.

Let $W_1 = (Y,X,Z_1,\cdots,Z_m)$, $W_2 = Z_0$, then
\begin{align*}
\Sigma_{1|2} &= \Sigma_{11}-\Sigma_{12}\Sigma_{22}^{-1}\Sigma_{21}\\
& = \begin{bmatrix}
    1 & xy & xyz_1 & \cdots & xyz_m\\
    xy & 1 & z_1 &\cdots & z_m\\
    xyz_1 & z_1& 1 & \cdots & z_1z_m\\
    \vdots & \vdots &\vdots & \ddots & \vdots \\
    xyz_m & z_m & z_1z_m &\cdots &1
    \end{bmatrix}-
    \begin{bmatrix}
    y\\
    x\\
    xz_1\\
    \vdots\\
    xz_m
    \end{bmatrix}
    \begin{bmatrix}
     y & x & xz_1 & \cdots & xz_m
    \end{bmatrix}\\
& = \begin{bmatrix}
     1-y^2 & \mathbf{0}_{1\times m}\\
     \mathbf{0}_{m\times 1} & A_{m\times m}
    \end{bmatrix} (A\text{ is }Cov(X,Z_1,\cdots,Z_m| Z_0))
\end{align*}

Since $(Z_0,Y,X,Z_1,\cdots,Z_m)$ follows multivariate normal distribution, a zero covariance matrix implies independence, i.e. $Y\indep X, Z_1,\cdots,Z_m | Z_0$. By weak union rules of conditional independence, this implies $Y\indep X | Z_0, Z_1,\cdots,Z_m$. 

Therefore, for any non-neighbor nodes $Y$, $Cov(Y, X| Z_0,Z_1,\cdots Z_m) = 0$, denote all non-neighbor nodes as $Y_1,\cdots,Y_l$, then $Cov((Y_1,\cdots,Y_l),X|  Z_0,Z_1,\cdots Z_m) = \mathbf{0}_{l\times 1}$. Since all nodes follows multivariate normal distribution, this implies $X_i\indep X_{V/\{i\cup nbd(i)\}}| X_{nbd(i)}$ for any nodes $X_i$.

For variables $X,Y,Z_0,\cdots,Z_m$ that don't have unit variance, just consider $X^* = X/\sigma_X,Y^* = Y/\sigma_Y,Z_0^* = Z_0/\sigma_{Z_0},\cdots,Z_m^* = Z_m/\sigma_{Z_m}$, the conditional independence between $X^*$ and $Y^*$ given $Z_0^*,\cdots,Z_m^*$ implies the conditional independence between $X$ and $Y$ given $Z_0,\cdots,Z_m$.

\subsubsection{Proof of Lemma \ref{lemma:semi_positive}}

For notation simplicity, we drop the superscript $\ell$ on $\Sigma^\ell$ and use $\Sigma$ to represent the matrix defined for all leaf nodes.
The proof is established through induction on the number of internal nodes. 
\paragraph{Lemma \ref{lemma:semi_positive} is true if $\T$ is a star shaped tree.} 
For any star-shaped tree with $n$ leaf nodes $1, \cdots, n$, and one internal nodes $n+1$, let $\theta_i = \exp(-d(i, n+1))$ for any leaf nodes $i\in [n]$, then 
        \[\Sigma = 
            \begin{bmatrix}
            1 & \theta_1\theta_2 &\cdots & \theta_1\theta_n\\
            \theta_1\theta_2& 1 & \cdots & \theta_2\theta_n\\
            \vdots &\vdots & \ddots & \vdots \\
            \theta_1\theta_n & \theta_2\theta_n &\cdots &1
            \end{bmatrix} = \theta\theta^T + D,\]
        where $\theta = (\theta_1, \theta_2, \cdots, \theta_n)$ and $D$ is a diagonal matrix with $D_{ii} = 1-\theta_i^2$. When $w_{uv}\geq 0$, $\theta_i \in (0, 1]$ for any $i\in [n]$. It's easy to verify that both $\theta\theta^T$ and $D$ is semi-positive definite and thus $\Sigma$ is also semi-positive definite.
        
\paragraph{If Lemma \ref{lemma:semi_positive} is true for $m\geq 1$ internal nodes, it is also true for $m+1$ internal nodes.} 

For any tree with $n$ leaf nodes and $m+1$ internal nodes, there exists two internal nodes $p$ and $q$ that are neighbors. Delete edge $(p, q)$ splits tree $\T$ into two components, w.l.o.g., we assume
        \begin{equation}
            \Sigma = 
            \begin{bmatrix}
            \Sigma_{11} & \Sigma_{12} \\
            \Sigma_{12}^T & \Sigma_{22}
            \end{bmatrix},
        \end{equation}
        where $\Sigma_{11}$ and $\Sigma_{22}$ are the sub-matrices corresponding to leaf nodes in the first component and the second component, respectively. Now construct a new tree $\widetilde \T = \left(\widetilde V, \widetilde E\right)$ with edge weights $\widetilde w$ and $\widetilde \Sigma$ by shrinking edge $(p, q)$ to zero and deleting node $p$. That is, 
        \begin{equation*}
        \widetilde V = V\setminus\{p\}, \ \widetilde E =  E \mathbin{\big\backslash}
        {\underbrace {\{(p, z): \forall z \text{ such that } (p, z)\in E\}}_{\displaystyle \mathclap{E_p}}} \bigcup {\underbrace {\{(q, z): \forall z\neq q \text{ such that } (p, z)\in E\}}_{\displaystyle \mathclap{E_q}}}
        \end{equation*}
        For any edge $(q, z)$ in the newly added edge set $E_q$, let \(\widetilde {w}_{qz} = w_{pz}\); for all edges $(u, v)\in \widetilde E$ that are originally presented in $E$, let \(\widetilde {w}_{uv} = w_{uv}\).

        Then $\widetilde \Sigma$ has a relationship to $\Sigma$: 
            \[\widetilde \Sigma =   
            \begin{bmatrix}
            \widetilde \Sigma_{11} & \widetilde \Sigma_{12} \\
            \widetilde\Sigma_{12}^T & \widetilde \Sigma_{22}
            \end{bmatrix}
             = 
            \begin{bmatrix}
            \Sigma_{11} & \left.\Sigma_{12}\middle /\theta_{pq}\right. \\
            \left.\Sigma_{12}^T\middle /\theta_{pq}\right. & \Sigma_{22}
            \end{bmatrix},\]
    where $\theta_{pq} = \exp(-w_{pq}) \in (0, 1]$. This implies that 
    \[
    \Sigma = \theta_{pq} \widetilde \Sigma + \begin{bmatrix}
            (1-\theta_{pq})\widetilde \Sigma_{11} & 0 \\
            0 & (1-\theta_{pq})\widetilde \Sigma_{22}
            \end{bmatrix}.
            \]
    Notice that $\widetilde \T$ is a tree with $m$ internal nodes, therefore $\widetilde \Sigma$ is positive semidefinite. As principal submatrices of $\widetilde \Sigma$, $\widetilde \Sigma_{11}$ and $\widetilde \Sigma_{22}$ are also positive semidefinite. This suggests that $\Sigma$ is the summation of two positive semidefinite matrices and thus is positive semidefinite.

\subsection{$\T_r$-Overlapping Blockmodel: Proof of Theorem \ref{thm:tsg->tob} and Theorem \ref{thm:tob->tsg}}\label{appendix:tob}

For any two nodes $i,j \in \T_r$, define the  \textbf{nearest common ancestor} of $i$ and $j$, $i\wedge j$, as the node in $V(i,j)$ that is closest to the root $r$.

\subsubsection{Proof of Theorem \ref{thm:tsg->tob}}\label{appendix:tsg->tob}
Consider any $\T$-Stochastic Graph and a rooted tree $\T_r$ with $r$ being any internal node in $V$. We want to construct a $\T_r$-Overlapping Blockmodel that is equivalent.

In order to achieve this, we define $\theta_i$ for any node $i\in V$, excluding the root node, as $\theta_i = \exp(-d(i, r))$. For the root node $r$, we set $\theta_r = 1$. These $\theta_i$ values will be used in our construction of the $\T_r$-Overlapping Blockmodel.

Specifically, we use the $\theta_i$ values for leaf nodes as the degree corrected parameter, while the $\theta_i$ values for internal nodes are used to define the connectivity matrix $B$. In particular, let $B$ be a diagonal matrix with 
\begin{equation}\label{eq:define_b_tob}
    B_{uu} = 
    \begin{dcases}
    \frac{1}{\theta_u^2}-\frac{1}{\theta_{p(u)}^2}  & u\neq r \\
    1 & u = r
    \end{dcases}
\end{equation}
for any internal node $u\neq r$ and let $B_{rr} = 1$. We can verify that $B_{uu} > 0$ since \[\frac{\theta_u}{\theta_p(u)} = \exp\left(-d(u, r) + d(p(u), r)\right) = \exp\left(-w_{u, p(u)}\right) < 1\] 
Then the following inequalities complete the proof

\begin{equation}\label{eq:tob_tsg}
    \begin{aligned}
        \lambda_{ij}^{OB} & = \theta_i \theta_j Z_i^T BZ_j \\
        & = \theta_i \theta_j \sum_{\scaleto{u \, \in \, V(r, i\wedge j)}{9.5pt}} B_{uu}\\
        & = \theta_i \theta_j \left(1 + \sum_{\scaleto{ u \, \in \, V(r, i\wedge j)\setminus \{r\}}{9.5pt}}  \left(\frac{1}{\theta_u^2}-\frac{1}{\theta_{p(u)}^2}\right)\right) \\
        &=  \frac{\theta_i \theta_j}{\theta_{i\wedge j}^2} = \exp(-d(i, j)) = \lambda^{TSG}_{ij}
    \end{aligned}
\end{equation}

\subsubsection{Proof of Theorem \ref{thm:tob->tsg}}\label{appendix:tob->tsg}

Given any $\T_r$-Overlapping Blockmodel with diagonal $B$ and node specific parameter $\theta_i$, construct $\theta_u$ for every internal node $u$ through the following procedure:
\begin{enumerate}
    \item set $\theta_r = 1$ for root $r$,
    \item for any internal node $u$ with $\theta_{p(u)}$ defined, let $\theta_u = \left. 1 \middle / \sqrt{(B_{uu} + 1/\theta^2_{p(u)} )}\right.$,
    \item repeat step 2 until $\theta_u$ is defined for all internal node $u$.
\end{enumerate}
Then $\theta_u$ satisfy Equation \eqref{eq:define_b_tob}. 

Now we construct a $\T$-Stochastic Graph by defining edge weight $w_{uv}$. For any pair of neighbor nodes $(u, v)\in V$, one of them must be the parent of the other in $\T_r$, without loss of generality, assume $u = p(v)$ and
construct edge weight $w_{uv}$ as 
\[w_{uv} = -\log(\theta_u) + \log\left(\theta_{p(u)}\right).\]
Since $\theta_u$ satisfy Equation \eqref{eq:define_b_tob} and $B_{uu}\in \R^+$, $w_{uv}$ is positive if $u$ and $v$ are all internal nodes. The proof is completed using the same equalities (Equation \eqref{eq:tob_tsg}) as in the pervious section.

\subsection{Stochastic Processes that generate graphs with equivalent distributions}

Lemma \ref{lemma:fastrg} \citep{rohe2018note}  is useful to prove equivalence between $\T$-Stochastic Graphs and edge generators. With Lemma \ref{lemma:fastrg}, the equivalence proof only need to show that the probability of sampling edge $(i, j)$ is proportional to $\lambda^{TSG}_{ij}$.

Specifically, for any random graph $A$ such that $A_{ij} \overset{i.d.}{\sim} \text{Poisson}(\lambda_{ij})$, if an edge generator $\mathcal{P}(\cdot)$ samples edge $(i, j)$ with probability proportional to \(\lambda_{ij}\),
then Algorithm \ref{alg:graph_generator} with sparsity parameter $\zeta = \sum_{i\neq j}\lambda_{ij}$ and edge generator $\mathcal{P}(\cdot)$ generates graphs with the same distribution as $A$.

On the other direction, given a sparsity parameter $\zeta$ and an edge generator $\mathcal{P}(\cdot)$ such that the probability $\mathcal{P}(\cdot)$ samples $(i, j)$ is proportional to $\lambda^{TSG}_{ij}$ for some $\T$-Stochastic Graph. Define another random graph $A$ with 
\[\lambda_{ij} = c \lambda^{TSG}_{ij},\] where \[c = \frac{\zeta}{\sum_{i\neq j} \lambda^{TSG}_{ij}}.\]
Then $A$ is also a $\T$-Stochastic Graph by Theorem \ref{thm:similar} (consider $w^1$ being the edge weight in the original $\T$-Stochastic Graph, and construct $w^2$ with $\delta_i \equiv -\log\left(\sqrt{c}\right)$). Further, $\lambda_{ij}$ in this new graph $A$ satisfies $\zeta = \sum_{i\neq j}\lambda_{ij}$.

\begin{lemma}\label{lemma:fastrg}
(Lemma 1 in \citep{rohe2018note}) Let $A\in \R^{n\times n}$ be a random graph with $A_{ij} \overset{i.d.}{\sim} \text{Poisson}(\lambda_{ij})$ for any $i\neq j$. Then
conditioned on $\sum_{i\neq j}A_{ij} = m$,
\[(A_{12}, A_{13}, \cdots, A_{n-1, n}) \sim \text{Multinomial}\left(m, \ \lambda  \middle / \sum_{i\neq j} \lambda_{ij}\right),\]
where $\lambda = (\lambda_{12}, \lambda_{13}, \cdots, \lambda_{n-1, n})$.
\end{lemma}

\subsubsection{$\T_r$-Top Down Stochastic Process: Definitions}\label{appendix:tds_definition}

\begin{definition}\label{def:tds}
    The $\T_r$-Top Down Stochastic Process is a Markov process $X_t$ such that 
    \begin{enumerate}
    \item the state space consists of all the internal nodes in $\T_r$ plus an absorbing state $S$,
    \item it always starts from the root, i.e., $X_0 = r$, 
    \item it only goes to the absorbing state or a child node, i.e., $X_{t+1} \in \ S\cup\mathcal{C}(X_t)$.
\end{enumerate}
\end{definition}

\begin{definition}\label{def:tdg}
    The $\T_r$-Top Down Generator 
    is parameterized by a $\T_r$-Top Down Stochastic Process and a set of degree parameters $\theta_i\in \R^+$ for each leaf node $i$. Specifically, it generates random edge $(I, J)$ with probability
    \begin{align}
        \pr\left(I = i, J = j \mid X_{\tau-1}\right) \propto \theta_i\theta_j \mathds{1}\left\{i, j \in \text{desc}(X_{\tau-1}) \right\}, \ \forall i\neq j \in V_\ell
    \end{align}
\end{definition}

\subsubsection{$\T_r$-Top Down Stochastic Process: Proof of Theorem \ref{thm:tsg->tds}}\label{appendix:tds_proof}

Theorem \ref{thm:tsg->tds} is a direct result of Lemma \ref{lemma:fastrg}, Lemma \ref{lemma:bridge_td}, Lemma \ref{lemma:bridge_ob}, and the equivalence between $\T_r$-Overlapping Blockmodels and $\T$-Stochastic Graphs (Theorem \ref{thm:tsg->tob} and Theorem \ref{thm:tob->tsg}). Specifically, we introduce a Bridge Generator (Definition \ref{def:bridge}) to connect the $\T_r$-Top Down Generator and the $\T_r$-Graphical Overlapping Blockmodel. 

The proof of Lemma \ref{lemma:bridge_td} and \ref{lemma:bridge_ob} can be found in Appendix \ref{appendix:bridge_td} and \ref{appendix:bridge_ob}.

\begin{definition}\label{def:bridge}
    The $\T_r$-Bridge Generator is parameterized by a random variable $X$ that takes values in all internal nodes in $\T_r$ and a set of degree parameters $\theta_i\in \R^+$ for each leaf node $i$. Specifically, it generates random edge $(I, J)$ with probability
     \begin{align}
        \pr\left(I = i, J = j \mid X\right) \propto \theta_i\theta_j \mathds{1}\left\{i, j \in \text{desc}(X) \right\}, \ \forall i\neq j \in V_\ell
    \end{align}
\end{definition}

\begin{lemma}\label{lemma:bridge_td}
Any $\T_r$-Bridge Generator has an equivalent $\T_r$-Top Down Generator, and vice versa. Two edge generators are equivalent if they sample edges with the same probability distribution.
\end{lemma}

\begin{lemma}\label{lemma:bridge_ob}
Any $\T_r$-Bridge Generator has an equivalent $\T_r$-Overlapping Blockmodel with diagonal $B$ matrix in the sense that 
\[\pr((I, J) = (i, j))\propto \lambda^{OB}_{ij},\]
and vice versa. 
\end{lemma}

\subsubsection{Proof of Lemma \ref{lemma:bridge_td}}\label{appendix:bridge_td}
Given any $\T_r$-Top Down Generator with Markov process $X_t$ and parameters $\theta_i$, define a Bridge Generator with the same set of $\theta_i$ and random variable $X$ such that 
\begin{equation}\label{eq:bridge}
    \pr(X = m) = \pr(X_{\tau-1} = m)
\end{equation} for any internal nodes $m$. Then $X_{\tau-1}$ and $X$ has the same distribution and these two generators are equivalent.

Now consider any $\T_r$-Bridge Generator with random variable $X$ and parameters $\theta_i$. Define a $\T_r$-Top Down Generator with the same set of $\theta_i$ and a Markov process $X_t$ such that 
\begin{equation}
\pr(X_{t+1} = h \mid X_t = g)  = 
\begin{dcases}
\frac{\pr(X = h)}{\mathscr{S}(h)} & h \text{  is the absorbing state $S$} \\
\frac{\mathscr{S}(g)}{\mathscr{S}(h)} & h \text{  is an internal node}
\end{dcases},
\end{equation}
where \[\mathscr{S}(h) = \sum_{\displaystyle m \in \, \{h\}\cup \text{desc}(h)} \pr(X = m).\]

Consider any internal node $m$ at level $k$ of the tree, denote nodes on the path from root node $r$ to $m$ as $u_0, u_1, u_2, \cdots, u_k$, where $u_0 = r$ and $u_k = m$. For notation simplicity, define \[\pr(g\rightarrow h) \coloneqq \pr(X_{t+1} = h\mid X_t = g),\] then the following equalities shows that $X_{\tau-1}$ and $X$ has the same distribution.
\begin{align*}
    \pr(X_{\tau-1} = u_k) & = \pr(u_k\rightarrow S)\prod_{i = 0}^k \pr(u_i\rightarrow u_{i+1})\\
    & = \frac{\pr(X = u_k)}{\mathscr{S}(u_k)} \prod_{i = 0}^k p(u_i\rightarrow u_{i+1})\\
    & = \frac{\pr(X = u_k)}{\mathscr{S}(u_k)}  \ \pr(u_{k-1}, u_k)\prod_{i = 1}^{k-1} \pr(u_i\rightarrow u_{i+1})\\
    & = \frac{\pr(X = u_k)}{\mathscr{S}(u_k)} \frac{\mathscr{S}(u_k)}{\mathscr{S}(u_{k-1})}\prod_{i = 0}^{k-1} \pr(u_i\rightarrow u_{i+1})\\
    & = \frac{\pr(X = u_k)}{\mathscr{S}(u_{k-1})}\prod_{i = 0}^{k-1} \pr(u_i\rightarrow u_{i+1})\\
    & = \frac{\pr(X = u_k)}{\mathscr{S}(u_0)}\\
    & = \pr (X = u_k)
\end{align*}

\subsubsection{Proof of Lemma \ref{lemma:bridge_ob}}\label{appendix:bridge_ob}
Given a $\T_r$-Overlapping Blockmodel with $\lambda_{ij}^{OB} = \theta_i \theta_j Z_i ^T BZ_j$, where $Z_i\in \{0, 1\}^K$ is defined as $Z_{iu} = \mathds{1}\{ i\in \text{desc}(u)\}\}$ and $B\in \R^{K\times K}$ is a diagonal matrix. Denote $\vartheta(u)$ as the summation of $\theta_i\theta_j$ for all leaf node descedent of $u$, i.e., 
\[
\vartheta(u) = 
\sum_{\scaleto{i\neq j \text{ and } i, j \in \text{desc}(u) \cap V_\ell}{9.5pt}} \theta_i \theta_j.
\]

Define a $\T_r$-Bridge Generator with the same set of $\theta_i$ and random variable $X$ such that 
\[\pr(X = u)\propto B_{uu} \, \vartheta(u).\]
Then the probability that this $\T_r$-Bridge Generator selects leaf nodes $i$, $j$ is proportional to $\lambda_{ij}^{OB}$
\begin{equation}\label{eq:bridge_ob}
\begin{aligned}
    \pr(I = i, J = j) &= \sum_{\scaleto{u \, \in \, V(r, i\wedge j)}{9.5pt}} \pr(X = u) \frac{\theta_i\theta_j}{\vartheta(u)}\\
    & \propto \sum_{\scaleto{u \, \in \, V(r, i\wedge j)}{9.5pt}} B_{uu} \theta_i\theta_j \\
     &= \theta_i \theta_j \sum_{\scaleto{u \, \in \, V(r, i\wedge j)}{9.5pt}} B_{uu}\\
    & = \theta_i \theta_j Z_i^T BZ_j \\
    & = \lambda_{ij}^{OB}
\end{aligned}
\end{equation}

Similarly, given any $\T_r$-Bridge Generator with a set of $\theta_i$ and a random variable $X$, one can define a $\T_r$-overlapping Blockmodel with the same set of $\theta_i$ and 
\[B_{uu} = \pr(X = u) / \vartheta(u).\]
Then by the above equalities (Equation \eqref{eq:bridge_ob}), $\lambda_{ij}^{OB}$ of this $\T_r$-Overlapping Blockmodel is proportional to the probability that the $\T_r$-Bridge Generator selects leaf nodes $i$, $j$.

\subsubsection{$\T$-Bottom Up Stochastic Process: Definitions}\label{appendix:bus_definition}

\begin{definition}\label{def:bus}
    The $\T$-Bottom Up Stochastic Process is a Markov process $X_t$ such that 
    \begin{enumerate}
        \item the state space consists of all the directed edges in $\T_{dir}$,
        \item each leaf node $i$ is assigned with a probability $\pi_i$ s.t. $\sum_{i\in V_\ell} \pi = 1$ and 
        \begin{equation}\label{eq:bus_initial_state}
            \pr( \, X_0 = \langle i, p(i)\rangle \, ) = \pi_i,
        \end{equation}
        \item any edge that ends at a leaf node is an absorbing state, 
        \item there exists a set of constants $\{c_{uv}: (u, v)\in E\}$ with $c_{uv} = c_{vu}$ such that for any non-absorbing state $\langle r, s \rangle$,
        \begin{equation}\label{eq:bus_transition_prob}
            \pr(X_{t+1} = \langle u, v\rangle \mid X_t = \langle r, s\rangle) \propto c_{uv} \mathds{1}\{s = u, v\neq r\}.
        \end{equation}
        Notice that $c_{uv} = c_{vu}$ here.

    \end{enumerate}
\end{definition}

\begin{definition}\label{def:bug}
    The $\T$-Bottom Up Generator is fully parameterized by a $\T$-Bottom Up Stochastic Process. It generates edge $(I, J)$ with probability
        \[\pr(I = i, J = j) = \pr(X_0 = \langle i, p(i)\rangle, X_\tau = \langle p(j), j \rangle)\]
\end{definition}

To simplify notations, the following notation is used in all subsequent sections about $\T$-Bottom Up Stochastic Process (Appendix \ref{appendix:tsg->bus}, \ref{appendix:bus->tsg}, \ref{appendix:non_symmetric_bu}, and \ref{appendix:negative_edge_bus}). Consider any two nodes $s$, $t$ connected by the path $(u_0 = s, u_1, u_2, \cdots, u_k = t)$, and another node $w$ that is a neighbor of node $s$. Given that the current state $X_t = \langle w, s\rangle$, the random walk can only reach $\langle u_{k-1}, u_k = t\rangle$ after exactly $k$ steps. For notation simplicity, we define
\begin{equation}\label{eq:simplification_1}
    \pr(s\rightarrow t \mid w) \coloneqq \pr\left(X_{t+k} = \langle u_{k-1}, v\rangle \mid X_t = \langle w, u\rangle\right),
\end{equation}
where $k$ is the number of edges in $E(s, t)$. Notice that when $w = u_1$, $\pr(s\rightarrow t \mid w) = 0$ since the random walk is non-backtracking;  when $k = 0$, that is, $s = t$, $\pr(s\rightarrow t \mid w) = 1$. For any leaf node $i\in V_{\ell}$, we employ a further simplification
\begin{equation}\label{eq:simplification_2}
    \pr(i\rightarrow u) \coloneqq \pr(p(i)\rightarrow u \mid i).
\end{equation}

\subsubsection{$\T$-Bottom Up Stochastic Process: Proof of Theorem \ref{thm:tsg->bus}}
\label{appendix:tsg->bus}

Consider a $\T$-Stochastic Graph with additive distance $d(\cdot, \cdot)$, this proof constructs an equivalent $\T$-Bottom up process by defining a set of $\{c_{uv}: (u, v)\in E\}$ and $\{\pi_i: i\in V_\ell\}$. We employ the notation $\pr(i\rightarrow u)$ defined in Appendix \ref{appendix:bus_definition} (Equation \eqref{eq:simplification_2}). 

The proof proceeds by first defining a set of $\{\widetilde c_{uv}: (u, v)\in E\}$ where $\widetilde c_{uv}$ does not necessarily equals to $\widetilde c_{vu}$, but satisfy Equation \eqref{eq:prop_bu} (the existence is by Lemma \ref{lemma:non_symmetric_bu}). Then define $\{c_{uv}: (u, v)\in E\}$ using algorithm \ref{alg:symmetriz_edge_constant}. By Equation \eqref{eq:symmetriz_edge_constant}, $\{c_{uv}\}$ adjust $\{\widetilde c_{uv}\}$ to be symmetric, that is, $c_{uv} = c_{vu}$, without altering the local ratio, that is, consider any node $u$, $c_{uw}\propto \widetilde c_{uw}$ for all of its neighbor $w$. Therefore, $\{c_{uv}\}$ defines a $\T$-Bottom Up Stochastic Process that satisfy Equation \eqref{eq:prop_bu}.

With this $\T$-Bottom Up Stochastic Process, we can construct a $\T$-Bottom Up Generator by defining $\{\pi_i, i\in V_\ell\}$  as 
\begin{equation}\label{eq:define_pi_bu}
    \pi_i = \dfrac{\mathscr{C}(i)}{\sum_{k = 1}^n\mathscr{C}(k)} \, ,
\end{equation}
where $\mathscr{C}(i)$ is the summation of $\lambda_{ij}$ for all other $j\neq i$:
\[\mathscr{C}(i) = \sum_{\{j: j\neq i\}}\lambda_{ij}.\]
Then the $\T$-Bottom Up Generator is \textbf{symmetric} as 
\begin{equation}
\begin{aligned}
    \frac{\pi_i \pr( i \rightarrow j)}{\pi_j\pr(j \rightarrow i)} &= \left. \left(\pi_i \frac{\lambda_{ij}}{\mathscr{C}(i)}\right) \middle/ \left(\pi_j \frac{\lambda_{ji}}{\mathscr{C}(j)}\right)\right. \text{ (by Equation \eqref{eq:prop_bu})}\\[10pt]
    & =  \left. \left(\frac{\pi_i}{\mathscr{C}(i)}\right) \middle/ \left(\frac{\pi_j}{\mathscr{C}(j)}\right)\right. \text{ (since $\lambda_{ij} = \lambda_{ji}$)}\\[10pt]
    & = 1 \text{ (by Equation \eqref{eq:define_pi_bu})}.
\end{aligned}
\end{equation}

We can verify that the probability of choosing edge $(i, j)$ is proportional to $\lambda_{ij}$ as
\begin{equation}
\begin{aligned}
        \pr((I, J) = (i, j)) & = \pi_i \pr( i\rightarrow j)+\pi_j\pr(j\rightarrow i)\\[10pt]
        & = 2 \, \pi_i \, \pr(i\rightarrow j)\\[5pt]
        & = 2 \, \dfrac{\mathscr{C}(i)}{\sum_{k = 1}^n\mathscr{C}(k)} \, \frac{\lambda_{ij}}{\mathscr{C}(i)} \\[5pt]
        & = 2 \, \dfrac{\lambda_{ij}}{\sum_{k = 1}^n\mathscr{C}(k)} \\[5pt]
        & = \dfrac{\lambda_{ij}}{\sum_{i \neq j} \lambda_{ij}} \, .
\end{aligned}
\end{equation}

\begin{lemma}\label{lemma:non_symmetric_bu}
   For any $\T$-Stochastic Graph, there exists a $\T$-Bottom Up Stochastic Process $X_t$ with a set of constants \(\left\{\widetilde c_{uv}: (u, v)\in E\right\}\) where $\widetilde c_{uv} = \widetilde c_{vu}$ \textbf{not} necessarily hold, such that for any leaf node $i$,
   \begin{equation}\label{eq:prop_bu}
      \pr(i \rightarrow j ) \propto \lambda_{ij}, 
   \end{equation}
   where $j \neq i$ is another leaf node in $\T$.
\end{lemma}

\begin{algorithm}
	\caption{Symmetrize edge constant $c_{uv}$}
	\label{alg:symmetriz_edge_constant}
	\KwIn{$\{\widetilde c_{uv}: (u, v)\in E\}$ with no guarantee for $c_{uv} = c_{vu}$}
	\KwOut{$\{c_{uv}: (u, v)\in E\}$ with $c_{uv} = c_{vu}$}
	choose any internal node $r$ as the root, and consider the rooted tree $\T_r$, denote the level of $\T_r$ as $L$, denote $\mathscr{L}^{(i)}$, $i \in [L]$, as the set of level $i$ nodes in $\T_r$\;
 set $c_{ru} = \widetilde c_{ru}$ for all $u\in \mathscr{L}^{(1)}$\;
	\For{$i \gets 1$ \KwTo $L$}{
	\For{$u \in  \mathscr{L}^{(i)} $}{
 set $c_{uw}$ for any node $w$ that is a neighbor of $u$ as
 \begin{equation}\label{eq:symmetriz_edge_constant}
     \scaleto{c_{uw} = \frac{c_{p(u), u}\, \widetilde c_{uw}}{\widetilde c_{u, p(u)}}}{30pt}
 \end{equation}
    }
    }
\end{algorithm}

\subsubsection{Proof of Lemma \ref{lemma:non_symmetric_bu}}\label{appendix:non_symmetric_bu}

We use the notation $\theta_{uv}$ to denote $\exp(-d(u, v))$ for any two nodes $u, v \in V$, notice that $\theta_{ij} = \lambda_{ij}$ for any two leaf nodes $i, j\in V_\ell$. We also employ the notations $\pr(s\rightarrow t \mid w)$ and $\pr(i\rightarrow u)$ defined in Appendix \ref{appendix:bus_definition} (Equation \eqref{eq:simplification_1} and \eqref{eq:simplification_2}). This proof aims to construct a $\T$-Bottom Up Stochastic Process with \(\left\{\widetilde c_{uv}: (u, v)\in E\right\}\) that satisfy Equation \eqref{eq:prop_bu}. Importantly, there is \textbf{no} requirement for $\widetilde c_{uv} = \widetilde c_{vu}$. Throughout the proof, the term ``$\T$-Bottom Up Stochastic Process'' refer to the Stochastic Process in Definition \ref{def:bus} without the requirement for $\widetilde c_{uv} = \widetilde c_{vu}$.

The key idea is to define $\widetilde c_{uv}$ sequentially, starting from edges connected to leaf nodes. To do this, we introduce the following definition of ``half-tree''. 

\begin{definition}
For any two neighbor nodes $r$ and $w$ in tree $\T$, deleting edge $(r, w)$ splits the tree into two subtrees. Root these subtrees at node $r$ and $w$, and denote the rooted trees as $\T_{r/w}$ and $\T_{w/r}$. We call any tree created in this fashion a \textbf{half-tree} of $\T$. Some examples of half trees are displayed in Figure \ref{fig:bus_subtree}.
\end{definition}

\begin{definition}
We say a $\T$-Bottom Up Stochastic Process is properly defined on a half-tree $\T_{r/w}$ if 
\begin{equation}\label{eq:half_tree_proper}
    \pr(r\rightarrow i\mid w)\propto \theta_{ri}
\end{equation}
for any leaf node $i$ in $\T_{r/w}$
\end{definition}

\begin{definition}
We say a $\T$-Bottom Up Stochastic Process is level $k$ properly defined if it is properly defined on all half-trees with level $\leq k$.
\end{definition}

The proof proceeds by induction on the level of half trees. The level of a tree is defined as the maximum level of all its nodes. For the definition of the level of a particular node, please refer to Appendix \ref{appendix:other_models_definition}. The induction is achieved by proving Lemma \ref{lemma:bus1} and Lemma \ref{lemma:bus2} below. By these two lemmas, a $\T$-Bottom Up Stochastic Process can be defined on tree $\T$ such that it is properly defined on all half trees. Lemma \ref{lemma:non_symmetric_bu} is an immediate result of this by noticing that 
\[\frac{\lambda_{ij}}{\lambda_{i\ell}} = \frac{\theta_{p(i), j}}{\theta_{p(i), \ell}}\]
for any leaf nodes $i, j, \ell$, 
and that $\T_{p(i)/i}$ for any leaf node $i$ is a half tree of $\T$.

\begin{lemma}\label{lemma:bus1}
Any $\T$-Bottom Up Stochastic Process is level 0 half properly defined.
\end{lemma}
\begin{proof}
For any level 0 half tree $\T_{r, w}$, there is only one leaf node, which is $r$, and $\pr(r\rightarrow r \mid w) = 1$.
\end{proof}

\begin{lemma}\label{lemma:bus2}
If there exists a level $k$ properly defined $\T$-Bottom Up Stochastic Process, then there exists a level $k+1$ properly defined $\T$-Bottom Up Stochastic Process.
\end{lemma}
\begin{proof}

Given a level $k$ properly defined $\T$-Bottom Up Stochastic Process, we construct a level $k+1$ properly defined $\T$-Bottom Up Stochastic Process by redefining all $\widetilde c_{ru}$ for neighbor nodes $(r, u)$ such that $\T_{u/r}$ is a level $k$ tree.

Specifically, for each $\T_{u/r}$ that is a level $k$ tree, select a leaf node in $\T_{u/r}$ and denote it as $\ell(u/r)$, define 
\[\widetilde c_{ru} = \frac{\theta_{r,\ \ell(u/r)}}{\pr(u\rightarrow \ell(u/r)\mid r)}.\]
As this definition does not change $c_{uv}$ in any level $\leq k$ half trees, the $\T$-Bottom Up Stochastic Process is still level $k$ properly defined. We just need to show it is properly defined on all level $k+1$ half trees.

Consider any two leaf nodes $i$ and $j$ that belong to a level $k+1$ half tree $\T_{r/w}$. If they also belongs to the same level $\leq k$ half tree $\T_{p/q}$ that is a subtree of $\T_{r/w}$, then 
\begin{equation}\label{eq:bus_proof_1}
    \begin{aligned}
        \frac{\pr(r\rightarrow i \mid w)}{\pr(r\rightarrow j \mid w)} & = \frac{\pr(r\rightarrow p \mid w) \pr(p\rightarrow i \mid q)}{\pr(r\rightarrow p \mid w) \pr(p\rightarrow j \mid q)}\\[10pt]
        & = \frac{\pr(p\rightarrow i \mid q)}{\pr(p\rightarrow j \mid q)} = \frac{\theta_{pi}}{\theta_{pj}} \text{ (because the process is properly defined on $\T_{p/q}$)}\\[10pt]
        & = \frac{\theta_{pi}\theta_{rp}}{\theta_{pj}\theta_{r p}} = \frac{\theta_{ri}}{\theta_{rj}} \,.
    \end{aligned}
\end{equation}
If nodes $i$ and $j$ are not part of the same level $\leq k$ half tree, then they must belong to two distinct level $k$ half trees, as shown in Figure \ref{fig:bus_subtree}. 
\begin{figure}[h]
    \centering
    \includegraphics[width = 3in]{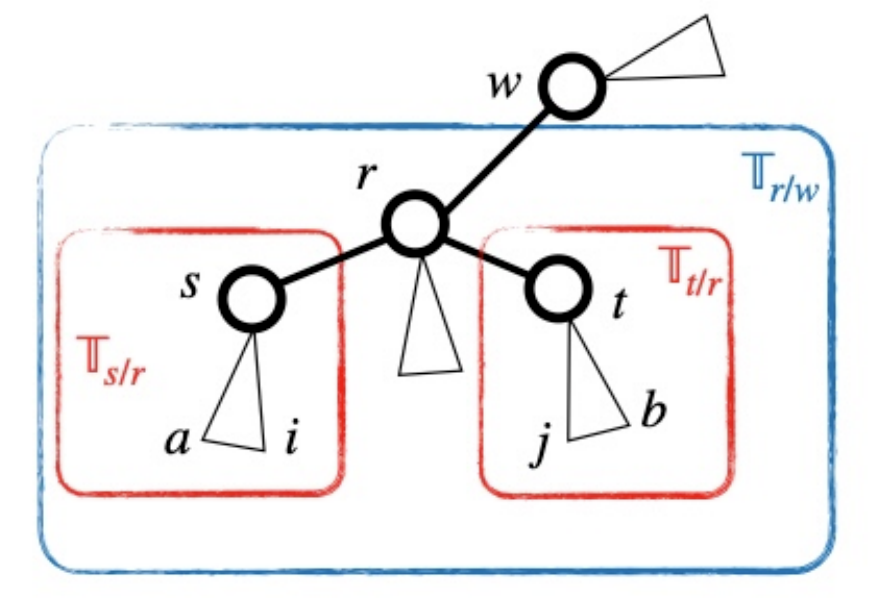}
    \caption{Triangles are subtrees and circles are internal nodes. Node $i$, $j$ are leaf nodes in subtree $\T_{s/r}$ and $\T_{t/r}$. Nodes $a = \ell(r/s)$, $b = \ell(r/t)$ are leaf nodes that help defining $\widetilde c_{rs}$ and $\widetilde c_{rt}$.}
    \label{fig:bus_subtree}
\end{figure}
Specifically, we denote the subtrees containing nodes $i$ and $j$ as $\T_{s/r}$ and $\T_{t/r}$. We also use $a$ and $b$ to denote the leaf nodes that help define $\widetilde c_{rs}$ and $\widetilde c_{rt}$, that is, $a = \ell(s/r)$ and $b = \ell(t/r)$. Then we have 
    \begin{align*}
        \frac{\pr(r\rightarrow i \mid w)}{\pr(r\rightarrow j \mid w)} 
        & = \dfrac{\pr(r\rightarrow s \mid w) \, \pr(s\rightarrow i \mid r)}{\pr(r\rightarrow t \mid w) \, \pr(t\rightarrow j \mid r)}\\[10pt]
         &= \dfrac{\widetilde c_{rs} \, \pr(s\rightarrow i \mid r)}{\widetilde c_{rt} \, \pr(t\rightarrow j \mid r)} \text{ (transition probabilities are proportional $\widetilde c_{uv}$ by definition)}\\[10pt]
        & = \left.\left(\frac{\theta_{ra}}{\pr(s\rightarrow a\mid r)} \, \pr(s\rightarrow i \mid r)\right) \middle / \left(\frac{\theta_{rb}}{\pr(t\rightarrow b\mid r)} \, \pr(t\rightarrow j \mid r)\right)\right. \text{ (plug in $\widetilde c_{ru}$ and $\widetilde c_{rv}$)}\\[10pt]
        & = \left.\left(\frac{\theta_{ra} \,  \theta_{si}}{\theta_{sa}} \right) \middle / \left(\frac{\theta_{rb} \, \theta_{tj}}{\theta_{tb}} \right)\right. \text{ (the process is level $k$ properly defined)}\\[10pt]
        & = \dfrac{\theta_{rs}\theta_{si}}{\theta_{rt}\theta_{tj}} = \frac{\theta_{ri}}{\theta_{rj}} \, . \stepcounter{equation}\tag{\theequation}\label{eq:bus_proof_2}
    \end{align*}

The proof is completed by combining Equation \eqref{eq:bus_proof_1} and Equation \eqref{eq:bus_proof_2}, which together demonstrate that Equation \eqref{eq:half_tree_proper} is valid for all level $k+1$ half trees.

\end{proof}

\subsubsection{$\T$-Bottom Up Stochastic Process: Proof of Theorem \ref{thm:bus->tsg}}\label{appendix:bus->tsg}

This section employs the notations $\pr(s\rightarrow t \mid w)$ and $\pr(i\rightarrow u)$ defined in Appendix \ref{appendix:bus_definition} (Equation \eqref{eq:simplification_1} and \eqref{eq:simplification_2}). 

We begin by demonstrating that $\{\pi_i: i\in V_\ell\}$ can be defined to make the $\T$-Bottom Up Generator symmetric. Choose a leaf node $i$ and set $c_i = 1$; for any other leaf node $j$, define $c_j$ as follows:
\[c_j = \frac{c_i \pr( i \rightarrow j)}{\pr(j\rightarrow i)}.\]
Normalize $c_i$ to obtain $\pi_i$:
\[\pi_i = \frac{c_i}{\sum_{k} c_k}.\]
It is evident that the symmetric condition in Equation \eqref{eq:symmetric} holds for leaf node $i$ and any other leaf node $j\neq i$: 
\begin{equation}\label{eq:symmetric_1}
\frac{\pr(X_0 = \langle i, p(i) \rangle, X_{\tau} = \langle p(j), j \rangle)}{\pr(X_0 = \langle j, p(j) \rangle, X_{\tau} = \langle p(i), i \rangle)}  = \frac{\pi_i \pr( i \rightarrow j)}{\pi_j \pr(j\rightarrow i)} = 1.
\end{equation}
To demonstrate that the symmetric condition also holds for two leaf nodes $j, \ell$ that are distinct from node $i$, additional notations are introduced, as shown in Figure \ref{fig:bus_symmetric}.
\begin{figure}[h]
    \centering
    \includegraphics[width = 2.5in]{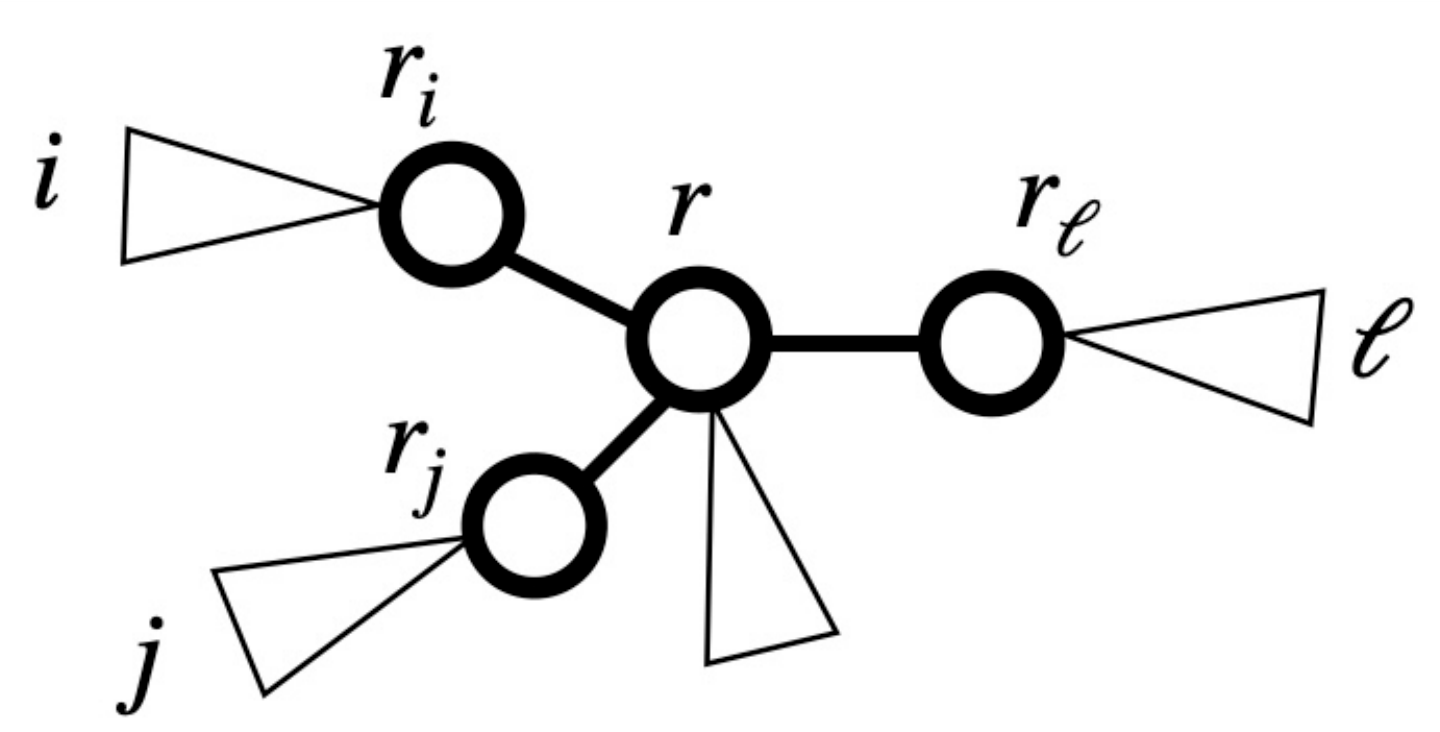}
    \caption{Triangles are subtrees and circles are internal nodes. Node $i, j, \ell$ are leaf nodes in the corresponding subtrees. Node $r$ is the median node of leaf nodes $i, j, \ell$, and nodes $r_i, r_j, r_\ell$ are neighbor nodes of $r$ that lie on the path between them.}
    \label{fig:bus_symmetric}
\end{figure}
 Specifically, we denote the median node\footnote{the definition of median node can be found in Section \ref{sec:hse}, Definition \ref{def:median}} of $i, j$, and $\ell$ as $r$; we also use $r_i$, $r_j$, and $r_\ell$ to denote neighbor nodes of $r$ that lie on the paths $(r,i)$, $(r,j)$, and $(r,\ell)$, respectively. Then we have
\begin{align*}
    \frac{\pi_j }{\pi_\ell } &= \left.\frac{ \pr(i\rightarrow j)}{ \pr(j\rightarrow i)} \middle/ \frac{ \pr(i\rightarrow \ell)}{\pr(\ell\rightarrow i)}\right.  \ (\text{by the construction of $\pi_j$ above})\\[5pt]
    & = \frac{\pr(i\rightarrow j)}{\pr(i\rightarrow \ell)} \times \frac{\pr(\ell\rightarrow i)}{\pr(j\rightarrow i)} \\[5pt]
    &= \frac{\pr(i\rightarrow r)\ \pr(r\rightarrow j \mid r_i)}{\pr(i\rightarrow r) \ \pr(r\rightarrow \ell \mid r_i)} \times \frac{\pr(\ell\rightarrow r)\ \pr(r\rightarrow i \mid r_\ell)}{\pr(j\rightarrow r)\ \pr(r\rightarrow i \mid r_j)}\\[5pt]
    & =  \frac{\pr(\ell\rightarrow r) \ \pr(r\rightarrow j \mid r_i) \ \pr(r\rightarrow i \mid r_\ell)}{\pr(j\rightarrow r) \ \pr(r\rightarrow \ell \mid r_i)\ \pr(r\rightarrow i \mid r_j)} \times \frac{\pr(r\rightarrow j \mid r_\ell) \ \pr(r\rightarrow \ell \mid r_j)}{\pr(r\rightarrow j \mid r_\ell) \ \pr(r\rightarrow \ell \mid r_j)} \\
    &   \text{\qquad (cancel out $\pr(i\rightarrow r)$, rearrange the equation, and multiply by one)}\\[5pt]
    & = \frac{\pr(\ell\rightarrow r)\ \pr(r\rightarrow j \mid r_\ell)}{\pr(j\rightarrow r) \ \pr(r\rightarrow \ell \mid r_j)}
    \times
    \frac{\pr(r\rightarrow j \mid r_i) \ \pr(r\rightarrow \ell \mid r_j)\ \pr(r\rightarrow i \mid r_\ell) }{\pr(r\rightarrow \ell \mid r_i)\ \pr(r\rightarrow i \mid r_j) \ \pr(r\rightarrow j \mid r_\ell)} \text{ \, (rearrange)}\\[5pt]
    & = \frac{\pr(\ell\rightarrow r)\ \pr(r\rightarrow j \mid r_\ell)}{\pr(j\rightarrow r) \ \pr(r\rightarrow \ell \mid r_j)} \times \frac{c_{rj} c_{r\ell} c_{ri}}{c_{r\ell} c_{ri} c_{rj}} \text{ \,  (transition probabilities are proportional to $c_{uv}$)} \\[5pt]
    & = \frac{\pr(\ell\rightarrow j)}{\pr(j\rightarrow \ell)}\, . \stepcounter{equation}\tag{\theequation}\label{eq:symmetric_2}
\end{align*}

When considered together, Equation \eqref{eq:symmetric_1} and Equation \eqref{eq:symmetric_2} demonstrate that the $\T$-Bottom Up generator is symmetric. Our next task is to prove that this generator produces graphs that are $\T$-HSE. Denote the probability of sampling edge $(i,j)$ as $\lambda_{ij}^{BU}$, then

\vspace{-0.2in}
\begin{align*}
\frac{\lambda^{BU}_{i\ell}}{\lambda^{BU}_{j\ell}} &= \frac{\pi_i \pr(i\rightarrow \ell) + \pi_\ell \pr(\ell\rightarrow i)}{\pi_j \pr(j\rightarrow \ell)+ \pi_\ell \pr(\ell\rightarrow j)} \\[5pt]
&= \frac{2\pi_\ell \pr(\ell\rightarrow i)}{2\pi_\ell \pr(\ell\rightarrow j)}\\[5pt]
& = \frac{\pr(\ell\rightarrow i)}{\pr(\ell\rightarrow j)}\\[5pt]
& = \frac{\pr(\ell\rightarrow r) \ \pr(r\rightarrow i \mid r_\ell )}{\pr(\ell\rightarrow r) \ \pr(r\rightarrow j \mid r_\ell )} \\[5pt]
& = \frac{c_{ri}}{c_{rj}}.
\end{align*}
Since the last line, $c_{ri}/c_{rj}$, only depends on $\ell$ through the median node $r$, thus the proof is completed by defining $c_{ij}^{(r)} =  c_{ri}/c_{rj}$.

\subsubsection{$\T$-Bottom Up Stochastic Process: an equivalent condition for negative edge weights}\label{appendix:negative_edge_bus}

\begin{theorem}\label{thm:bus_negative}
(Equivalent condition for negative edge weights) Consider any internal edge $(e, f)$ and neighbor nodes labeled as in Figure \ref{fig:subtree}, then 
\[w_{ef} < 0 \Longleftrightarrow \frac{\pr(X_{t+1} = \langle e, b \rangle \mid X_{t} = \langle a, e \rangle)}{\pr(X_{t+2} = \langle f, d \rangle \mid X_{t} = \langle a, e \rangle)} < \frac{\pr(X_{t+2} = \langle e, b \rangle \mid X_{t} = \langle c, f \rangle)}{\pr(X_{t+1} = \langle f, d \rangle \mid X_{t} = \langle c, f \rangle)}\]
\end{theorem}

\begin{figure}[!ht] %
   \centering
   \includegraphics[width=3in]{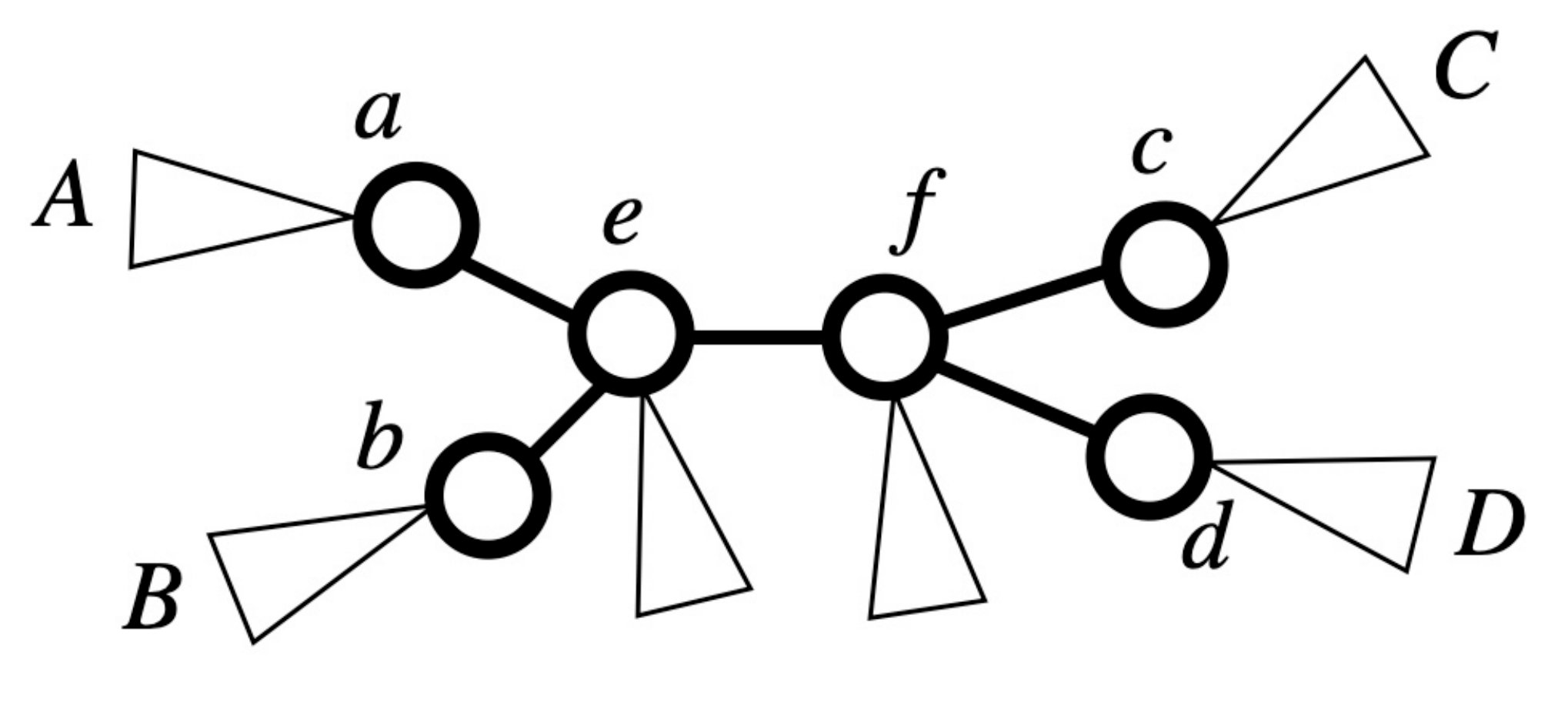} 
   \caption{An example of local structures for internal edge $(e, f)$. Triangles are subtrees and circles are leaf nodes. Nodes $A, B, C, D$ are leaf nodes in the corresponding subtrees.}
   \label{fig:subtree}
\end{figure}

\begin{proof}
This proof utilize the notations $\pr(s\rightarrow t \mid w)$ and $\pr(i\rightarrow u)$ defined in Appendix \ref{appendix:bus_definition} (Equation \eqref{eq:simplification_1} and \eqref{eq:simplification_2}). Consider the tree structure displayed in Figure \ref{fig:subtree}, where nodes $A, B, C, D$ represent leaf nodes from subtrees linked to nodes $a, b, c, d$, respectively, then the proof is complete by the following equivalences:
\[\frac{\pr(e\rightarrow b \mid a)}{\pr(e\rightarrow d \mid a)}<\frac{\pr(f\rightarrow b \mid c)}{\pr(f\rightarrow d \mid c)}\]
\[\big\Updownarrow\]
\[\frac{\pr(A\rightarrow e) \ \pr(e\rightarrow b \mid a) \ \pr(b\rightarrow B \mid e)}{\pr(A\rightarrow e) \ \pr(e\rightarrow d \mid a) \ \pr(d\rightarrow D \mid e)}< \frac{\pr(C \rightarrow f) \ \pr(f\rightarrow b \mid c) \ \pr(b\rightarrow B \mid e)}{\pr(C\rightarrow f) \ \pr(f\rightarrow d \mid c) \ \pr(d\rightarrow D \mid e)}\]
\[\big\Updownarrow\]
\[\frac{\lambda^{TSG}_{AB}}{\lambda^{TSG}_{AD}}<\frac{\lambda^{TSG}_{CB}}{\lambda^{TSG}_{CD}}\]
\[\big\Updownarrow\]
\[\exp\{-d(A, B) + d(A, D)\} < \exp\{-d(C, B) + d(C, D)\} \]
\[\big\Updownarrow\]
\[\exp\{-d(A, e) - d(e, B) +d(A, e) + d(e, f)+d(f, D)\} < \exp\{-d(C, f) - d(f, e) - d(e, B)  + d(C, f) + d(f, D)\} \]
\[\big\Updownarrow\]
\[\exp(2d(e, f))<1\]
\[\big\Updownarrow\]
\[d(e, f)<0.\]
\end{proof}

\subsection{Hierarchical Stochastic Equivalence: Proof of Theorem \ref{thm:tse->tsg} and Theorem \ref{thm:tsg->tse}} \label{sec:hsetheorem}

\subsubsection{Proof of Theorem \ref{thm:tse->tsg}}\label{appendix:tse->tsg}
Before proceeding to the formal proof, we first offer some insight of it. Consider a graph $A$ with $\lambda_{ij}$ such that $A$ is $\T$-HSE. The aim is to create positive values $\theta_{uv}$ on all edges $(u,v)$ in $\T$ such that $\lambda_{ij}$ is equal to the product of the $\theta_{uv}$'s for all edges $(u,v)\in E(i,j)$. Then an equivalent $\T$-Stochastic Graph can be constructed by defining edge weights $w_{uv}$ as $-\log(\theta_{uv})$. 

We would construct $\theta$ values layer by layer, starting from leaf nodes. Consider any two leaf nodes $i$ and $j$ that share the same parent node, denoted as $r$. Then, for any other leaf node $\ell$, the median node of nodes $i$, $j$, and $\ell$ is node $r$, as displayed in Figure \ref{fig:hse_intuition}. Given the $\lambda$ values, we have two equations and two unknowns (which can be solved):
\begin{equation}
\begin{dcases}
    \lambda_{ij} = \theta_{ir}  \theta_{rj},\\[5pt]
    \dfrac{\lambda_{i\ell}}{\lambda_{j\ell}} = \frac{\theta_{ir}}{ \theta_{rj}}.
\end{dcases}
\end{equation}
\begin{figure}[h]
    \centering
    \includegraphics[width=1.5in]{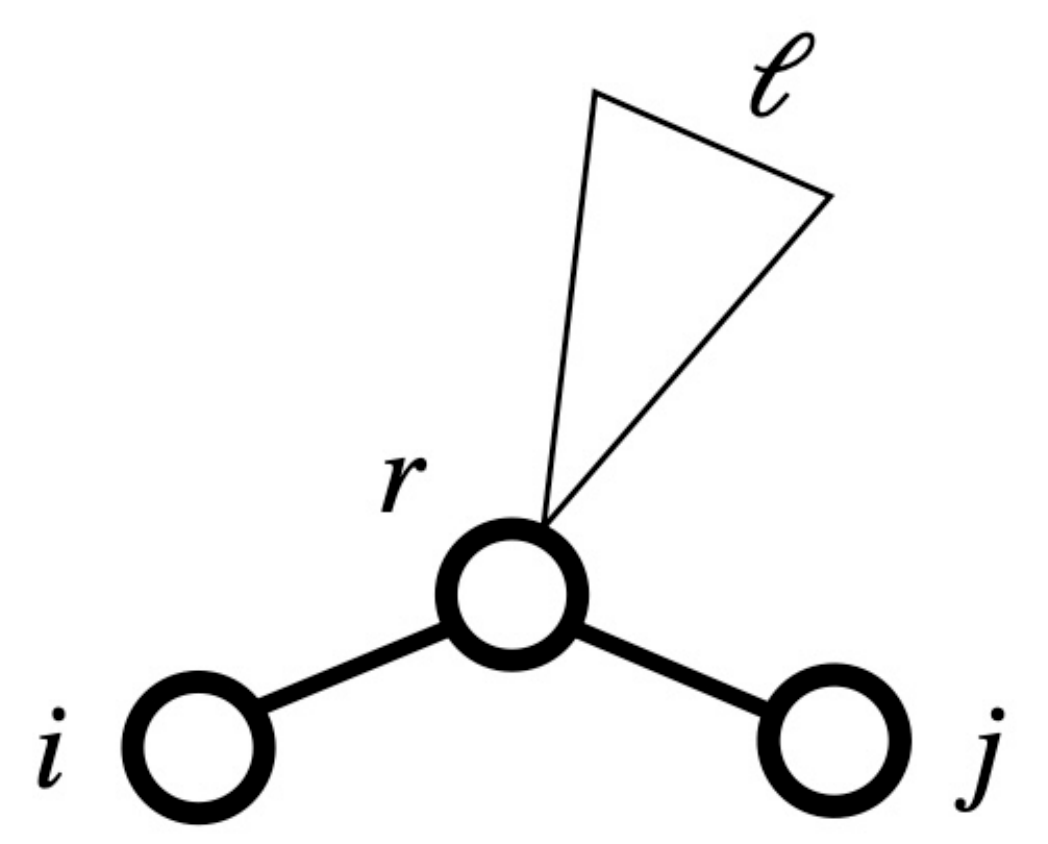}
    \caption{The triangle represents a subtree. Nodes $i, j$ are leaf nodes; node $r$ is the parent of $i, j$, that is, $r = p(i) = p(j)$; node $\ell$ is another leaf node in the rest part of the tree.}
    \label{fig:hse_intuition}
\end{figure}

Because the $\lambda_{ij}$ are positive, the $\theta$'s can be solved as positive values too.
After constructing these $\theta$'s for all leaf nodes, we proceed to construct the $\theta$'s for the edges above their parents. Iteratively, solving these equations constructs $\theta$ values that are consistent with the $\lambda$'s and defines edge weights $w$ for an equivalent $\T$-Stochastic Graph.

The formal proof is carried out using induction on the tree level. To ensure that ``levels'' are appropriately defined, Lemma \ref{lemma:level1} and Lemma \ref{lemma:levelk} discuss Theorem \ref{thm:tsg->tse} under a rooted tree $\T_r$. It is noteworthy that $\T_r$-HSE and $\T$-HSE are equivalent for any root node $r$, as the definition of HSE (Definition \ref{def:hse}) does not involve the root node. The level of a tree is defined as the maximum level of all its nodes. For the definition of the level of a particular node, please refer to Appendix \ref{appendix:other_models_definition}.

\begin{lemma}\label{lemma:level1}
Theorem \ref{thm:tse->tsg} holds if $\T_r$ is a level one tree.
\end{lemma}
\begin{proof}
To begin with, it is important to observe that a level one tree is a star-shaped tree consisting of a single internal node that also serves as the root node $r$. Furthermore, for any triplet of leaf nodes $i, j, \ell$, it holds that $m(i, j, \ell) = r$ and thus 
\begin{equation}\label{eq:hse_level1_0}
    \lambda_{i, \ell} /\lambda_{j, \ell} = c^{(r)}_{ij}
\end{equation}
For notation simplification, we drop the superscript $(r)$ in $c^{(r)}_{ij}$ in the following proof (this is valid as all triplets of leaf nodes share the same median node).

Pick any two leaf nodes $i, j$ and define 
\begin{equation}\label{eq:hse_level1_1}
    \begin{dcases}
    \theta_{ir} = \sqrt{\lambda_{ij}c_{ij}},\\[5pt]
    \theta_{jr} = \sqrt{\lambda_{ij}/c_{ij}}.
    \end{dcases}
\end{equation}
For any other leaf node $\ell$, define
\begin{equation}\label{eq:hse_level1_2}
    \theta_{\ell r} = \lambda_{i \ell}/\theta_{ir}
\end{equation}
Now we want to show that under this definition, $\lambda^{HSE}_{pq} = \theta_{qr}\theta_{pr}$ for any two leaf nodes $p$ and $q$. We prove this by discussing all possibilities of how $(p, q)$ involves nodes $(i, j)$:

\begin{enumerate}
    \item $p = i, q = j$: plug in Equation \eqref{eq:hse_level1_1} for $\theta_{ir}$ and $\theta_{jr}$, \[\theta_{ir}\theta_{jr} = \lambda_{ij}\] 
    \item $p = i, q \neq j$: plug in Equation \eqref{eq:hse_level1_2} for $\theta_{qr}$, 
    \begin{equation}\label{eq:hse_level1_3}
        \theta_{ir}\theta_{qr} = \theta_{ir} \times \lambda_{i q}/\theta_{ir} = \lambda_{i q}
    \end{equation}
    \item $p = j, q \neq i$: the second equality uses the result in Equation \eqref{eq:hse_level1_3}, the last equality follows from the HSE property in Equation \eqref{eq:hse_level1_0}
    \begin{equation*}
        \theta_{jr}\theta_{qr} 
        = \frac{\theta_{jr}}{\theta_{ir}} \times \theta_{ir}\theta_{qr} 
        = \frac{1}{c_{ij}} \times  \lambda_{iq} 
        = \lambda_{jq} 
   \end{equation*}
    \item $p, q \neq i, j$: 
        \begin{align*}
        \theta_{pr}\theta_{qr} 
        &= \frac{\lambda_{ip}\lambda_{iq}}{\theta_{ir}^2} 
        \quad \text{(plug in Equation \eqref{eq:hse_level1_2})} \\[5pt]
        &= \frac{\lambda_{ip}}{\lambda_{pq}} \times \frac{\lambda_{iq}\lambda_{pq}}{\theta_{ir}^2}
        \quad \text{(multiply both the numerator and the denominator by $\lambda_{pq}$)} \\[5pt]
        &= \frac{\lambda_{ij}}{\lambda_{jq}}\times \frac{\lambda_{iq}\lambda_{pq}}{\theta_{ir}^2}
        \quad \text{(by HSE property)} \\[5pt]
        & = \frac{\lambda_{iq}}{\lambda_{jq}} \times \frac{\lambda_{ij}\lambda_{pq}}{\theta_{ir}^2 }
        \quad \text{(rearrange)} \\[5pt]
        & = \frac{c_{ij}\lambda_{ij}\lambda_{pq}}{\theta_{ir}^2 }
         \quad \text{(by HSE property)} \\[5pt]
         & = \lambda_{pq}
        \quad \text{(by Equation \eqref{eq:hse_level1_1})}
   \end{align*}
\end{enumerate}

\end{proof}

\begin{lemma}\label{lemma:levelk}
If Theorem \ref{thm:tse->tsg} holds when $\T_r$ is a level $k$ tree, then it also holds when $\T_r$ is a level $k+1$ tree.
\end{lemma}
\begin{proof}
Given a set of $\lambda_{ij}$ that is HSE on a $k+1$-level tree $\T_r$, denote the set of nodes at level $\ell$ as $\mathscr{L}^\ell$. Notice that $\mathscr{L}^{k+1}$ only contains leaf nodes, that is, $\mathscr{L} \subset V_\ell$. Deleting all nodes in $\mathscr{L}^{k+1}$ creates a level $k$ tree, denote this tree as $\T^{k}_r$ and denote the set of leaf nodes of $\T^{k}_r$ as $V_{\ell}^k$. See Figure \ref{fig:hse_level_k} for an example of $\T_r$.

\begin{figure}[h]
    \centering
    \includegraphics[width = 5.5in]{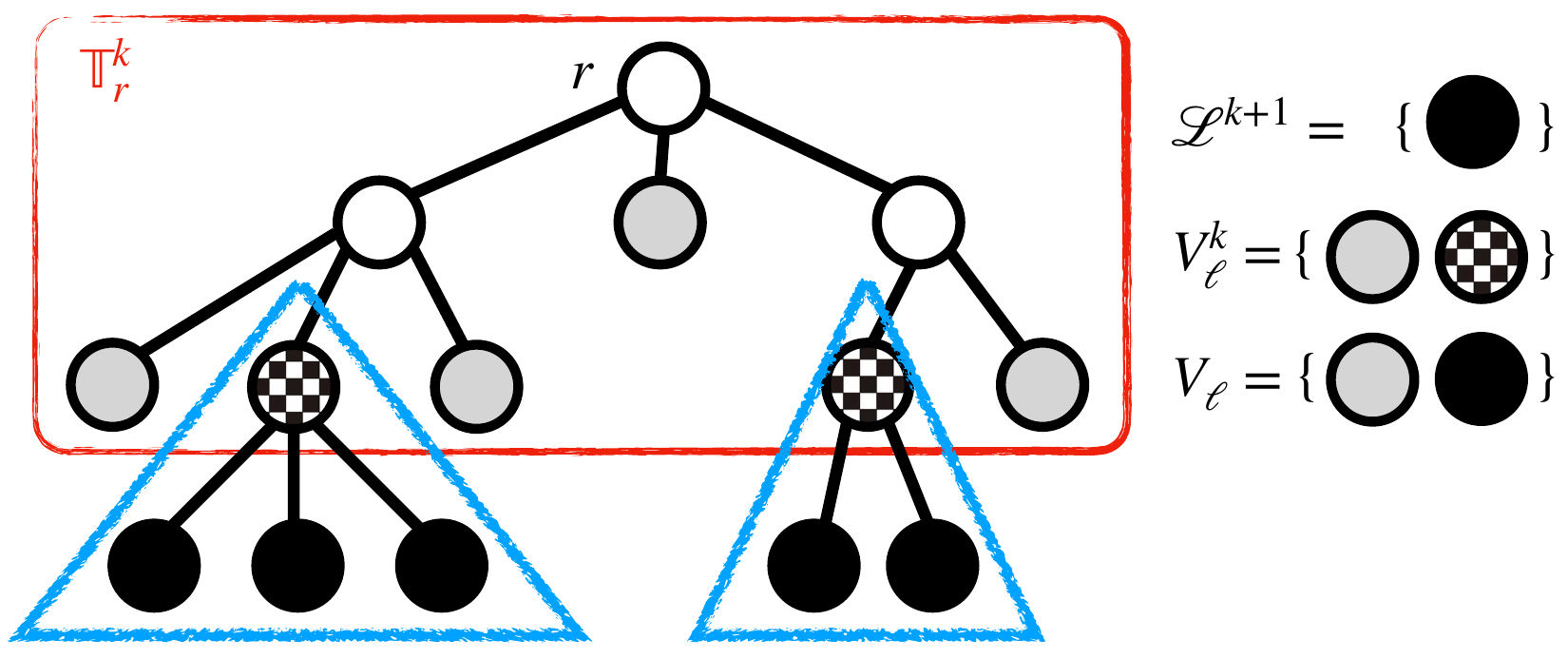}
    \caption{An example of 3-level rooted tree $\T_r$}
    \label{fig:hse_level_k}
\end{figure}

The proof can be summarized into three steps
\begin{enumerate}
    \item define $\theta_{i, p(i)}$ values for each level $k+1$ leaf node $i \in \mathscr{L}^{k+1}$ (edges in the blud triangles in Figure \ref{fig:hse_level_k}),
    \item create a set of $\{\lambda_{uv}^k: u, v\in V_{\ell}^k\}$ values such that $\lambda_{uv}^k$ is $\T_r^k$-HSE. Since $\T_r^k$ is a level $k$ tree, there exists a set of $\theta$'s for edges in $\T_r^k$ (edges in the red rectangle in Figure \ref{fig:hse_level_k}),
    \item show that $\theta_{uv}$ constructed in the first step and the second step together defines a set of valid $\theta$'s for $\T_r$.
\end{enumerate}

\paragraph{Define $\theta_{i, p(i)}$}: 
For each node $u \in V_{\ell}^k \setminus V_{\ell}$, there must exists at least two level $k+1$ leaf nodes connected to it, denote them as $i$ and $j$. Then the median node of $i, j$, and any other leaf nodes $\ell \in V_\ell$ is node $u$. Define 
\begin{equation}\label{eq:hse_levelk_1}
    \begin{dcases}
    \theta_{iu} = \sqrt{\lambda_{ij}c^{(u)}_{ij}},\\[5pt]
    \theta_{ju} = \sqrt{\lambda_{ij}/c^{(u)}_{ij}},\\[5pt]
    \theta_{\ell u} = \lambda_{i \ell}/\theta_{iu} \quad \text{for $\forall \ell \in V_\ell$}.\\
    \end{dcases}
\end{equation}
It is easy to verify that for any two leaf nodes $p, q\in V_{\ell}$, 
\begin{equation}\label{eq:hse_levelk_prod}
\theta_{up}\theta_{uq} = \lambda_{pq}
\end{equation}
holds if at least one of the nodes $p, q$ is a neighbor node of node $u$. The details are omitted since it is similar to the proof of Lemma \ref{lemma:level1}. Notice that this procedure defines more $\theta_{ui}$ values than we need (we only need $\theta_{ui}$ for $i \in \mathscr{L}^{k+1}$), those $\theta_{ui}$ defined for $i \in V_{\ell}\setminus \mathscr{L}^{k+1}$ are used to define $\lambda_{uv}$'s in the following step.

\paragraph{Define $\lambda_{uv}$ for $\T_r^k$}: Consider two nodes $u, v$ that are leaf nodes in $\T^k_r$. Depending on whether $u, v$ are also leaf nodes in $\T_r$, there are three cases for defining $\lambda^k_{uv}$:

\begin{enumerate}
    \item If both $u$ and $v$ belong to $V_\ell$, define 
    \begin{equation}\label{eq:hse_levelk_def_1}
        \lambda^k_{uv} = \lambda_{uv}
    \end{equation}
    \item If only one of the nodes belongs to $V_\ell$, then the other node belongs to $V_\ell^k \setminus V_\ell$. In this case, we define 
    \begin{equation}\label{eq:hse_levelk_def_2}
        \lambda^k_{uv} = \theta_{uv}
    \end{equation}
    where $\theta_{uv}$ is defined in the previous step.
    \item If neither $u$ nor $v$ belongs to $V_\ell$, then both $u, v \in V_\ell^k \setminus V_\ell$. This implies that, for each of them, there must be one leaf node in $\T_r$ connected to it, denoted the leaf nodes connected to $u$ and $v$ as $\mathscr{L}(u)$ and $\mathscr{L}(v)$, accordingly. Then we define 
    \begin{equation}\label{eq:hse_levelk_def_3}
        \lambda_{uv}^k = \frac{\lambda_{\mathscr{L}(u), \mathscr{L}(v)}}{\theta_{u, \mathscr{L}(u)}\theta_{v, \mathscr{L}(v)}},
    \end{equation}
    where $\theta_{u, \mathscr{L}(u)}$ values are defined in the previous step.
\end{enumerate}

We will now demonstrate that $\lambda^k$ is HSE on $\T_r^k$. To begin, consider any three leaf nodes, $u$, $v$, and $w$, in $\T_r^k$. Let $i$, $j$, and $\ell$ denote the leaf nodes in $\T_r$ that are neighbors to, or identical with, nodes $u$, $v$, and $w$, respectively. For instance, if $u$ is also a leaf node in $\T_r$, then $i=u$; if $u$ is not a leaf node in $\T_r$, then $i$ is a leaf node in $\T_r$ that is a neighbor node of $u$, specifically, we would choose $i = \mathscr{L}(u)$, where $\mathscr{L}(u)$ is the notation used in Equation \eqref{eq:hse_levelk_def_3}. Under this setup, 
\[m(u, v, w) = m(i, j, \ell).\]
To establish the HSE of $\lambda^k$, we consider all possible cases of whether nodes $u$, $v$, and $w$ belong to $V_\ell$. We show that, under all cases, 
\begin{equation}\label{eq:hse_levelk_prop}
  \frac{\lambda^k_{uw}}{\lambda^k_{vw}} = \psi_{uv} \cdot  \frac{\lambda_{i\ell}}{\lambda_{j\ell}},   
\end{equation}
where $\varphi_{uv}$ is a constant that only depends on $i$ and $j$. Then the HSE of $\lambda^k$ comes from the fact that $\lambda$ is HSE.

\begin{enumerate}
    \item If both $u, v$ belong to $V_\ell$, Equation \eqref{eq:hse_levelk_prop} holds with $\psi_{uv} = 1$
    
        \begin{itemize}
        
            \item when $w\in V_{\ell}$, by definition in Equation \eqref{eq:hse_levelk_def_1}, 
            \[\lambda^k_{uw}/\lambda^k_{vw} = \lambda_{i\ell}/\lambda_{j\ell}\]
            
            \item when $w\in  V_k\setminus V_\ell$, we have  
            \[\lambda^k_{uw}/\lambda^k_{vw} = \theta_{iw}/\theta_{jw} = \dfrac{\theta_{iw} \theta_{w\ell}}{\theta_{jw}\theta_{w\ell}} = \lambda_{i\ell}/\lambda_{j\ell},\]
            where the first equality is by definition in Equation \eqref{eq:hse_levelk_def_2} and last equality is by Equation \eqref{eq:hse_levelk_prod}.
            
        \end{itemize}
    
    \item If only one of $u, v$ belongs to $V_\ell$, without loss of generality, assume $u\in V_{\ell}$ and $v\in V_k\setminus V_\ell$. Equation \eqref{eq:hse_levelk_prop} holds with $\psi_{ij} = \theta_{vj}$, notice that $j = \mathscr{L}(v)$, therefore $\theta_{vj}$ only depends on $v$.
        \begin{itemize}
        
            \item when $w\in V_{\ell}$,
            \[\lambda^k_{uw}/\lambda^k_{vw} = \lambda_{i\ell}/\theta_{v\ell} = \theta_{vj}\frac{\lambda_{i\ell}}{\theta_{vj} \theta_{v\ell}} = \theta_{vj} \frac{\lambda_{i\ell}}{\lambda_{j\ell}},\]
            where the first equality is by definition in Equation \eqref{eq:hse_levelk_def_1} and \eqref{eq:hse_levelk_def_2}, and the last equality is by Equation \eqref{eq:hse_levelk_prod}.
            
            \item when $w\in  V_k\setminus V_\ell$,  
            \[\lambda^k_{uw}/\lambda^k_{vw} = \frac{\theta_{iw}}{\lambda_{j\ell}/\left(\theta_{vj}\theta_{w\ell}\right)} = \theta_{vj} \frac{\theta_{iw}\theta_{w\ell}}{\lambda_{j\ell}}= \theta_{vj} \frac{\lambda_{i\ell}}{\lambda_{j\ell}}, \]
            where the first equality is by definition in Equation \eqref{eq:hse_levelk_def_2}  and \eqref{eq:hse_levelk_def_3}, and the last equality is by Equation \eqref{eq:hse_levelk_prod}.
            
        \end{itemize}
        
    \item If neither $u$ nor $v$ belongs to $V_\ell$. Equation \eqref{eq:hse_levelk_prop} holds with $\psi_{uv} = \theta_{vj}/\theta_{ui}$, notice that $j = \mathscr{L}(v)$ and $i = \mathscr{L}(u)$, therefore $\psi_{uv}$ only depends on $u, v$.
            
            \begin{itemize}
        
            \item when $w\in V_{\ell}$,
            \[\lambda^k_{uw}/\lambda^k_{vw} = \theta_{u\ell}/\theta_{v\ell} = \frac{\theta_{vj}}{\theta_{ui}} \cdot \frac{\theta_{u\ell}\theta_{ui}}{\theta_{v\ell}\theta_{vj}} = \frac{\theta_{vj}}{\theta_{ui}} \cdot \frac{\lambda_{i\ell}}{\lambda_{j\ell}},\]
            where the first equality is by definition in Equation \eqref{eq:hse_levelk_def_2}, and the last equality is by Equation \eqref{eq:hse_levelk_prod}.
            
            \item when $w\in  V_k\setminus V_\ell$,  
            \[\lambda^k_{uw}/\lambda^k_{vw} = \frac{\lambda_{i\ell}/\left(\theta_{ui}\theta_{w\ell}\right)}{\lambda_{j\ell}/\left(\theta_{vj}\theta_{w\ell}\right)} =\frac{\theta_{vj}}{\theta_{ui}} \cdot \frac{\lambda_{i\ell}}{\lambda_{j\ell}}, \]
            where the first equality is by definition in Equation \eqref{eq:hse_levelk_def_3}.
            
        \end{itemize}
\end{enumerate}

\paragraph{Verify the product of $\theta$'s equals to $\lambda_{ij}$}: We want to show that 
\begin{equation}\label{eq:hse_levelk_prod_equal}
    \lambda_{ij} = \prod_{(u, v)\in E(i, j)}\theta_{uv}
\end{equation}
for any pair of nodes $i, j \in V_\ell$. Notice that the product of $\theta_{uv}$'s defined for edges in $\T^k_r$ equals to $\lambda^k$. When at least one of nodes $i, j$ are in $V_{\ell}^k$, Equation \eqref{eq:hse_levelk_prod_equal} holds by the definition of $\lambda^k$'s in Equation \eqref{eq:hse_levelk_def_1} and \eqref{eq:hse_levelk_def_2}. When both $i$ and $j$ are in $V_{\ell} \setminus V_k$, denote their parent nodes as $u$, $v$, accordingly, then find another node $\ell \in V_\ell$.
\begin{align*}
     \prod_{(u, v)\in E(i, j)}\theta_{uv} 
     &= \theta_{iu}\,\theta_{jv}\, \lambda^k_{u, v} \\
     &= \frac{\theta_{iu}\,\theta_{jv}\, \lambda_{\mathscr{L}(u), \mathscr{L}(v)}}{\theta_{\mathscr{L}(u), u} \theta_{\mathscr{L}(v), v}}
     \quad \text{(by definition in Equation \eqref{eq:hse_levelk_def_3})}
     \\[5pt]
     &= \frac{\theta_{iu} \, \theta_{uj} \, \theta_{jv} \, \theta_{vi}\, \lambda_{\mathscr{L}(u), \mathscr{L}(v)}}{\theta_{\mathscr{L}(u), u} \, \theta_{uj} \, \theta_{\mathscr{L}(v), v} \, \theta_{vi}}
     \\[5pt]
     & = \frac{\lambda_{ij} \, \lambda_{ij} \, \lambda_{\mathscr{L}(u), \mathscr{L}(v)}}{ \lambda_{\mathscr{L}(u), j} \, \lambda_{\mathscr{L}(v), i}}
    \quad \text{(by Equation \eqref{eq:hse_levelk_prod})}\\[5pt]
    & = \lambda_{ij} \cdot \frac{\lambda_{ij} \, \lambda_{\mathscr{L}(u), \mathscr{L}(v)}}{ \lambda_{\mathscr{L}(u), j} \, \lambda_{\mathscr{L}(v), i}}\\[5pt]
    & = \lambda_{ij} \cdot \scaleto{\frac{c_{i, \, \mathscr{L}(u)}^{m\left(i,\, \mathscr{L}(u),\, j\right)}}{c_{i, \, \mathscr{L}(u)}^{m\left(i,\, \mathscr{L}(u), \,\mathscr{L}(v)\right)}}}{50pt}
    \quad \text{(since $\lambda$ is $\T_r$-HSE)}
    \\[5pt]
    & = \lambda_{ij} \quad \text{(since $m(i, \mathscr{L}(u), \mathscr{L}(v)) =  m(i, \mathscr{L}(u), j)$)}.
\end{align*}

\end{proof}

\subsubsection{Proof of Theorem \ref{thm:tsg->tse}}\label{appendix:tsg->tse}

Suppose we have three leaf nodes $i,j,\ell$ and the path from $i$ to $\ell$ passes through the median node $m$.  So,  $d(i,\ell) = d(i, m) + d(m, \ell)$. Similarly for $j$ and $\ell$.  Because $e^{x+y} = e^x e^y$, both the numerator and the denominator in the ratio \eqref{eq:hse}, contain $\exp(- d(m, \ell))$. Canceling these terms removes any dependence on $\ell$.

\newpage
\section{Additional Notes about Models}

\subsection{Negative edge weight}\label{appendix: negative_edge_weight}

In a tree $\T$, we say an edge is an \textbf{internal edge} if it connects two internal nodes, and an \textbf{external edge} if it connects a leaf node and its parent. Negative distances for internal edges and external edges are very different. The subsequent lemma is useful for demonstrating why modifying the weights of external edges has minimal impact on the overall structure.

\begin{theorem}\label{thm:similar}
Consider a tree $\T = (V, E)$ and two sets of edge weights $\{w^1_{uv}\}$ and $\{w^2_{uv}\}$ on $\T$ such that
\begin{enumerate}
    \item $w^1_{uv} = w^2_{uv}$ for any pair of non leaf nodes $u, v$,
    \item $w^1_{i, p(i)} = w^2_{i, p(i)} + \delta_i$ for any leaf nodes $i$ and its parent $p(i)$.
\end{enumerate}
Let $d^1$, $d^2$ be two additive distances defined by $w^1, w^2$, respectively. Define two matrices $\Sigma^1, \Sigma^2 \in \R^{N\times N}$ as \(\Sigma^f_{ij} = \exp(-d^f(i, j))\) for $f = 1, 2$. Then 
\begin{equation}\label{eq:similar}
    \Sigma^1 = C\Sigma^2 C,
\end{equation}
where $C$ is a diagonal matrix with 
\begin{equation}
C_{ii} = 
\begin{cases}
    \exp(-\delta_i) & \text{if $i$ is a leaf node}\\
    1 & \text{otherwise}
\end{cases}
\end{equation}

\end{theorem}
\begin{proof}
    The proof can be achieved by verifying Equation \eqref{eq:similar} directly.
\end{proof}

The following corollary is a direct result of Theorem \ref{thm:similar}.

\begin{corollary}\label{coro:negative_external}
    Consider a $\T$-Stochastic Graph with $\lambda_{ij} = \exp(-d(i, j))$ for any pair of leaf nodes $i, j \in V_\ell$. If $d(\cdot, \cdot)$ is an additive distance on $\T$ with all internal edge weights being positive (external edge weights can be negative), then there exists an additive distance $\widetilde d(\cdot, \cdot)$ on $\T$ with all edge weights being positive (external edge weights are also negative) and a constant $c>0$ such that \[\lambda_{ij} = \exp(-d(i, j)) = c\exp\left(-\widetilde d(i, j)\right)\]
\end{corollary}

\begin{proof}
Using Theorem \ref{thm:similar} and consider $w^1$ and $w^2$ such that $\delta_i \equiv 1-\min_{j\in V_\ell} w^1_{j, p(j)}$ for all $i\in V_\ell$. The proof is complete by considering $d = d^1$ and $\widetilde d = d^2$.
\end{proof}

Corollary \ref{coro:negative_external} shows that for any $\T$-Stochastic Graph with negative external edge weights, one can construct a set of new edge weights $\widetilde w_{uv} > 0$ such that $\lambda_{ij} = c\exp\left(-\widetilde d(i, j)\right)$. However, this is not the case for negative internal edge weights. 

\begin{figure}[h]
    \centering
    \includegraphics[width = 2.5in]{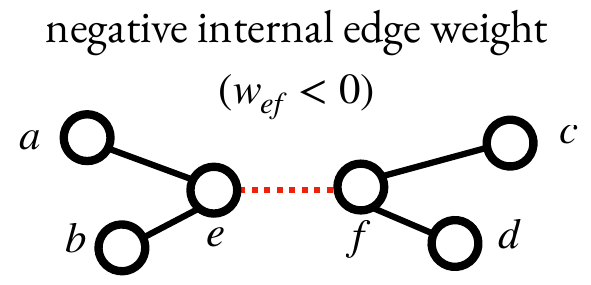}
    \caption{four leaf node tree with internal edge weights being negative}
    \label{fig:negative_internal_edge}
\end{figure}

We show this by contradiction. Consider a binary tree $\T$ with four leaf nodes shown in Figure \ref{fig:negative_internal_edge}. When the internal edge weight $w_{ef}$ is negative, we have 
\begin{equation}\label{eq:negative_internal_edge}
    d(a, b) + d(c, d) > d(a, c) + d(b, d) = d(a, d)+d(b, c).
\end{equation}
If there exists $\widetilde w_{uv}>0$ such that $\lambda_{ij} = c\exp\left(-\widetilde d(i, j)\right)$, then 
\begin{equation}\label{eq:positive_internal_edge}
    \widetilde d(a, b) + \widetilde d(c, d) < \widetilde d(a, c) + \widetilde d(b, d) = \widetilde d(a, d)+\widetilde d(b, c).
\end{equation}
However, $d(i, j) = \widetilde d(i, j) + log(1/c)$ for any pair of leaf nodes $i, j$, this means Equation \eqref{eq:negative_internal_edge} and \eqref{eq:positive_internal_edge} contradict with each other.

In fact, Equation \eqref{eq:positive_internal_edge} is also known as the four-point condition in phylogenetic studies and is widely used to identify topology structure in tree reconstruction problems.

\subsubsection{Negative edge weight in HSE}\label{appendix:neg_edge_hse}

Negative internal edge weights are inevitable in HSE graphs. One can find that the proof of Theorem \ref{thm:tsg->tse} doesn't require edge weights to be positive. In other words, any $\T$-Stochastic Graph with some negative internal edges is still HSE.

This aligns with the original Stochastic Equivalence condition proposed by \cite{holland}. 
In SBMs, diagonal elements in the connectivity matrix $B$ are not required to be larger than off-diagonal elements, that is, edges within the same block could be sparser than edges between blocks. 

Consider an SBM with 2 blocks and 4 nodes, where each block contains two nodes, and the connectivity matrix $B$ is defined as $B = \begin{bmatrix}
0.25 & 0.5 \\
0.5 & 0.25
\end{bmatrix}$. Also, consider a $\T$-Stochastic Graph with $\T$ being a four leaf node tree as in Figure \ref{fig:negative_internal_edge}, and define edge weight $w_{uv}$ as 
\begin{align*}
    w_{uv} = 
    \begin{cases}
    \log(0.5) \qquad &\text{if $u$ and $v$ are all internal nodes,} \\
    \log(2)  &\text{otherwise}. 
    \end{cases}
\end{align*}
It can be verified that these two models are equivalent, and this is an example of a $\T$-Stochastic Graph with negative internal edge weight ($\log(0.5)$ is negative) being Stochastic Equivalent and Hierarchical Stochastic Equivalent.

\subsection{None zero mean vector and general covariance matrix for Gaussian Graphical Model}\label{appendix:covariance}

Definition \ref{def:ggm} constrains $\mu = 0$ and $\Sigma_{ii} = 1$ for all $i$. This section shows that the equivalence between models can be extended to the general cases without these two assumptions. Since GGMs are only involved in the definitions of the $\T$-Graphical Blockmodel and the $\T$-Graphical-RDPG, we are going to discuss the equivalence between these two models and the $\T$-Stochastic Graph. Specifically, we show that

\begin{enumerate}
    \item When $\mu \neq 0$, Theorem \ref{thm:tsg_tgb} remains valid, Theorem \ref{thm:tsg_tg-rdpg} holds if we applies a centering step to the latent positions.
    \item When $\Sigma_{ii}\neq 1$, both Theorem \ref{thm:tsg_tgb} and Theorem \ref{thm:tsg_tg-rdpg} remain valid.
\end{enumerate}

\subsubsection{Non zero mean vector ($\mu\neq 0$)}

\paragraph{$\T$-Graphical Blockmodel}: Since Definition \ref{def:tgb} solely considers $\Sigma$ and does involve $\mu$, the validity of Theorem \ref{thm:tsg_tgb} remains unaffected when removing the constraint of $\mu=0$.

\paragraph{$\T$-Graphical RDPG}: If the mean vector $\mu\neq 0$, a centering step can be taken for all features. Specifically, instead of using $x_i$ as the latent position for node $i$, we consider $x_i - \bar{x}_i$, where $\bar{x}_i$ is the mean of all elements in $x_i$. According to the Law of Large Numbers, Theorem \ref{thm:tsg_tg-rdpg} holds as $<x_i - \bar{x}_i, x_j - \bar{x}j>/q$ converges asymptotically to $\Sigma{ij}$.

\subsubsection{Non unit variance in the covariance matrix ($\Sigma_{ii}\neq 1$)}

Since this extension only affects the $\T$-Graphical RDPG and the $\T$-Graphical Blockmodel, we only need to show one direction of the equivalence: any $\T$-Graphical RDPG (or $\T$-Graphical Blockmodel) defined by a GGM with general covariance matrix is a $\T$-Stochastic Graph.

\paragraph{$\T$-Graphical Blockmodels with $\Sigma_{ii}\neq 1$ are $\T$-Stochastic Graphs.}
In a $\T$-Graphical Blockmodel, $\mathscr{A} = ZBZ^T$ with $B = \Sigma$. If we have a general covariance matrix $\Sigma = SRS$, where $R$ is the corresponding correlation matrix and $S$ is a diagonal matrix with $S_{ii} = \sqrt{\Sigma_{ii}}$, then $\mathscr{A} = Z\Sigma Z^T = ZSRSZ^T = Z_S R Z_S^T$, where $Z_S = ZS$. Since $R$ can be considered as a covariance matrix with unit variance, and $Z_S$ is a valid block membership matrix for DCSBM, this $\T$-Graphical Blockmodel is equivalent to a $\T$-Graphical Blockmodel with a unit variance GGM. Hence any $\T$-Graphical Blockmodel with a general covariance matrix is still a $\T$-Stochastic Graph.

\paragraph{$\T$-Graphical RDPGs with $\Sigma_{ii}\neq 1$ are $\T$-Stochastic Graphs.}
Using the same notation as in Section \ref{sec:tg-rdpg}, let $s_i = \sqrt{\Sigma_{ii}}$ and denote $x_i/s_i$ as $\widetilde x_i$. By the Law of Large Numbers and Proposition \ref{prop:cov->dist}, there exists an additive distance $\widetilde d(\cdot,\cdot)$ on $\T$ such that  \[\left|\langle \widetilde  x_i, \widetilde x_j \rangle\right| \xrightarrow{q\rightarrow \infty} \exp\left(- \widetilde d(i, j)\right),\]
Consider a $\T$-Stochastic Graph with edge weights $w_{uv}$ defined as
\vspace{0.1in}
\begin{align*}
   w_{uv} =
  \begin{cases}
  \widetilde d(u, v) \qquad \qquad &\text{if both $u$ and $v$ are internal nodes,}\\
  \widetilde d(u, v) + log(1/s_u)  &\text{if $u$ is a leaf node and $v$ is the neighbor node of $u$, i.e. $v = p(u)$,}\\
  \widetilde d(u, v) + log(1/s_v) &\text{if $v$ is a leaf node and $u$ is the neighbor node of $v$, i.e. $u = p(v)$.}\\
  \end{cases}
\end{align*} 
\vspace{0.2in}
The proof is completed by the following equalities
\begin{align*}
\widetilde \lambda_{ij}^{RDPG} &= q^{-1} |\langle x_i, x_j \rangle| \\
&= s_is_jq^{-1} |\langle x^*_i, x^*_j \rangle| \\
&\rightarrow s_is_j\exp\left(- \widetilde d(i, j)\right)\\
&= \exp\left(- \widetilde d(i, j) - \log(1/s_i) - \log(1/s_j)\right)\\
& = \exp(- d(i, j) )\\
& = \lambda^{TSG}_{ij}
\end{align*}

\newpage
\section{Algorithms: \texttt{vsp} and \texttt{SpaseNJ}}

\subsection{\texttt{vsp}: consistent estimation of DCSBMs}\label{appendix:vsp}

The Vintage Sparse PCA (\texttt{vsp}) algorithm proposed by \cite{varimax_rohe} estimates $Z$ through two primary steps: spectral embedding and Varimax rotation. The first step applies eigen decomposition on $A$ to construct a matrix $U$ that shares the same column space as $Z$; the second step finds an orthogonal rotation matrix $R_U$ by maximizing the Varimax function $v(R,U)$, which is defined as
\begin{equation}
\label{varimax}
	v(R,U) = \sum_{l = 1}^k\frac{1}{n}\sum_{i = 1}^n\left([UR]^4_{il}-\left(\frac{1}{n}\sum_{q = 1}^n[UR]_{ql}^2\right)^2\right),
\end{equation}
and outputs $\widehat Z = UR_U$. It should be noted that the Varimax objective function remains unchanged when the columns of $UR$ are reordered or flipped in sign. As a result, the maximizer of Varimax is an equivalence class that allows for sign flips and column reordering.

The original \texttt{vsp} algorithm does not include a sign-adjustment step since it studies the more general semi-parametric setting where $Z_{ij}$ is allowed to be negative. In the context of $\T$-Stochastic Graphs, the column signs are identifiable as the DCSBM assumes nonnegative values for $\theta_i$. We utilize this fact to add a sign-adjustment step. This additional step removes the dependence on sign flip operation while not affecting the asymptotic performance of $\widehat Z$. Proofs can be found in Appendix \ref{appendix:sign_adjust}.

\begin{algorithm}
	\caption{\texttt{vsp} (step 1)}
	\label{vsp}
\vspace{0.03in}
	\KwIn{ \\ \Indp \Indp
        adjacency matrix $A\in \R^{n\times n}$;\\
        number of blocks $k$.}
	\KwOut{ \\ \Indp \Indp
        block membership matrix $\widehat Z \in \R^{n\times k}$.}
  \vspace{0.05in}
	Spectral Embedding: compute the top $k$ eigenvectors $u_1, \cdots, u_k\in \R^{n}$ of $A$ and place them into columns of $U\in \R^{n\times k}$\;
	Varimax Rotation:  find an orthogonal matrix $R_U$ that maximize the Varimax function, \[R_U = \mathrm{argmax}_{R\in \mathcal{O}(k)} \ v(R, U),\] where \(\mathcal{O}(k) = \{R\in \mathbb{R}_{k\times k}: R^TR = RR^T = I_k\}\) is the set of orthogonal matrices in $\R^{k\times k}$\; 
	Sign Adjustment: let $\widehat{Z} = \sqrt{n}UR_U$, calculate a diagonal matrix $\widehat D\in \R^{k\times k}$ with \(\widehat D_{jj} = \text{sgn}\left(\sum_{i = 1}^n\widehat Z_{ij}\right)\), where $\text{sgn}(0) = 0$ and if $x\neq 0$, $\text{sgn}(x) = x/|x|$; then, set $\widehat{Z} \leftarrow \widehat{Z}\widehat D$\;
\end{algorithm}

\subsection{\texttt{SparseNJ}: NJ and edge thresholding}\label{appendix:sparse_nj_nj_algorithm}

The NJ algorithm is an agglomerative clustering method that iteratively combines a pair of nodes iteratively. It consists of two main parts: the selection criterion and the updating criterion. 
Algorithm \ref{alg:nj} provides a detailed outline of the NJ procedure, with superscript $(t)$ indicating the quantities at the $t$th iteration. Algorithm \ref{alg:sparse_nj} describes the procedure of the \texttt{SparseNJ} algorithm, which involves applying NJ first and then shrinking short edges to zero.

\begin{algorithm}
\caption{Neighbor Joining}
	\label{alg:nj}
	\KwIn{
 \\ \Indp \Indp
 pairwise distances $\widehat D\in \R^{k\times k}$ between nodes in $V_o$}
	\KwOut{ \\ \Indp \Indp
 binary tree $\widehat \T_{nj} = \left(\widehat V_{nj}, \widehat E_{nj}\right)$ with edge distances $\left\{\widehat d(u, v): (u, v)\in \widehat E_{nj}\right\}$}
 \vspace{0.1in}
  Initialization: let $\widehat V_{nj} \leftarrow V_o$, $\widehat E_{nj} = \{\}$, $V_o^{(1)} = V_o$, $\widehat d^{(1)}(i, j) \leftarrow \widehat D_{ij}$ for any pair of nodes $i, j\in V_o$\;
  \vspace{0.1in}
\For{$t\gets1$ \KwTo $k-2$}{
  \textbf{Selection}: choose the pair of nodes $\left(i^{(t)}, j^{(t)}\right)\in V_o^{(t)}$ such that the following criterion is minimized \[S^{(t)}_{i,j} = (k - 2) \widehat d^{(t)}(i, j)-\sum_{m = 1}^k \widehat d^{(t)}(i, m) - \sum_{m = 1}^k \widehat d^{(t)}(j, m),\]
and combine nodes $i^{(t)}$ and $j^{(t)}$ by adding a new node $x^{(t)}$ as their parent:
  \begin{align*}
     & V_o^{(t+1)} \leftarrow V_o^{(t)} \cup \{x^{(t)}\} \setminus \left\{i^{(t)}, j^{(t)}\right\},\\
     & \widehat V_{nj} \leftarrow \widehat V_{nj} \cup \{x^{(t)}\}, \quad \widehat E_{nj} \leftarrow \widehat E_{nj} \cup \left\{(i^{(t)}, x^{(t)}), (j^{(t)}, x^{(t)})\right\};
  \end{align*}
  \textbf{Updating}: update the pairwise distances between nodes in $V_o$ as
  \begin{equation}
  \widehat d^{(t+1)}(u, \ell) =
      \begin{dcases}
        \widehat d^{(t)}(u, \ell) & \text{ if $u, \ell \neq x^{(t)}$}\\
          \frac{1}{2}\,\widehat d^{(t)}\left(u, i^{(t)}\right)+\frac{1}{2}\,\widehat d^{(t)}\left(u, j^{(t)}\right) & \text{ if $u = x^{(t)}$}
      \end{dcases}.
  \end{equation}
    Then estimate the distance of $\left(i^{(t)}, x^{(t)}\right)$ as
    \begin{equation}\label{eq:distance_update}
    \begin{aligned}
    \widehat d\left(i^{(t)}, x^{(t)}\right) &= 
    \frac{1}{2}\, \widehat d^{(t)}\left(i^{(t)}, j^{(t)}\right) 
    + \frac{1}{2}\, \widehat d^{(t)}_{-j}\left(i^{(t)}\right)
    -\frac{1}{2}\, \widehat d^{(t)}_{-i}\left(j^{(t)}\right) \\
   & \qquad  -\frac{1}{2}\, \sum_{m = 1}^{t-1}\mathds{1}\left\{i^{(t)} = x^{(m)}\right\} \, \widehat d^{(m)}\left(i^{(m)}, j^{(m)}\right),
    \end{aligned}
    \end{equation}
    where $\widehat d^{(t)}_{-j}\left(i^{(t)}\right)$ is the mean of the distances between node $i$ and any other node $u\neq i^{(t)} , j^{(t)}$, that is \[
    \widehat d^{(t)}_{-j}\left(i^{(t)}\right) = \frac{1}{k-t-1}\sum_{\scaleto{u \, \in \, 
    \left.V^{(t)}_o \middle \backslash \left\{i^{(t)}, j^{(t)}\right\}\right.}{18pt}} \widehat d^{(t)}\left(i^{(t)}, u\right).\]
    The same idea holds for estimating $\widehat d\left(j^{(t)}, x^{(t)}\right)$\;
  }
  \vspace{0.1in}
  After $k-2$ iteration, there are only two nodes left in $V_o^{(k-1)}$, denote them as $u$ and $v$, then update $\widehat E_{nj} \leftarrow \widehat E_{nj} \cup \left\{(u, v)\right\}$ and estimate the edge distance as 
  \begin{equation}\label{eq:distance_update_final}
      \widehat d(u, v) = d^{(k-1)}(u, v) - \frac{1}{2}\,d^{(k-2)}\left(i^{(k-2)}, j^{(k-2)}\right).
  \end{equation}
\end{algorithm}

\begin{algorithm}[H]
\caption{\texttt{SparseNJ} (step 3)}
	\label{alg:sparse_nj}
	\KwIn{
 \\ \Indp \Indp
 pairwise distances $\widehat D$ between nodes in $V_o$; \\
    cutoff value $\varphi$.}
	\KwOut{ \\ \Indp \Indp
 tree estimation $\widehat \T = \left(\widehat V, \widehat E\right)$.}
	\textbf{Estimation}: apply the NJ algorithm on $\widehat D$ to reconstruct a tree $\widehat \T_{nj} = \left(\widehat V_{nj}, \widehat E_{nj}\right)$ and edge distances $\left\{\widehat d(u, v): (u, v)\in \widehat E_{nj}\right\}$ (Algorithm \ref{alg:nj} provides the detailed NJ algorithm)\;
 \vspace{0.1in}
	\textbf{Shrinking}: set $\widehat V \rightarrow \widehat V_{nj}$, $\widehat E = \widehat E_{nj}$, then update them by looping through all edges in $\widehat E_{nj}$ \\ \Indp
 \vspace{0.1in}
 \For{$(u, v)\in \widehat E_{nj}$}{
 \vspace{0.05in}
\If{$\widehat d\left(u, v\right) \leq \varphi$}{
\vspace{0.05in}
combine nodes $u$ and $v$ by 
\begin{align*}
& \widehat V \leftarrow \widehat V \setminus \{u\}\\
& \widehat E \leftarrow \widehat E \mathbin{\big\backslash}
        \left\{(u, z): \forall z, \text{ s.t. } (u, z)\in \widehat E\right\} \bigcup 
        \left\{(v, z): \forall z\neq u, \text{ s.t. } (v, z)\in \widehat E\right\}
\end{align*}
}
}
\end{algorithm}

\clearpage

\section{Asymptotic Behaviour of \pps (Algorithm \ref{alg:pps})}\label{appendix: meta_algorithm_theory}

\subsection{Proof of Theorem \ref{thm:tsgdist}}\label{appendix:tsgdist}

We first show that
\begin{equation}\label{eq:B_converge}
    \left\|\widehat B - P_n^TBP_n \right\|_{F} = O_p(e_n).
\end{equation}
This can be decomposed into two parts: the first part shows that $\widehat B^{nn}$ and $\rho_n B$ are close up to a scaling matrix (Proposition \ref{prop:tildeB}); the second part (presented in this section) shows that $\widehat S^{-1}$ is a good estimate of $S$ and thus $\widehat B = \left[\widehat S^{-1}\widehat B^{nn}\widehat S^{-1}\right]_{\le 1}$ is a good estimate of $B$. The proof of Proposition \ref{prop:tildeB} can be found in Appendix \ref{appendix:tildeB}.

\begin{prop}\label{prop:tildeB}
Under the conditions in Theorem \ref{thm:tsgdist}, there exists a sequence of reordering matrix $P_n\in \mathscr{P}(k)$ such that
\[\left\|\widehat B^{nn} - \rho_nP_n^TSBSP_n\right\|_{F} = O_p\left(\rho_n e_n \right),\]
where $S$ is a random diagonal matrix with positive elements such that $S_{jj}\asymp 1$, and $\widehat B^{nn}$ can be one of three options: the original estimator proposed in Algorithm \ref{algo:tsg_dist} or the two alternatives discussed in Remark \ref{rmk:bnn_estimator}.
\end{prop}

\begin{prop}\label{prop:S}
Under the conditions in Theorem \ref{thm:tsgdist}, let $S$ and $P_n$ be the same diagonal matrix and reordering matrices as in Proposition \ref{prop:tildeB}, then
\begin{equation}\label{s-shat}
    \left \|\widehat S^{-1} - \rho_n^{-.5}P_n^TS^{-1}P_n\right\|_F = O_p\left(\rho_n^{-.5}e_n\right),
\end{equation}
\begin{equation}\label{shat}
    \left \|\widehat S^{-1}\right\|_F = O_p(\rho_n^{-.5}).
\end{equation}
\end{prop}
\begin{proof}
Let $a = \left[\widehat B^{nn}\right]_{jj}$, $b = \rho_n\left[ P_n^TSBSP_n \right]_{jj}$, by Proposition \ref{prop:tildeB}, 
\begin{equation}\label{eq: a-b}
    \left|a-b\right| = O_p\left(\rho_n e_n\right).
\end{equation}
Given that both matrices $B$ is fixed, $P_n \in \mathscr{P}(k)$, and $S_{jj}\asymp 1$, we have $|b|\asymp\rho_n$, which implies that $|a|\asymp\rho_n$. Therefore 
\begin{equation}\label{eq:convergence_a_b}
    \left|\frac{1}{a}\right| = O_p\left(\rho_n^{-1}\right),\quad \left|\frac{1}{b}\right| = O_p\left(\rho_n^{-1}\right).
\end{equation}
Since $B$ has unit diagonal elements, 
\begin{equation}\label{eq:S_diagonal}
    S_{jj} = \sqrt{\left[SBS\right]_{jj}} \,.
\end{equation}
By step 3 (scaling step) in Algorithm \ref{algo:tsg_dist}, 
\begin{equation}\label{eq:hat_S_diagonal}
    \widehat S_{jj} = \sqrt{\widehat B^{nn}_{jj}}\,.
\end{equation}
Taken Equation \eqref{eq:convergence_a_b}, \eqref{eq:S_diagonal}, and \eqref{eq:hat_S_diagonal} together, we have
\begin{align*}
\left|\left[\widehat S^{-1}\right]_{jj} - \left[\rho_n^{-.5}P_n^TS^{-1}P_n\right]_{jj}\right| &= 
    \left|\frac{1}{\sqrt{a}} - \frac{1}{\sqrt{b}}\right|  
    =\left| \frac{\sqrt{a}-\sqrt{b}}{\sqrt{ab}}\right|\\
    &= \left|\frac{a - b}{(\sqrt{a}+\sqrt{b})\sqrt{ab}}\right|\\ 
    &= O_p\left(\rho_n^{-.5}e_n\right).
\end{align*}

Since the above statement holds for $\forall j\in [k]$, we proved Equation (\ref{s-shat}). Equation (\ref{shat}) follows directly from Equation (\ref{s-shat}) by noting that  \[\left\|\rho_n^{-.5}P_n^TS^{-1}P_n\right\|_F = O_p(\rho_n^{-.5}).\]
\end{proof}

\begin{prop}
\label{prop:projection}
For $M, \widehat M \in \mathbb{R}^{p_1\times p_2}$,
If $M_{ij}\ge 0$, \[\left\|\widehat M_+ - M\right\|_{2\rightarrow \infty}\le \left\|\widehat{M} -M\right\|_{2\rightarrow \infty}.\]
If $M_{ij}\le 1$, \[\left\|\widehat M_{\le 1} - M\right\|_{F}\le \left\|\widehat{M} -M\right\|_{F}.\]
\end{prop}
\begin{proof}
If $M_{ij}\ge 0$, $\left|\left[\widehat M_+\right]_{ij} - M_{ij}\right|\le \left|\widehat M_{ij} - M_{ij}\right|$. By Proposition 6.1 in \citep{cape2019two}, $\|N\|_{2\rightarrow \infty}$ corresponds to the maximum row-wise Euclidean norm of $N$, therefore $\|\widehat M_+ - M\|_{2\rightarrow \infty}\le \|\widehat{M} -M\|_{2\rightarrow \infty}$. The argument for $M_{ij}\le 1$ is similar, just notice that $\|N\|_F = \sum_{ij}N_{ij}^2$.
\end{proof}

Proposition \ref{prop:tildeB} bounds the distance between $ B^{nn}$ and $\rho_n B$, the detailed proofs are in Appendix \ref{appendix:tildeB}. Proposition \ref{prop:S} gives the asymptotic behavior of the estimated scaling matrix. Proposition \ref{prop:projection} provides justification for element-wise projection, it shows that projection only shortens the distance between the estimate and the underlying true value. 

To simplify notations, denote $SBS$ as $B_S$.
With Proposition \ref{prop:tildeB}, \ref{prop:S}, and \ref{prop:projection}, 
the bound for $\left\|\widehat B - P^T_nBP_n\right\|_{F}$ can be derived as follows.
\begin{align*}
    & \left\|\widehat B - P_n^TBP_n \right\|_{F}
     =  \left\|\left[\widehat S^{-1}\widehat B^{nn}\widehat S^{-1}\right]_{\le 1} - P_n^TS^{-1}B_S S^{-1}P_n\right\|_{F} \\[5pt]
    & \le \left\|\widehat S^{-1}\widehat B^{nn}\widehat S^{-1} -P_n^TS^{-1}B_S S^{-1}P_n\right\|_{F} \quad (\text{by Proposition \ref{prop:projection}})\\[5pt]
    & \le \left\|\widehat S^{-1}\widehat B^{nn}\widehat S^{-1} - \rho_n\widehat S^{-1}P_n^TB_S P_n\widehat S^{-1}\right\|_{F} + \left\|\rho_n\widehat S^{-1}P_n^T B_SP_n\widehat S^{-1} - P_n^TS^{-1}B_S S^{-1}P_n\right\|_{F}\\[5pt]
    & \le \left\|\widehat S^{-1} \left(\widehat B^{nn} - \rho_nP_n^TB_SP_n \right)\widehat S^{-1} \right\|_{F}
    + \left\|\rho_n\widehat S^{-1}P_n^T B_SP_n\widehat S^{-1} - P_n^TS^{-1}P_n(P_n^TB_S P_n)P_n^TS^{-1}P_n\right\|_{F}\\[5pt]
    & \le \left\|\widehat S^{-1} \left(\widehat B^{nn} - \rho_nP_n^TB_SP_n \right)\widehat S^{-1} \right\|_{F}
    + \left\|\rho_n\widehat S^{-1}P_n^T B_S P_n\widehat S^{-1} - \rho_n^{.5}\widehat S^{-1}(P_n^TB_S P_n) P_n^TS^{-1}P_n\right\|_{F}\\
    & \qquad + \left\|\rho_n^{.5}\widehat S^{-1} (P_n^TB_SP_n) P_n^TS^{-1}P_n - P_n^TS^{-1}P_n(P_n^TB_S P_n)P_n^TS^{-1}P_n\right\|_{F}\\[5pt]
    & \le \left\|\widehat S^{-1}\right\|_{F} \left\|\widehat B^{nn} - \rho_nP_n^TB_SP_n \right\|_{F}\left\|\widehat S^{-1} \right\|_{F}
    + \left\|\rho_n^{.5}\widehat S^{-1} P_n^TB_SP_n\right\|_{F} \left\|\rho_n^{.5}\widehat S^{-1} - P_n^TS^{-1}P_n\right\|_{F}\\
    & \qquad + \left\|\rho_n^{.5}\widehat S^{-1}- P_n^TS^{-1}P_n\right\|_{F}\left\| (P_n^TB_SP_n)P_n^T  S^{-1}P_n\right\|_F\\[5pt]
    & \le O_p(\rho_n^{-.5})O_p(\rho_n e_n)O_p(\rho_n^{-.5}) + O_p(1)O_p( e_n) + O_p(e_n)O_p(1)\quad (\text{by Proposition \ref{prop:tildeB} and \ref{prop:S}})\\[5pt]
    &  =  O_p\left(e_n\right) 
\end{align*}

Now we show 
\begin{equation}\label{eq:D_converge}
    \left\|\widehat D - P_n^TDP_n \right\|_{F} = O_p(e_n).
\end{equation}
Since all elements in the covariance matrix of a GGM are positive, there exists a constant $\gamma>0$ such that $\min_{ij}B_{ij}>\gamma$. 
Since $\left\|\widehat B - P_n^TBP_n \right\|_{F} = O_p(e_n)$, the probability that $\min_{ij}\widehat B_{ij}>\gamma$ goes to one. Equation \eqref{eq:D_converge} can be proved by noticing that $\widehat D_{ij} = -\log\left(\widehat B_{ij}\right)$ and 
\[\left|\log(x)-\log(y)\right|\leq \frac{1}{\gamma}\left|x-y\right|\]
for any $x, y\in [\gamma, \infty)$ with $\gamma > 0$.

\subsection{Proof of Proposition \ref{prop:tildeB}}\label{appendix:tildeB}

The proof of Proposition \ref{prop:tildeB} requires several lemmas. Lemma \ref{lemma:norm2} is a rephrase of Proposition 6.3 in \citep{cape2019two}. The proof of Lemma \ref{lemma:Zconverge} to Lemma \ref{lemma:hatZZ_Sigma} are contained in Appendix \ref{sec:lemmas_for_Bconverge}. 

To simplify notation, denote 
\[\widehat \Sigma = \widehat Z_+^T\widehat Z_+, \ \Sigma = Z^T Z/n.\] 
We introduce the symbol $\, \widetilde{\cdot}\, $ to represent matrices that have been scaled by the matrix $C$, and utilize the subscript $P$ to denote matrices that have been reordered by the permutation matrix $P_n$. Specifically, we define
\[
    \widetilde Z = ZC, \ \ 
    \widetilde \Sigma =  \widetilde Z^T \widetilde Z/n = C\Sigma C, \ \ 
    \widetilde B = C^{-1}BC^{-1}, 
    \]
\[
    \widetilde Z_P = \widetilde ZP_n, \ \ 
    \widetilde \Sigma_P = \widetilde Z_P^T \widetilde Z_P/n = P_n^T \widetilde \Sigma P_n, \ \ 
    \widetilde B_P = P_n^T\widetilde B P_n.
    \]
Notice that 
\[\mathscr{A} = ZBZ^T = \widetilde Z \widetilde B \widetilde Z^T = \widetilde Z_P \widetilde B_P \widetilde Z^T_P, 
\]

This section shows that all 
\[\widehat B^{nn}_1 = \left[\widehat \Sigma_Z^{-1} \widehat Z_+^T A \widehat Z_+ \widehat \Sigma_Z^{-1}\right]_+, \ 
\widehat B^{nn}_2 = \frac{1}{n^2}\widehat Z_+^T A \widehat Z_+,  \ \text{and } \
\widehat B^{nn}_3 = \frac{1}{n^2}\left[\widehat Z^T A \widehat Z\right]_+, \] 
are close to $\rho_nB$ up to rows and columns permutations, and scaling matrices 
\[S_1 = C^{-1},\ S_2 = \widetilde \Sigma C^{-1},  \ \text{and } \ S_3 = \widetilde \Sigma C^{-1},\] respectively. 

\begin{lemma}\label{lemma:norm2}
(Proposition 6.3 in \cite{cape2019two})
For $M\in \mathbb{R}^{p_1\times p_2}$, then
\[\|M\|_{2\rightarrow \infty} \le \|M\|_2 \le \min\{\sqrt{p_1} \|M\|_{2\rightarrow \infty}, \sqrt{p_2} \|M^T\|_{2\rightarrow \infty}\}\]
\end{lemma}

\begin{lemma}\label{lemma:Zconverge}
There exists a sequence of permutation matrix $P_n\in \mathcal{P}(k)$, such that 
\begin{equation}\label{eq:Zconverge1}
   \left\|\widehat Z - \widetilde Z_P \right\|_{2\rightarrow \infty} = O_p\left( e_n\right), 
\end{equation}
\begin{equation}\label{eq:Zconverge2}
   \left\|\widehat Z_+ - \widetilde Z_P \right\|_{2\rightarrow \infty} = O_p\left( e_n\right). 
\end{equation}
\end{lemma}

\begin{lemma}\label{lemma:Aconverge}
$\left \| A - \mathscr{A} \right\|_2 = O_p\left( n \rho_n e_n \right)$
\end{lemma}

\begin{lemma}\label{lemma:Bconstant}
$\left\|\widetilde B_P\right\|_2 = O_p(1)$
\end{lemma}

\begin{lemma}\label{lemma:Znorm}

$\left\|\widetilde
Z_P\right\|_2 = O_p\left(n^{.5}\right)$
\end{lemma}

\begin{lemma}\label{lemma:SigmaZnorm}
$\left\|\widehat Z\right\|_2 = O_p(n^{.5})$, $\left\|\widehat Z_+\right\|_2 = O_p(n^{.5})$, 
$\left\|
\widehat \Sigma^{-1} \widehat Z_+^T \right\|_2 = O_p\left(n^{-.5}\right)$
\end{lemma}

\begin{lemma}\label{lemma:hatZZ_Sigma}
    $\left\|\widehat Z_+^T\widetilde Z_P-n\widetilde \Sigma_P\right\|_2 = O_p\left(ne_n\right)$,
    $\left\|\widehat Z^T\widetilde Z_P-n\widetilde \Sigma_P\right\|_2 = O_p\left(ne_n\right)$,
\end{lemma}

\begin{align*}
 & \left\|\widehat B^{nn}_1 - \rho_n P_n^T S_1 B S_1 P_n\right\|_{2} 
  = 
 \left\|
 \left[\widehat \Sigma^{-1} \widehat Z_+^T A \widehat Z_+ \widehat \Sigma^{-1}\right]_+ - \rho_n\widetilde B_P\right\|_{2} \\
 & \le
 \sqrt{k}\left\|\widehat \Sigma^{-1} \widehat Z_+^T A \widehat Z_+ \widehat \Sigma^{-1} - \rho_n\widetilde B_P\right\|_{2} \quad (\text{Proposition \ref{prop:projection} and Lemma \ref{lemma:norm2}})\\
 & \le \sqrt{k} 
\left\|
\widehat \Sigma^{-1} \widehat Z_+^T A \widehat Z_+ \widehat \Sigma^{-1} 
- \widehat \Sigma^{-1} \widehat Z_+^T \mathscr{A} \widehat Z_+ \widehat \Sigma^{-1}
\right \|_{2}
+
\sqrt{k}\left\|\widehat \Sigma^{-1}
\widehat Z_+^T
\mathscr{A}
\widehat Z_+
\widehat \Sigma^{-1} - \rho_n\widetilde B_P\right\|_{2}\\
 & = \sqrt{k} 
\left\|
\widehat \Sigma^{-1} \widehat Z_+^T A \widehat Z_+ \widehat \Sigma^{-1} 
- \widehat \Sigma^{-1} \widehat Z_+^T \mathscr{A} \widehat Z_+ \widehat \Sigma^{-1}
\right \|_{2} + 
\sqrt{k}\rho_n\left\|\widehat \Sigma^{-1}
\widehat Z_+^T \widetilde Z_P\widetilde B_P\widetilde Z_P^T
\widehat Z_+
\widehat \Sigma^{-1} - \widetilde B_P\right\|_{2}\\
& \le 
\sqrt{k}\left\|
\widehat \Sigma^{-1} \widehat Z_+^T (A - \mathscr{A}) \widehat Z_+ \widehat \Sigma^{-1} \right\|_{2}
+
\sqrt{k}\rho_n\left\|\widehat \Sigma^{-1}
\widehat Z_+^T
\widetilde Z_P \widetilde B_P \widetilde Z^T_P
\widehat Z_+
\widehat \Sigma^{-1} - \widetilde B_P \widetilde Z_P^T\widehat Z_+\widehat \Sigma^{-1}
\right\|_{2}\\
& \qquad + \sqrt{k}\rho_n\left\|\widetilde B_P \widetilde Z^T_P\widehat Z_+\widehat \Sigma^{-1} - \widetilde B_P
\right\|_{2}\\
&\le
\sqrt{k}\left\|
\widehat \Sigma^{-1} \widehat Z_+^T (A - \mathscr{A}) \widehat Z_+ \widehat \Sigma^{-1} \right\|_{2}
+
\sqrt{k}\rho_n\left\|\widehat \Sigma^{-1}
\widehat Z_+^T
\left(\widetilde Z_P - \widehat Z_+\right) \widetilde B_P \widetilde Z_P^T
\widehat Z_+
\widehat \Sigma^{-1}
\right\|_{2}\\
& \qquad + \sqrt{k}\rho_n\left\|\widetilde B_P \left(\widetilde Z_P - \widehat Z_+\right)^T\widehat Z_+\widehat \Sigma^{-1} 
\right\|_{2}\\
&\leq
\sqrt{k}\left\|
\widehat \Sigma^{-1} \widehat Z_+^T\right\|_2^2 \left\|A - \mathscr{A}\right\|_2 
+ \sqrt{k}\rho_n\left\|
\widehat \Sigma^{-1} \widehat Z_+^T\right\|^2_{2}\left\|\widetilde Z_P - \widehat Z_+\right\|_{2}\left\|\widetilde B_P\right\|_2 \left\|\widetilde Z_P\right\|_2\\
& \qquad + \sqrt{k}\rho_n\left\|\widetilde B_P\right\|_{2} \left\|\widetilde Z_P - \widehat Z_+\right\|_{2}\left\|\widehat \Sigma^{-1}\widehat Z_+^T\right\|_2 \\
&\le O_p(n^{-1})O_p\left(n\rho_ne_n\right)
+\rho_n O_p(n^{-1})O_p\left(n^{.5} e_n\right)O_p(1)O_p(n^{.5})+\rho_n O_p(n^{-1})O_p(1)O_p\left(n^{.5}e_n\right)O_p(n^{-.5}) \\
&\qquad (\text{Lemma \ref{lemma:norm2} to Lemma \ref{lemma:hatZZ_Sigma}, the transformation from $\Vert \cdot \Vert_2$ to $\Vert \cdot \Vert_{2\rightarrow \infty}$ introduce an extra $n^{.5}$})\\
&=O_p\left(\rho_n e_n\right)
\end{align*}

\begin{align*}
 &\left\|\widehat B^{nn}_2 - \rho_n P_n^T S_2 B S_2 P_n\right\|_{2} 
  = 
    \left\|
        \frac{1}{n^2}\widehat Z_+^T A \widehat Z_+ 
        - \rho_n P_n^T\widetilde \Sigma \widetilde B \widetilde \Sigma P_n 
    \right\|_{2} \\
 & \le 
    \left\| 
        \frac{1}{n^2} \widehat Z_+^T A \widehat Z_+ 
        - \frac{1}{n^2} \widehat Z_+^T \mathscr{A} \widehat Z_+ 
    \right \|_{2}
    +\left\| 
        \frac{1}{n^2} \widehat Z_+^T \mathscr{A} \widehat Z_+ 
        - \rho_n \widetilde \Sigma_P \widetilde B_P \widetilde \Sigma_P
    \right\|_{2}\\
 & = 
    \frac{1}{n^2}\left\|
        \widehat Z_+^T A \widehat Z_+
        - \widehat Z_+^T \mathscr{A} \widehat Z_+
    \right \|_{2} 
    + \rho_n\left\|
        \frac{1}{n^2} \widehat Z_+^T \widetilde Z_P\widetilde B_P\widetilde Z_P^T \widehat Z_+ 
        - \widetilde \Sigma_P \widetilde B_P \widetilde \Sigma_P
    \right\|_{2}\\
& \le 
    \frac{1}{n^2}\left\| 
        \widehat Z_+^T (A - \mathscr{A}) \widehat Z_+ 
    \right\|_{2}
    + \rho_n\left\|
        \frac{1}{n^2} \widehat Z_+^T \widetilde Z_P \widetilde B_P \widetilde Z^T_P \widehat Z_+ 
        - \frac{1}{n}\widetilde \Sigma_P \widetilde B_P \widetilde Z_P^T\widehat Z_+
    \right\|_{2} 
    + \rho_n\left\|
        \frac{1}{n}\widetilde \Sigma_P \widetilde B_P \widetilde Z_P^T\widehat Z_+ 
        - \widetilde \Sigma_P\widetilde B_P\widetilde \Sigma_P
    \right\|_{2} \\
&\le 
    \frac{1}{n^2}\left\| 
    \widehat Z_+^T (A - \mathscr{A}) \widehat Z_+
    \right\|_{2}
    + \frac{\rho_n}{n^2}\left\|
        \left(\widehat Z_+^T\widetilde Z_P-n\widetilde \Sigma_P\right) \widetilde B_P \widetilde Z_P^T\widehat Z_+
    \right\|_{2} 
    + \frac{\rho_n}{n}\left\|
        \widetilde \Sigma_P\widetilde B_P\left(\widetilde Z_P^T\widehat Z_+ - n\widetilde \Sigma_P \right)
    \right\|_{2}\\
&\leq 
    \frac{1}{n^2} 
        \left\|
            \widehat Z_+
        \right\|_2^2 
        \left\|
            A - \mathscr{A}
        \right\|_2 
    + \frac{\rho_n}{n^2} \left\|
            \widehat Z_+^T\widetilde Z_P-n\widetilde \Sigma_P
        \right\|_{2}
        \left\|
            \widetilde B_P
        \right\|_2 
        \left\|
            \widetilde Z_P
        \right\|_2
        \left\|
            \widehat Z_+
        \right\|_2\\
& \qquad 
    + \frac{\rho_n}{n}
        \left\|
             \widetilde \Sigma_P \widetilde B_P
        \right\|_{2} 
        \left\|
            \left(\widehat Z_+^T\widetilde Z_P-n\widetilde \Sigma_P\right)^T
        \right\|_{2} \\
&\le 
    \frac{1}{n^2} 
        O_p\left(n\right) 
        O_p\left(n\rho_ne_n\right)
    +\frac{\rho_n}{n^2}
        O_p\left(ne_n\right)
        O_p\left(1\right)
        O_p\left(n^{.5}\right)
        O_p\left(n^{.5}\right)
    +\frac{\rho_n}{n}
        O_p\left(1\right)
        O_p\left(n e_n\right)\\
&\qquad (\text{Lemma \ref{lemma:norm2} to Lemma \ref{lemma:hatZZ_Sigma}})\\
&=O_p\left(\rho_n e_n\right)
\end{align*}

\begin{align*}
 &\left\|\widehat B^{nn}_3 - \rho_n P_n^T S_3 B S_3 P_n\right\|_{2} = 
 \left\|\frac{1}{n^2} \left[ \widehat Z^T A \widehat Z\right]_{+} - \rho_n P_n^T\widetilde \Sigma \widetilde B \widetilde \Sigma P_n \right\|_{2} \\
 &\leq 
 \sqrt{k}\left\| \frac{1}{n^2} \widehat Z^T A \widehat Z - \rho_n\widetilde \Sigma_P \widetilde B_P \widetilde \Sigma_P\right\|_{2} (\text{Proposition \ref{prop:projection} and Lemma \ref{lemma:norm2}})\\
 & \le 
\sqrt{k} \left\| \frac{1}{n^2} \widehat Z^T A \widehat Z
- \frac{1}{n^2} \widehat Z^T \mathscr{A} \widehat Z
\right \|_{2}
+
\sqrt{k} \left\| \frac{1}{n^2}
\widehat Z^T
\mathscr{A}
\widehat Z - \rho_n\widetilde \Sigma_P \widetilde B_P \widetilde \Sigma_P\right\|_{2}\\
 & = \frac{\sqrt{k}}{n^2}\left\| \widehat Z^T (A - \mathscr{A}) \widehat Z \right\|_{2} + 
\sqrt{k} \rho_n\left\|
\frac{1}{n^2} \widehat Z^T \widetilde Z_P\widetilde B_P\widetilde Z_P^T
\widehat Z - \widetilde \Sigma_P \widetilde B_P \widetilde \Sigma_P\right\|_{2}\\
& \leq
    \frac{\sqrt{k}}{n^2}\left\| 
        \widehat Z^T (A - \mathscr{A}) \widehat Z 
    \right\|_{2}
    + \sqrt{k}\rho_n\left\|
        \frac{1}{n^2} \widehat Z^T \widetilde Z_P \widetilde B_P \widetilde Z^T_P \widehat Z 
        - \frac{1}{n}\widetilde \Sigma_P \widetilde B_P \widetilde Z_P^T\widehat Z
    \right\|_{2} 
    + \sqrt{k}\rho_n\left\|
        \frac{1}{n}\widetilde \Sigma_P \widetilde B_P \widetilde Z_P^T\widehat Z 
        - \widetilde \Sigma_P\widetilde B_P\widetilde \Sigma_P
    \right\|_{2} \\
&\le 
    \frac{\sqrt{k}}{n^2}\left\| 
    \widehat Z^T (A - \mathscr{A}) \widehat Z
    \right\|_{2}
    + \frac{\sqrt{k}\rho_n}{n^2}\left\|
        \left(\widehat Z^T\widetilde Z_P-n\widetilde \Sigma_P\right) \widetilde B_P \widetilde Z_P^T\widehat Z
    \right\|_{2} 
    + \frac{\sqrt{k}\rho_n}{n}\left\|
        \widetilde \Sigma_P\widetilde B_P\left(\widetilde Z_P^T\widehat Z - n\widetilde \Sigma_P \right)
    \right\|_{2}\\
&\leq 
    \frac{\sqrt{k}}{n^2} 
        \left\|
            \widehat Z
        \right\|_2^2 
        \left\|
            A - \mathscr{A}
        \right\|_2 
    + \frac{\sqrt{k}\rho_n}{n^2} \left\|
            \widehat Z^T\widetilde Z_P-n\widetilde \Sigma_P
        \right\|_{2}
        \left\|
            \widetilde B_P
        \right\|_2 
        \left\|
            \widetilde Z_P
        \right\|_2
        \left\|
            \widehat Z
        \right\|_2\\
& \qquad 
    + \frac{\rho_n}{n}
        \left\|
             \widetilde \Sigma_P \widetilde B_P
        \right\|_{2} 
        \left\|
            \left(\widehat Z^T\widetilde Z_P-n\widetilde \Sigma_P\right)^T
        \right\|_{2} \\
&\le 
    \frac{1}{n^2} 
        O_p\left(n\right) 
        O_p\left(n\rho_ne_n\right)
    +\frac{\rho_n}{n^2}
        O_p\left(ne_n\right)
        O_p\left(1\right)
        O_p\left(n^{.5}\right)
        O_p\left(n^{.5}\right)
    +\frac{\rho_n}{n}
        O_p\left(1\right)
        O_p\left(n e_n\right)\\
&\qquad (\text{Lemma \ref{lemma:norm2} to Lemma \ref{lemma:hatZZ_Sigma}})\\
&=O_p\left(\rho_n e_n\right)
\end{align*}

\subsection{Proof of Lemma \ref{lemma:Zconverge} to \ref{lemma:hatZZ_Sigma}}\label{sec:lemmas_for_Bconverge}

\subsubsection{Proof of Lemma \ref{lemma:Zconverge}, \ref{lemma:Aconverge}, and \ref{lemma:Bconstant}}
Lemma \ref{lemma:Aconverge} and Equation \eqref{eq:Zconverge1} in Lemma \ref{lemma:Zconverge} are assumptions in Theorem \ref{thm:tsgdist}. Equation \eqref{eq:Zconverge2} in Lemma \ref{lemma:Zconverge} holds by Proposition \ref{prop:projection}. 
Lemma \ref{lemma:Bconstant} is evident given that $C$, $B$ are fixed matrix and $P_n\in \mathscr{P}(k)$.

\subsubsection{Proof of Lemma \ref{lemma:Znorm}}
\begin{lemma}\label{Matrix Bernstein Inequality} (Matrix Bernstein Inequality \citep{tropp2012user})
Let $X_1, X_2, \cdots, X_n$ be independent random $n\times n$ symmetric matrix. Assume $\|X_i-\mathbb{E}(X_i)\|_2\le M, \forall i$. Write $v^2 = \|\sum_i Var(X_i)\|_2$, $X = \sum_i X_i$. Then for any $a > 0$, 
\[\mathbb{P}(\|X-\mathbb{E}(X)\|_2 \ge a)\le 2n\exp\left(-\frac{a^2}{2v^2+2Ma/3}\right)\]
\end{lemma}

Let $\widetilde Z_i$ be the $i$th row of $\widetilde Z$. Define $X_i = \widetilde Z_i^T \widetilde Z_i$, $X = \sum_i X_i$. Notice that $X_i$ are independently identically distributed and $X = \widetilde Z^T\widetilde Z$.
Since $Z_{ij}$ follows bounded distribution and $C$ is a fixed scaling matrix. There exist constants $C_1, C_2$ such that $\widetilde Z_{1j}^2 \le C_1, \ Var(\widetilde Z_{1j}^2)\le C_2,\ \forall j \in [k]$. Therefore, \[\|X_i - E(X_i)\|_2  = \left\|X_i - \widetilde \Sigma\right\|_2\le C_1,\]
\begin{align*}
    v^2 &= \left\|\sum_i Var(X_i)\right\|_2\\
    & = n \|Var(X_1)\|_2 \\
    & = n \max_{j} Var(\widetilde Z_{1j}^2)\\
    & \le nC_2.
\end{align*}
Denote $\E\left(\widetilde Z_i^T \widetilde Z_i\right)$ as $\mathscr{Z}$, then by Lemma \ref{Matrix Bernstein Inequality},
\[\mathbb{P}\left(\left\|\widetilde Z^T \widetilde Z- n\mathscr{Z} \right\|_2 \ge a\right)\le 2k\exp\left(-\frac{a^2}{2nC_2+2C_1a/3}\right).\]
This implies $\left\|\widetilde Z^T \widetilde Z- n\mathscr{Z} \right\|_2 = O_p\left(n^{.5}\right)$, thus
\begin{align*}
    \left\|\widetilde Z\right\|_2
    &= \sqrt{\left\|\widetilde Z^T\widetilde Z\right\|_2}\\
    &\le \sqrt{\left\|\widetilde Z^T\widetilde Z - n\mathscr{Z} \right\|_2 + \left\|n\mathscr{Z} \right\|_2}\\
    & =  \sqrt{O_p\left(n^{.5}\right) +  O_p\left(n\right)}\\
    & = O_p\left(n^{.5}\right).
\end{align*} 
Since $\widetilde Z_P = \widetilde Z P_n$, $\left\|\widetilde Z_P\right\|_2 = O_p(n^{.5})$

\subsubsection{Proof of Lemma \ref{lemma:SigmaZnorm}}
\begin{lemma}\label{inverse}
(Equation 1.1 in \citep{el2002inversion}) For $M, \ \Delta M \in \mathbb{R}^{p_1\times p_1}$, $\det M \neq 0$, $\det (M + \Delta M) \neq 0$, then
\[\left\|(M + \Delta M)^{-1} - M^{-1}\right\|_2 \le \left\|M^{-1}\right\|^2_2 \left\|\Delta M\right\|_2\]
\end{lemma}

The proof of Lemma \ref{lemma:Znorm} implies 
\begin{equation}\label{ztz-i}
\left\|\widetilde Z_P^T\widetilde Z_P/n - \mathscr{Z}_P\right\|_2 = O_p\left(n^{-.5}\right),
\end{equation} 
where $\mathscr{Z}_P = P_n^T \mathscr{Z} P_n = \E(\widetilde Z^T_P \widetilde Z_P/n)$. Also, 

\begin{align*}
\left\|\widehat Z_+\right\|_2 &\le \left\|\widehat Z_+ -\widetilde Z_P \right\|_2 +\left\|\widetilde Z_P\right\|_2\\
& \le \sqrt{n}\left\|\widehat Z_+ -\widetilde Z_P \right\|_{2\rightarrow\infty} +\left\|\widetilde Z_P\right\|_2\\
& = O_p\left(n^{.5}e_n\right) +O_p\left(n^{.5}\right)\\
& = O_p\left(n^{.5}\right), 
\end{align*} 
and 
\begin{align*}
\left\|\widehat \Sigma - \widetilde Z_P^T \widetilde Z_P\right\|_2 &= \left\|\widehat Z_+^T\widehat Z_+ - \widetilde Z_P^T \widetilde Z_P\right\|_2 \\
    & \le \left\|\widehat Z_+^T\widehat Z_+ - \widehat Z^T_+ \widetilde Z_P\right\|_2 + \left\|\widehat Z_+^T\widetilde Z_P - \widetilde Z_P^T \widetilde Z_P\right\|_2\\
    & \le \left\|\widehat Z_+\right\|_2\left\|\widehat Z_+ - \widetilde Z_P\right\|_2+\left\|\widehat Z_+ - \widetilde Z_P\right\|_2 \left\|\widetilde Z_P\right\|_2\\
    & \le O_p\left(n^{.5}\right) O_p\left(n^{.5} e_n\right) +  O_p\left(n^{.5} e_n\right)O_p\left(n^{.5}\right)\\
    & \le O_p\left(n\right).
\end{align*}

Equation (\ref{ztz-i}) implies that for $\forall \epsilon > 0$, $\exists \ C_{\epsilon}$ such that 
\[\sup_n\mathbb{P}\left(\left\|\widetilde Z_P^T\widetilde Z_P/n - \mathscr{Z}_P\right\|_2\ge C_{\epsilon}n^{-.5}\right)< \epsilon.\]
Therefore by Weyl's inequality, \[\left\|\left(\widetilde Z^T_P\widetilde Z_P/n\right)^{-1}\right\|_2 = O_p(1), \ \left\|\left(\widetilde Z_P^T\widetilde Z_P\right)^{-1}\right\|_2 = O_p\left(n^{-1}\right).\]
By Lemma \ref{inverse}, 
\begin{align*}
    \left\|\widehat \Sigma^{-1}\right\|_2 
    &\le \left\|\widehat \Sigma^{-1} - \left(\widetilde Z_P^T \widetilde Z_P\right)^{-1}\right\|_2 +  \left\|\left(\widetilde Z_P^T\widetilde Z_P\right)^{-1}\right\|_2\\
    &\le \left\|\left(\widetilde Z_P^T\widetilde Z_P\right)^{-1}\right\|_2^2\left\|\widehat \Sigma - \widetilde Z_P^T \widetilde Z_P\right\|_2 + O_p\left(n^{-1}\right)\\
    &\le O_p\left(n^{-2}\right)O_p\left(n\right)+ O_p\left(n^{-1}\right)\\
    & = O_p\left(n^{-1}\right).
\end{align*}

Therefore \[\left\|
\widehat \Sigma^{-1} \widehat Z_+^T \right\|_2 \leq \left\|\widehat \Sigma^{-1}\right\|_2\left\|\widehat Z_+\right\|_2\le O_p\left(n^{-.5}\right)\] 

\subsubsection{Proof of Lemma \ref{lemma:hatZZ_Sigma}}

This is a direct result of Lemma \ref{lemma:norm2}, Lemma \ref{lemma:Zconverge}, and Lemma \ref{lemma:Znorm}.
\begin{align*}
    \left\|\widehat Z_+^T\widetilde Z_P-n\widetilde \Sigma_P\right\|_2 
    &= \left\|\widehat Z_+^T\widetilde Z_P-\widetilde Z_P^T\widetilde Z_P\right\|_2
 = \left\|\widehat Z_+ - \widetilde Z_P\right\|_2 \left\|\widetilde Z_P\right\|_2 \\
    & \leq \sqrt{n}\left\|\widehat Z_+ - \widetilde Z_P\right\|_{2\rightarrow \infty} \left\|\widetilde Z_P\right\|_2 \\
    & = O_p(n^{.5}e_n)O_p(n^{.5})\\
    & = O_p(ne_n).
\end{align*}

\subsection{Stableness of NJ}\label{appendix:nj}

To discuss the performance of NJ, we remove the identifiability assumption in Remark \ref{rmk:identifiability}, which constrains all internal nodes to have degree $\geq 3$, and consider $\T$ as any general tree graph. Given $\T = (V, E)$, distance-based recovery methods take the pairwise distances $\widehat D$ between a subset of nodes $V_o\subseteq V$ as input, and reconstruct a tree structure $\widehat \T = \left(\widehat V, 
\widehat E\right)$. We refer to nodes in $V_o$ as ``observed nodes'' since the distance information between them is observed. The reconstructed tree $\widehat \T$ should always contain all the observed nodes, i.e., $V_o\subseteq \widehat V$. 
We say two trees $\T = (V, E)$ and $\widehat \T = \left(\widehat V, 
\widehat E\right)$ are \textbf{equivalent}, written as $\T \sim \widehat \T$, if there exists an isomorphism mapping $f:V\rightarrow \widehat V$ that preserves adjacency structure, that is, $\widehat E = \left\{(f(u), f(v)): (u, v)\in E\right\}$, and the observed nodes, that is, $f(u) = u$ for any $u\in V_o$.

When $\T$ and $\widehat \T$ are not equivalent, it is important to have a  measure of the proximity between them. To this end, we introduce the definition of ``split'' following \citep{bandelt1986reconstructing, atteson1997performance}. This definition is formulated with respect to the observed set $V_o$.

\begin{definition}
Let $\T = (V, E)$ be a tree with an observed set $V_o$. Removing an edge $e\in E$ from $\T$ splits it into two components. Denote $s(\T-e, V_o)$ as how nodes in $V_o$ are distributed between these two components, then $s(\T-e, V_o)$ is the split of edge $e$. Let $S(\T, V_o)$ be the set of all the splits in tree $\T$, i.e., $S(\T, V_o) = \{s(\T-e, V_o): e\in E\}$.
\end{definition}

\begin{prop}\label{prop:tree_shape}
(\cite{bandelt1986reconstructing}) If
$\{v\in V: deg(v)<3\}\subseteq V_o$ and $\left\{v\in \widehat V: deg(v)<3\right\} \subseteq V_o$, then 
\(S\left(\T, V_o\right) = S\left(\widehat \T, V_o\right) 
\text{ if and only if } 
\ \T \sim \widehat \T.\)
\end{prop}

When $V_o$ contains all the nodes with degree smaller than three, the splits of a tree $\T$ completely parametrize its topology. When $S(\T, V_o) = S\left(\widehat \T, V_o\right)$, two trees are equivalent; when $S(\T, V_o)\neq S\left(\widehat \T, V_o\right)$, the cardinality of their interaction set
provides reasonable measure about how close $\widehat \T$ and $\T$ are to each other. For example, the popular Robinson-Foulds distance \citep{robinson1981comparison} calculates the number of splits that differ in $S(\T, V_o)$ and $S\left(\widehat \T, V_o\right)$.

\begin{definition}
For a tree $\T = (V, E)$ and an estimation $\widehat \T = \left(\widehat V, \widehat E\right)$. We say $\widehat \T$ correctly reconstruct edge $e\in E$ if there exists an edge $ e^\prime \in \widehat E$ such that $s(\T-e, V_o) = s\left(\widehat \T -e^\prime, V_o\right)$.
\end{definition}

\subsubsection{Proof for Theorem \ref{thm:cutoff}}\label{appendix:cutoff}

The proof of Theorem \ref{thm:cutoff} is established using Lemma \ref{lemma:edge/4}, which is an extension of Theorem 23 in \citep{mihaescu2009neighbor} to allow $V_o$ to contain non-leaf nodes.

\begin{lemma}\label{lemma:edge/4}
Given any tree $\T = (V, E)$ and the estimated pairwise distances $\widehat D$ between nodes in $V_o\subseteq V$, let \(\delta = \max_{u, v} \left|\widehat D_{uv} - D_{uv}\right|\). If \(\{i\in V : deg(i) < 3\} \subseteq V_o\), then NJ correctly reconstruct any edge $e = (u, v)$ with 
\[d(u, v)> 4\delta, \] 
where $d(\cdot, \cdot)$ is the true edge distances in $\T$.
\end{lemma}

\begin{proof}
According to Theorem 23 in \citep{mihaescu2009neighbor}, Lemma \ref{lemma:edge/4} is valid when $V_o = V_\ell$, and all internal nodes in $\T$ have degrees $\geq 3$. To remove these restrictions, we introduce another tree $\widetilde \T = \left(\widetilde V, \widetilde E\right)$ and its corresponding true distance $\widetilde d(\cdot, \cdot)$ by converting all internal nodes in $V_o$ to leaf nodes with zero edge weights. Initially, we set $\widetilde \T = \T$, and $\widetilde d(\cdot, \cdot) = d(\cdot, \cdot)$. Then, we remove all nodes $u\in V_o\setminus V_\ell$ by relabeling them as $\widetilde u$ in $\widetilde \T$ and $\widetilde d(\cdot, \cdot)$. Subsequently, we add all $u\in V_o\setminus V_\ell$ back as leaf nodes by updating
\begin{equation*}
\widetilde V \leftarrow \widetilde V \cup \{u: u\in V_o\setminus V_\ell\},\quad \widetilde E = E \cup \{(\widetilde u, u): u\in V_o\setminus V_\ell\}, \quad \text{and }\widetilde d\left(\widetilde u, u\right) = 0, \, \text{for }\forall u \in V_o\setminus V_\ell.
\end{equation*}
After this process, all internal nodes in $\widetilde \T$ have degree $\geq 3$, and the set of leaf nodes in $\widetilde \T$ is identical to the set of observed nodes in $\T$. 

\end{proof}

Let $\widehat \T_{nj} = \left(\widehat V_{nj}, \widehat E_{nj}\right)$ with $\left\{\widehat d(u, v): (u, v)\in \widehat E_{nj}\right\}$ be the output of the NJ algorithm. By Lemma \ref{lemma:edge/4}, for any edge $e\in E$, there exists an edge $e^{\prime} \in \widehat E_{nj}$ such that \[S(\T-e, V_o)=S\left(\widehat\T_{nj} - e^{\prime}, V_o\right).\]
To better analyze the performance of NJ, it is essential to gain an intuitive understanding of the structure of $\widehat \T_{nj}$. Firstly, $\widehat \T_{nj}$ is a binary tree where all nodes in $V_o$ are labeled as leaf nodes, resulting in $\widehat E_{nj}$ typically containing more edges than $E$. We can classify edges in $\widehat E_{nj}$ into two types. The first type is those correctly reconstructed edges, they create the same splits as some edge in $E$; the second type is the wrongly reconstructed edges, they create splits that differ from all edges in $E$. In the following paragraphs, we abuse the notation a little bit, for two edges $e\in E$ and $e^{\prime} \in \widehat E_{nj}$ that create the same splits, we denote both of them as $e$. In other word, for any edge $e\in \widehat E_{nj}$, if we say $e\in E$, that means $e$ is corrected reconstructed.

While Lemma \ref{lemma:edge/4} implies that all edges in $E$ can be correctly reconstructed by NJ, we want to further study the estimated edge distances $\widehat d$. However, the topology structure of $\T$ and $\widehat \T_{nj}$ are different, creating difficulties to compare $\widehat d$ and $d$. To this end, we construct a set of edge distances $\left\{\widetilde d(e): e \in \widehat E_{nj}\right\}$ for tree $\widehat\T_{nj}$ to serve as a bridge between $\widehat d$ and $d$. Specifically, 
\begin{equation}
        \widetilde d(e) \begin{dcases}
        d(e) & \text{ if $e\in E$ ($e$ is correctly reconstructed),
        }\\
        0  & \text{ otherwise.}
    \end{dcases}
\end{equation}
Then $\widetilde d(u, v) = d(u, v)$ for any pair of nodes $(u, v)\in V_o$. Now, we can consider $\widehat \T_{nj}$ and $\widetilde d$ (instead of $\T$ and $d$) as the underlying truth and study how $\widehat d$ differs from $\widetilde d$. 

To analyze how errors propagate at each step, we follow the procedure in \citep{atteson1997performance} to create a series of ground truth trees $\left\{\T^{(t)}\right\}_{t = 1}^{k-2}$ and their corresponding distances $\left\{ d^{(t)}\right\}_{t = 1}^{k-2}$. To accommodate cases where $\T$ is not binary and $V_o$ may contain non-leaf nodes, we set $\T^{(1)}$ to be $\widehat \T_{nj}$ and $d^{(1)}$ to be $\widetilde d$. Further details can be found in Lemma 6 of \citep{atteson1997performance}. Lemma \ref{lemma:error_not_accumulate} below shows that the errors in the estimated distances do not accumulate through iterations.

\begin{lemma}\label{lemma:error_not_accumulate}
(Lemma 7 in \cite{atteson1997performance})
Estimation error of pairwise distance doesn't accumulate through iteration 
\[\max_{i',j'}\left| \widehat d^{(t)}(i', j') - d^{(t)}(i', j')\right|\leq \max_{i, j}\left| \widehat d^{(1)}(i, j) - d^{(1)}(i, j)\right|= \delta, \ \forall t\geq 1\]
\end{lemma}

Now we can assess the error in the estimated distances. Since the edge lengths in $\widehat \T_{nj}$ are estimated using Equation \eqref{eq:distance_update} or \eqref{eq:distance_update_final}, we only need to focus on these two equations. For Equation \eqref{eq:distance_update}, it's important to note that $x^{(t)}$ are different for each iterations $t$, thus the last summation term can have at most one non-zero entry. Assuming the $h$th term is non-zero, then the estimated distance becomes:

\[ \widehat d\left(i^{(t)}, x^{(t)}\right) = \frac{1}{2}\, \left(\widehat d^{(t)}\left(i^{(t)}, j^{(t)}\right) 
    + \, \widehat d^{(t)}_{-j}\left(i^{(t)}\right)
    - \, \widehat d^{(t)}_{-i}\left(j^{(t)}\right)
    - \, \widehat d^{(h)}\left(i^{(h)}, j^{(h)}\right)\right),\]
Then the error term
\begin{align*}
  \left|\widehat d\left(i^{(t)}, x^{(t)}\right) - \widetilde d\left(i^{(t)}, x^{(t)}\right)\right| &= \frac{1}{2}\left|\left(\widehat d^{(t)}\left(i^{(t)}, j^{(t)}\right) 
    + \, \widehat d^{(t)}_{-j}\left(i^{(t)}\right)
    - \, \widehat d^{(t)}_{-i}\left(j^{(t)}\right)
    - \, \widehat d^{(h)}\left(i^{(h)}, j^{(h)}\right)\right)\right. \\
  & \qquad \left. -\left( d^{(t)}\left(i^{(t)}, j^{(t)}\right) 
    + \, d^{(t)}_{-j}\left(i^{(t)}\right)
    - \, d^{(t)}_{-i}\left(j^{(t)}\right)
    - \, d^{(h)}\left(i^{(h)}, j^{(h)}\right)\right)\right|\\[5pt]
  &\leq \frac{1}{2} \left\{
  \left|\widehat d\left(i^{(t)}, x^{(t)}\right) - d\left(i^{(t)}, x^{(t)}\right)\right| 
  + \left|\widehat d^{(t)}_{-j}\left(i^{(t)}\right) - d^{(t)}_{-j}\left(i^{(t)}\right)\right| \right.\\
  & \qquad \left.
  + \left|\widehat d^{(t)}_{-i}\left(j^{(t)}\right) - d^{(t)}_{-i}\left(j^{(t)}\right)\right| + \left|\widehat d^{(h)}\left(i^{(h)}, j^{(h)}\right) - d^{(h)}\left(i^{(h)}, j^{(h)}\right)\right|\right\}\\
  & \leq 2\delta \text{ (by lemma \ref{lemma:error_not_accumulate})},
\end{align*}
where $d^{(t)}_{-j}\left(i^{(t)}\right)$ is the corresponding mean for the true distances, that is, \[
    d^{(t)}_{-j}\left(i^{(t)}\right) = \frac{1}{k-t-1}\sum_{\scaleto{u \, \in \, 
    \left.V^{(t)}_o \middle \backslash \left\{i^{(t)}, j^{(t)}\right\}\right.}{18pt}}  d^{(t)}\left(i^{(t)}, u\right).\]
By similar steps, we can show \(\left|\widehat d\left(i^{(t)}, x^{(t)}\right) - \widetilde d\left(i^{(t)}, x^{(t)}\right)\right|\leq 3/2 \delta\) for Equation \eqref{eq:distance_update_final}. Therefore, $\left|\widehat d(e) - \widetilde d(e)\right| \leq 2\delta$ holds for any edge $e\in \widehat E_{nj}$.

For any edge $e\in E$, $\widetilde d(e) = d(e)>4\delta$, by triangle inequality, 
\[\widehat d(e) \geq \widetilde d(e) - \left|\widehat d(e) - \widetilde d(e)\right|> 4\delta - 2\delta = 2\delta.\]
Similarly, for any $e_0 \notin E$ but $e_0\in \widehat \T_{nj}$, then $\widetilde d(e_0) = 0$,
\[\widehat d(e_0)\leq \widetilde d(e_0) + \left|\widehat d(e) - \widetilde d(e)\right|\leq 0 + 2\delta = 2\delta.\]

\subsubsection{Proof for Corollary \ref{coro:cutoff}}\label{appendix:cutoff_coro}
This is a direct result of Theorem \ref{thm:cutoff} and Proposition \ref{prop:tree_shape}.

\subsection{Proof of Theorem \ref{thm:meta_theorem}}\label{appendix:meta_theorem}

The proof mainly relies on the consistency result in Theorem \ref{thm:tsgdist} and Theorem \ref{thm:cutoff}. The major component of the proof is showing that conditions in these two theorems are satisfied in the context of Theorem \ref{thm:meta_theorem}.

The convergence of $\widehat D$ is a direct result of Lemma \ref{lemma:Zconverge_main_thm}, \ref{lemma:Aconverge_main_thm}, and Theorem \ref{thm:tsgdist} (with $e_n = \Delta_n^{-.24}\log^{2.75}n$).

\begin{lemma}\label{lemma:Zconverge_main_thm}
Consider $\widehat Z$ being the output of Algorithm \ref{vsp}, then there exists a series of reordering matrix $P_n\in \mathscr{P}(k)$ and a fixed diagonal matrix $C$ such that
    \[\left\|\widehat Z - \widetilde ZCP_n\right\|_{2\rightarrow \infty} = O_p(\Delta_n^{-.24}\log^{2.75}n).\]
\end{lemma}

\begin{lemma}\label{lemma:Aconverge_main_thm}
$\left\|A - \mathscr{A}\right\|_2 = O_p(n\rho_n \Delta_n^{-.24}\log^{2.75}n)$.
\end{lemma}

To demonstrate the consistency of $\widehat \T_z$, we apply Theorem \ref{thm:cutoff} to tree $\T_z$ with the observed set being $V_z$. We would like to verify that the two conditions in Theorem \ref{thm:cutoff} are satisfied.

The first condition requires $V_z$ to contain all nodes with degree $<3$ in $\T_z$. If $u$ is a node in $\T_z$ such that $u\notin V_z$ and $deg(u)<3$, then $u$ is also an internal node of $\T$ such that $deg(u)<3$, as the creation of $\T_z$ from $\T$ only affects nodes in $V_\ell \cup V_z$. This contradicts our assumption that all internal nodes in $\T$ are of degree $>3$. The second condition requires $\min_e d(e) > 4\delta$. By the convergence of $\widehat D$, this condition holds with a probability that goes to one.

Now we want to show that the cutoff value is chosen properly. For any cutoff value $\varphi$ that grows faster than the convergence rate of $\widehat D$, 
\begin{equation}\label{eq:nj_1}
    \pr\left(\varphi > 2\delta\right)\xrightarrow{n\rightarrow \infty} 1.
\end{equation}
Since $d(e)>0$ for any edges $e\in E$, by the convergence of $\widehat D$ and the fact that $\varphi \rightarrow 0$, 
\begin{equation}\label{eq:nj_2}
    \pr\left(\min_{e\in E}\widehat d(e) > \varphi \right)\xrightarrow{n\rightarrow \infty} 1
\end{equation}
Notice that $2\delta \geq \max_{e_0\in \widehat E\setminus E} \widehat d(e_0)$. Together Equation \eqref{eq:nj_1} and \eqref{eq:nj_2} implies that
\[    \pr\left(\min_{e\in E}\widehat d(e) > \varphi > \max_{e_0\in \widehat E\setminus E} \widehat d(e_0) \right)\xrightarrow{n\rightarrow \infty} 1.\]
This means by applying the shrinking step in the \texttt{SparseNJ} algorithm, all wrongly recovered edges are deleted with a probability that goes to one. Then by Proposition \ref{prop:tree_shape}, we have
\[\pr\left(\widehat \T_z\sim \T_z\right)\xrightarrow{n\rightarrow \infty} 1.\]
The consistency of $\widehat \T_z$, together with Lemma \ref{lemma:Zconverge_main_thm}, implies the consistency of $\widehat \T$.

\subsubsection{Proof of Lemma \ref{lemma:Zconverge_main_thm} (Consistency of $\widehat Z$ after the Sign Flipping Step)}\label{appendix:sign_adjust}
In this proof, we use $\widehat Z$ to denote $UR_U$ (without sign flip) and $\widehat Z_D$ to denote $\widehat Z\widehat D$ (with sign flip) in Algorithm \ref{vsp}. We would like to show 
    \[\left\|\widehat Z_D - \widetilde ZP_n\right\|_{2\rightarrow \infty} = O_p(\Delta_n^{-.24}\log^{2.75}n).\]
By Corollary C.1 in \citep{varimax_rohe},
\begin{equation}\label{vsp_result}
    \left\|\widehat Z - \widetilde ZP_n^{sign}\right\|_{2\rightarrow \infty} = O_p(\Delta_n^{-.24}\log^{2.75}n)
\end{equation}
where $P^{sign}_n$ is a series of permutation matrix that allows sign flip, i.e., $P^{sign}_n\in \mathscr{P}^{sign}(k) = \{P\in \R^{k\times k}: P^TP = PP^T = I_{k\times k}, \ P_{ij}\in \{0, -1, 1\}\}$. Therefore we only need to deal with the fluctuation caused by the sign adjustment in Algorithm \ref{vsp}.

Decompose $P_n^{sign} = P_n D_n$, where $D_n$ is a series of diagonal matrix with elements in $\{-1, 1\}$ and $P_n\in \mathcal{P}(k)$. Then $D_n$ performs the sign flip and $P_n$ reorders the columns, by Equation (\ref{vsp_result}), 
\[\left\|\widehat ZD_n - \widetilde Z P_n\right\|_{2\rightarrow \infty}= O_p(\Delta_n^{-.24}\log^{2.75}n).\]

Let $\widehat\mu_j = \sum_j \widehat Z_{ij}/n$, $\mu_j = \sum_j \left[\widetilde ZP_n\right]_{ij}/n$, $\widehat d_j = \text{sgn}(\widehat \mu_j)$, $d_j = \left[D_n\right]_{jj}$. We drop the subscript $n$ for simplicity here.

Since $\|\cdot\|_{2\rightarrow \infty}$ corresponds to the maximum row-wise Euclidean norm,  
\begin{equation}\label{mu_converge}
 \left|\widehat \mu_jd_j - \mu_j\right| = O_p(\Delta_n^{-.24}\log^{2.75}n).   
\end{equation}
Notice that $\mu_j\geq 0$ by definition, and \(\widehat \mu_j\widehat d_j\geq 0\), 
\begin{equation}\label{mu_converge1}
\left|\widehat \mu_j\widehat d_j - \mu_j\right|\leq \left|\widehat \mu_jd_j - \mu_j\right| = O_p(\Delta_n^{-.24}\log^{2.75}n).
\end{equation}
By Equation (\ref{mu_converge}) and (\ref{mu_converge1}), 

\begin{equation}\label{d_converge_0}
    \left|\widehat \mu_j\widehat d_j - \widehat \mu_jd_j\right| = O_p(\Delta_n^{-.24}\log^{2.75}n).
\end{equation}
Since $\widetilde Z_{ij}$ are bounded, $|\mu_jd_j| = O_p(1)$. By Equation (\ref{mu_converge}), \(\left|\widehat \mu_j - \mu_jd_j\right| = O_p(\Delta_n^{-.24}\log^{2.75}n)\). These imply $|\widehat \mu_j| = O_p(1)$. Then by Equation (\ref{d_converge_0}),
\begin{equation}
    \left|\widehat d_j - d_j\right| = O_p(\Delta_n^{-.24}\log^{2.75}n).
\end{equation}
Define $\widehat D_n$ to be a series of diagonal matrix with $[\widehat D_n]_{jj} = \widehat d_j$, this implies
\begin{equation}
    \left\|\widehat D_n - D\right\|_2 = O_p(\Delta_n^{-.24}\log^{2.75}n).
\end{equation}

Now we can show
\begin{align}
   \left\|\widehat Z\widehat D_n - \widetilde Z P_n\right\|_{2\rightarrow \infty} & \leq 
   \left\|\widehat Z\widehat D_n - \widehat Z D_n\right\|_{2\rightarrow \infty}+ \left\|\widehat ZD_n - \widetilde Z P_n\right\|_{2\rightarrow \infty}\\
   & \leq \left\|\widehat Z\right\|_{2\rightarrow \infty}\left\|\widehat D_n - D\right\|_2 + O_p(\Delta^{-.24}\log^{2.75}n)\\
   & = O_p(1)O_p(\Delta_n^{-.24}\log^{2.75}n)+O_p(\Delta_n^{-.24}\log^{2.75}n)\\
   & = O_p(\Delta_n^{-.24}\log^{2.75}n).
\end{align}
$\left\|\widehat Z\right\|_{2\rightarrow \infty} = O_p(1)$ because $\left\| \widetilde ZP_n^{sign}\right\|_{2\rightarrow \infty} = O_p(1)$ and Equation (\ref{vsp_result}), the second step follows from a basic property of $\|\cdot\|_{2\rightarrow \infty}$. (see Proposition 6.5. in \citep{cape2019two})

This shows that the sign adjustment step does not affect the asymptotic behavior of $\widehat Z$. The result for $\widehat Z_+$ follows from Proposition \ref{prop:projection}.

\subsubsection{Proof of Lemma \ref{lemma:Aconverge_main_thm}}
By Lemma G.4 in \cite{varimax_rohe}, \[\left\|A - \mathscr{A}\right\|_2 = O_p((n\rho_n\log^3 n)^{.5}) = O_p(\Delta_n^{.5}\log^{1.5} n)).\]
When $e_n = \Delta^{-.24}\log^{2.75}n$, we have $n\rho_ne_n = \Delta_n^{.76}\log^{2.75}n > \Delta_n^{.5}\log^{1.5} n$, thus 
\[\left\|A - \mathscr{A}\right\|_2 = O_p((n\rho_n\log^3 n)^{.5}) = O_p(n\rho_n \Delta_n^{-.24}\log^{2.75}n)).\]

\end{document}